\documentclass[pra,twocolumn,superscriptaddress]{revtex4-1}%

\usepackage{amsfonts}
\usepackage{amsmath}
\usepackage{amssymb}
\usepackage{graphicx}
\usepackage[ colorlinks = true, linkcolor = blue, urlcolor  = blue, citecolor
= red,              anchorcolor = green, ]{hyperref}%
\providecommand{\U}[1]{\protect\rule{.1in}{.1in}}
\newtheorem{theorem}{Theorem}

\newtheorem{corollary}{Corollary}

\newtheorem{definition}{Definition}

\newtheorem{lemma}{Lemma}

\newtheorem{proposition}{Proposition}
\newtheorem{remark}{Remark}

\newenvironment{proof}[1][Proof]{\noindent\textbf{#1.} }{\ \rule{0.5em}{0.5em}}
\allowdisplaybreaks

\begin{document}
\preprint{ }
\title[Resource theory of asymmetric distinguishability for quantum channels]{Resource theory of asymmetric distinguishability for quantum channels}
\author{Xin Wang}
\email{wangxinfelix@gmail.com}
\affiliation{Joint Center for Quantum Information and Computer Science, University of
Maryland, College Park, Maryland 20742, USA}
\affiliation{Institute for Quantum Computing, Baidu Research, Beijing 100193, China}
\author{Mark M. Wilde}
\email{mwilde@lsu.edu}
\affiliation{Hearne Institute for Theoretical Physics, Department of Physics and Astronomy,
Center for Computation and Technology, Louisiana State University, Baton
Rouge, Louisiana 70803, USA}

\begin{abstract}
This paper develops the resource theory of asymmetric distinguishability for
quantum channels, generalizing the related resource theory for states
[Matsumoto, arXiv:1010.1030; Wang and Wilde, Phys.~Rev.~Research~1, 033170 (2019)]. The key constituents of the channel
resource theory are quantum channel boxes, consisting of a pair of quantum
channels, which can be manipulated for free by means of an arbitrary quantum
superchannel (the most general physical transformation of a quantum channel).
One main question of the resource theory is the approximate channel box
transformation problem, in which the goal is to transform an initial channel
box (or boxes) to a final channel box (or boxes), while allowing for an
asymmetric error in the transformation. The channel resource theory is richer
than its counterpart for states because there is a wider variety of ways in
which this question can be framed, either in the one-shot or $n$-shot regimes,
with the latter having parallel and sequential variants. As in our prior work
[Wang and Wilde, Phys.~Rev.~Research~1, 033170 (2019)], we consider two special cases of the general channel box
transformation problem, known as distinguishability distillation and
dilution. For the one-shot case, we find that the optimal
values of the various tasks are equal to the non-smooth or smooth channel min-
or max-relative entropies, thus endowing all of these quantities with
operational interpretations. In the asymptotic sequential setting, we prove
that the exact distinguishability cost is equal to the channel max-relative
entropy and the distillable distinguishability is equal to the amortized
channel relative entropy of [arXiv:1808.01498]. This latter result can also be
understood as a solution to Stein's lemma for quantum channels in the
sequential setting. Finally, the theory simplifies significantly for
environment-seizable and classical--quantum channel boxes.

\end{abstract}
\date{\today}
\startpage{1}
\endpage{10}
\maketitle

%\tableofcontents

\section{Introduction}

In many scientific fields of interest, distinguishability is an important
concept. More generally, it can be considered as a resource in that it allows
for making decisions, and furthermore, the more distinguishable that two
possibilities are, the easier and faster it is to make a decision.

In a recent paper, we formalized the notion of distinguishability as a
resource by developing the resource theory of asymmetric distinguishability in
detail \cite{WW19states}, following the original proposal from
\cite{Mats10,M11}. This resource theory demonstrates that distinguishability
is truly a fundamental resource that can be manipulated and interconverted
into different forms. The benefit of developing this resource theory is that,
not only can fundamental tasks such as quantum hypothesis testing \cite{HP91,ON00,Hay03,Hay04,WR12,Hay17} be recast
into an intuitive approach based on resource-theoretic thinking, but also new
information processing tasks emerge, such as distinguishability dilution,
which is related to concepts such as simulation and synthesis of quantum
states. The present paper illustrates further benefits of the
resource-theoretic approach by using it to solve some outstanding questions in
the theory of quantum channel discrimination.

In the resource theory of asymmetric distinguishability for states \cite{WW19states}, the basic object to
be manipulated is a quantum \textquotedblleft box\textquotedblright%
\ $(\rho,\sigma)$\ consisting of two quantum states $\rho$ and $\sigma$. The descriptor ``asymmetric'' applies to this resource theory because it allows for a slight error in the transformation of the first state of the box, while not allowing for any error in the transformation of the second state of the box.   One
basic task is to distill as many bits of asymmetric distinguishability as
possible from this box by processing it with an arbitrary quantum channel \cite{WW19states}.
Another basic task is to dilute bits of asymmetric distinguishability to
prepare the box $(\rho,\sigma)$, with the goal being to use as few bits of
asymmetric distinguishability as possible in order to do so \cite{WW19states}. These tasks give
operational meaning to fundamental entropic measures such as the min-relative
entropy \cite{D09}, the smooth min-relative entropy \cite{BD10,BD11,WR12}, the
max-relative entropy \cite{D09}, and the smooth max-relative entropy
\cite{D09}. One of the core results for this resource theory is that it is
reversible, and the fundamental rate of interconversion is characterized by
the quantum relative entropy \cite{WW19states}.

The main goal of the present paper is to generalize these concepts from
quantum states to quantum channels, given the prominent role of the latter in
quantum information and beyond. We note here that recently there has been much effort
more generally in extending concepts from resource theories of quantum states
to resource theories for quantum channels (see, e.g.,
\cite{BHLS03,BenDana2017,DBW17,GFWRSCW18,LBL18,BDW18,TR19,TEZP19,WW19PPT,Seddon2019,WWS19channels,LY19,LW19,YLZRTG19}%
). In the resource theory of asymmetric distinguishability for quantum
channels presented here, the basic object to be manipulated is a quantum
channel box $(\mathcal{N},\mathcal{M})$, which consists of two quantum
channels $\mathcal{N}$ and $\mathcal{M}$. The idea is that the input and
output ports of the channel box are accessible to an agent in the resource
theory, while the particular choice of the channel is unknown to the agent. A
key difference between this resource theory and the former one for quantum
states is that a quantum channel can be probed by means of both an input port
and an output port, which implies that the way they are manipulated is by
means of a quantum superchannel \cite{CDP08}. As a simple example of a superchannel, consider that the encoding and decoding, i.e., pre- and post-processing, of a channel commonly employed in quantum Shannon theory \cite{W17book} realize a physical transformation of a channel.  The incorporation of superchannels into the resource theory implies that the channel resource
theory is more involved\ than it is for boxes consisting only of states.

More
generally, we allow for quantum strategy boxes
\cite{GW07,CDP08a,CDP09,G09,G12,Gutoski2018fidelityofquantum} and manipulate
them by means of general physical transformations \cite{CDP09} that take
quantum strategies to other quantum strategies (note that quantum strategies
are in one-to-one correspondence with quantum combs \cite{CDP09}). By the results of \cite{CDP09}, such physical transformations are in fact quantum strategies themselves, so that our generalization of the resource theory to quantum strategies is a significant generalization.

We consider several fundamental tasks in this resource theory, which can be
understood as extensions of the tasks considered in \cite{WW19states}. The
first basic one is distinguishability distillation, in which the goal is to
distill as many bits of asymmetric distinguishability as possible from a
single channel box in the one-shot setting, or multiple channel boxes in the
$n$-shot setting. This task is intimately related to asymmetric hypothesis
testing for quantum channels \cite{Cooney2016} (see \cite{Hayashi09} for the
classical case), which is a particular kind of quantum channel discrimination.
For this task, there are a variety of possibilities to consider, including the
one-shot case and the $n$-shot case, in the latter using either a parallel or
sequential strategy \cite{GW07,CDP08a,G09,Duan09,Harrow10,G12,Cooney2016}. We
also consider this task for quantum strategy boxes. Another basic task of
interest is distinguishability dilution, in which the goal is to dilute bits
of asymmetric distinguishability to a single or multiple channel boxes, using
as few bits of asymmetric distinguishability as possible. This task also has a
variety of possibilities, including one- and $n$-shot, the latter having
parallel and sequential variants as well. We likewise consider this task for
quantum strategy boxes. This task is also intimately related to quantum
channel simulation
\cite{ieee2002bennett,BDHSW09,BCR09,BBCW13,BRW14,B13,BenDana2017,GFWRSCW18,FBB18,FWTB18,Wilde2018}%
, but here takes on a specific form due to the structure of the resource
theory of asymmetric distinguishability.

One of the major tasks in this resource theory is to convert one channel box
to another, doing so either exactly or approximately. As a variant of this
problem, another task is to determine the rate at which it is possible to
convert $n$ channel boxes, with each box consisting of the same pair of
channels, to $m$ boxes consisting of another pair of channels, when $n$ is
allowed to be arbitrarily large. More generally, we consider the
conversion of an $n$-round quantum strategy box to an $m$-round strategy box.
The simpler transformation problem for state boxes was solved in \cite{WW19states} and is
relevant for addressing the channel box transformation problem for particular
channel boxes that are environment-seizable, as defined in~\cite{BHKW18}.

\section{Summary of results}

We now summarize the main contributions and results of our paper:

\begin{enumerate}
\item We establish the resource theory of asymmetric distinguishability for
quantum channels, with the basic objects being quantum channel boxes, the free
operations to manipulate them being quantum superchannels \cite{CDP08}, and
the basic units of currency being bits of asymmetric distinguishability (see
Section~\ref{sec:channel-res-theory}). Later we accomplish the same for
quantum strategy boxes, with the free operations to manipulate them being
quantum strategies (see Section~\ref{sec:seq-case-gen-trans}).

\item We prove that the approximate channel box transformation problem is
characterized by a semi-definite program and thus can be calculated
efficiently with respect to the input and output dimensions of the channels   (see
Section~\ref{sec:ch-box-tr-problem}).

\item The exact one-shot distillable distinguishability of a quantum channel
box is equal to the channel min-relative entropy, which is a particular case
of the generalized channel divergence of \cite{Cooney2016,LKDW18}. The exact
one-shot distinguishability cost of a quantum channel box is equal to the
channel max-relative entropy, which is a particular case of the generalized
channel divergence of \cite{Cooney2016,LKDW18} and explored in more detail in
\cite{GFWRSCW18,BHKW18}. See
Section~\ref{sec:exact-1shot-dilution-distillation} for both of these results.

\item The approximate one-shot distillable distinguishability of a quantum
channel box is equal to the smooth channel min-relative entropy of
\cite{Cooney2016}, the latter also known as channel hypothesis testing
relative entropy \cite{Cooney2016}. The approximate one-shot
distinguishability cost of a quantum channel box is equal to the 
smooth channel max-relative entropy, again a particular case of the generalized channel
divergence of \cite{LKDW18}\ and explored in more detail in \cite{GFWRSCW18}.
See Section~\ref{sec:approx-1shot-dilution-distillation} for both of these results.

\item We consider asymptotic parallel versions of the above tasks in
Section~\ref{sec:parallel-asymptotic-tasks}. We find that the exact
distillable distinguishability is given by the regularized channel
min-relative entropy (see Section~\ref{sec:exact-parallel-distill}). By means
of an example from \cite{Acin01}, we conclude that the regularization seems to be
necessary because the channel min-relative entropy is highly non-additive. We
then prove that the exact distinguishability cost is equal to the channel
max-relative entropy (see Section~\ref{sec:exact-parallel-cost}). The
distillable distinguishability is equal to the regularized channel relative
entropy (see Section~\ref{sec:approx-parallel-distill}), and the same
quantity is a lower bound on the distinguishability cost (see
Section~\ref{sec:approx-parallel-cost}). These latter operational tasks
simplify for both environment-seizable and classical--quantum channel boxes.

\item Section~\ref{sec:gen-box-trans-parallel}\ considers the asymptotic
parallel version of the general channel box transformation problem, giving
basic definitions and some bounds that apply to this case. Again, the results
simplify for the case of environment-seizable and classical--quantum channel boxes.

\item Section~\ref{sec:seq-case-gen-trans} considers the quantum strategy box
transformation problem. To begin with, this section introduces the generalized
quantum strategy divergence as a generalization of the strategy distance of
\cite{CDP08a,CDP09,G12} and establishes a data processing inequality for this
distinguishability measure. The section then establishes several bounds on how
well one can perform a physical transformation from one strategy box to
another strategy box. All of the results apply to sequential channel boxes
because these are special cases of strategy boxes. Furthermore, we consider an
asymptotic version of the box transformation problem for sequential channel
boxes and prove concrete results for environment-seizable and
classical--quantum channel boxes.

\item We then consider distillation and dilution of strategy boxes in
Section~\ref{sec:dist-dil-strat-seq-boxes}. Our key results here, specialized
to sequential channel boxes, include single-letter formulas for the asymptotic exact
sequential distinguishability cost and the asymptotic sequential distillable
distinguishability, expressed respectively as the channel max-relative entropy
and the amortized channel relative entropy of \cite{BHKW18}, giving these
quantities fundamental operational interpretations in the resource theory of
asymmetric distinguishability. The latter result can be alternatively
understood as a solution to Stein's lemma for quantum channels in the
sequential setting.
\end{enumerate}

In the rest of the paper, we discuss details of the resource theory of
asymmetric distinguishability for quantum channels, as well as the
contributions listed above.

\section{Resource theory of asymmetric distinguishability for quantum
channels}

\label{sec:channel-res-theory}We begin by generalizing the resource theory of
asymmetric distinguishability from \cite{WW19states}\ to the setting of
quantum channels, by considering a channel box of the following form:%
\begin{equation}
(\mathcal{N},\mathcal{M}), \label{eq:channel-box}%
\end{equation}
where $\mathcal{N}$ and $\mathcal{M}$ are quantum channels, each acting on an
input system $A$ and outputting a system $B$. Recall that a quantum channel is a completely positive, trace-preserving (CPTP) map. We also write these as
$\mathcal{N}_{A\rightarrow B}$ and $\mathcal{M}_{A\rightarrow B}$ in what
follows in order to indicate the input and output systems explicitly. The
channel box generalizes the state box $(\rho,\sigma)$ from \cite{WW19states},
which consists of a pair of quantum states $\rho$ and $\sigma$. In fact, a
state box is a special case of a channel box in which the input systems are trivial.

One interpretation of the channel box in \eqref{eq:channel-box} is that a
distinguisher is allowed to prepare any state $\rho_{RA}$ of a reference
system $R$ and the channel input $A$, either the channel $\mathcal{N}%
_{A\rightarrow B}$ or $\mathcal{M}_{A\rightarrow B}$ is applied, and then the
distinguisher is allowed to perform any post-processing on the reference $R$
and the channel output $B$ in order to decide which channel was applied. That
is, by inputting an arbitrary state $\rho_{RA}$ to the channel box
(pre-processing) and then applying the channel $\mathcal{P}_{RB\rightarrow S}$
(post-processing), one can transform it to the following state box:%
\begin{equation}
(\mathcal{P}_{RB\rightarrow S}(\mathcal{N}_{A\rightarrow B}(\rho
_{RA})),\mathcal{P}_{RB\rightarrow S}(\mathcal{M}_{A\rightarrow B}(\rho
_{RA}))).
\end{equation}
More generally, the agent who has access to the channel box in
\eqref{eq:channel-box} can perform a quantum superchannel \cite{CDP08}\ on it
in order to transform it to another channel box, as discussed in
Section~\ref{sec:superchannel}\ below.

As stated earlier, the channel box in \eqref{eq:channel-box} indeed
generalizes the state box $(\rho,\sigma)$ considered previously in
\cite{WW19states}. Another way of seeing this is to take the channels
$\mathcal{N}$ and $\mathcal{M}$ in \eqref{eq:channel-box} to be replacer
channels with the following action:%
\begin{align}
\mathcal{N}_{A\rightarrow B}(\omega_{A})  &  =\operatorname{Tr}_{A}[\omega
_{A}]\rho_{B},\label{eq:replacer-ch-1}\\
\mathcal{M}_{A\rightarrow B}(\omega_{A})  &  =\operatorname{Tr}_{A}[\omega
_{A}]\sigma_{B}. \label{eq:replacer-ch-2}%
\end{align}
Then no matter what state $\tau_{RA}$ is input to the channel box
$(\mathcal{N},\mathcal{M})$, it reduces to the state box $(\tau_{R}\otimes
\rho_{A},\tau_{R}\otimes\sigma_{B})$, which, by the discussion in
\cite[Section~III]{WW19states},\ is equivalent by a free operation  to the state box $(\rho_{B},\sigma_{B})$.

\subsection{Environment-parametrized and -seizable channels}

\label{sec:env-param-seize}

Other simple classes of channel boxes that are strongly related to state
boxes, generalizing the above example of a replacer channel box in
\eqref{eq:replacer-ch-1}--\eqref{eq:replacer-ch-2}, include those that are
environment parametrized \cite{TW2016}\ and the subclass of
environment-seizable channel boxes \cite{BHKW18}. Note that 
environment-parametrized channel boxes are related to
programmable channels \cite{NC97,DP05}.

A channel box $(\mathcal{N}_{A\rightarrow B},\mathcal{M}_{A\rightarrow B}%
)$\ is \textit{environment parametrized} with associated environment states $\rho_E$ and $\sigma_E$ if there exists a common interaction
channel $\mathcal{P}_{AE\rightarrow B}$  such that
\begin{align}
\mathcal{N}_{A\rightarrow B}(\omega_{A})  &  =\mathcal{P}_{AE\rightarrow
B}(\omega_{A}\otimes\rho_{E}), \label{eq:envir-param-1}\\
\mathcal{M}_{A\rightarrow B}(\omega_{A})  &  =\mathcal{P}_{AE\rightarrow
B}(\omega_{A}\otimes\sigma_{E}) \label{eq:envir-param-2}
\end{align}
for all inputs $\omega_{A}$ \cite{TW2016}. In this way, any pre-processing of an
environment-parametrized channel box as%
\begin{equation}
(\mathcal{N}_{A\rightarrow B},\mathcal{M}_{A\rightarrow B})\rightarrow
(\mathcal{N}_{A\rightarrow B}(\omega_{A}),\mathcal{M}_{A\rightarrow B}%
(\omega_{A}))\
\end{equation}
can be viewed as a postprocessing of the state box $(\rho_{E},\sigma_{E})$,
via%
\begin{align}
(\rho_{E},\sigma_{E})  &  \rightarrow(\omega_{A}\otimes\rho_{E},\omega
_{A}\otimes\sigma_{E})\nonumber\\
&  \rightarrow(\mathcal{P}_{AE\rightarrow B}(\omega_{A}\otimes\rho
_{E}),\mathcal{P}_{AE\rightarrow B}(\omega_{A}\otimes\sigma_{E}))\nonumber\\
&  =(\mathcal{N}_{A\rightarrow B}(\omega_{A}),\mathcal{M}_{A\rightarrow
B}(\omega_{A})), \label{eq:post-process-env-param}%
\end{align}
so that the distinguishability of the channel box $(\mathcal{N},\mathcal{M})$
is always limited by that of the state box $(\rho_{E},\sigma_{E})$, as
observed in \cite{TW2016} (see \cite{JWDFY08,DM14}\ for related observations
in quantum estimation theory).

We should emphasize that an arbitrary channel box $(\mathcal{N},\mathcal{M})$ is environment-parametrized with associated environment states that are orthogonal \cite{DW17}. That is, we can set $\rho_E = |0\rangle \langle 0 |_E$ and $\sigma_E = |1\rangle \langle 1|_E$ and the common interaction channel $\mathcal{P}_{AE \to B}$ as
\begin{multline}
\mathcal{P}_{AE \to B} (\zeta_{AE}) = \mathcal{N}_{A\to B}(\langle 0 |_E \zeta_{AE} | 0 \rangle_E) \\+ \mathcal{M}_{A\to B}(\langle 1 |_E \zeta_{AE} | 1 \rangle_E).
\label{eq:trivial-env-param-int-ch}
\end{multline}
In this way, the channels $\mathcal{N}_{A\to B}$ and $\mathcal{M}_{A \to B}$ are realized as in \eqref{eq:envir-param-1}--\eqref{eq:envir-param-2}, by starting from the state box $(|0\rangle \langle 0 |_E, |1\rangle \langle 1 |_E)$ and applying the common interaction channel $\mathcal{P}_{AE \to B}$ in \eqref{eq:trivial-env-param-int-ch}. However, this realization of the channels is the least efficient from the perspective of the resource theory of asymmetric distinguishability, because a state box consisting of a pair of orthogonal states is equivalent to an infinite number of bits of asymmetric distinguishability \cite{WW19states}. (See \cite{WW19states} for the notion of  bits of asymmetric distinguishability, and Section~\ref{sec:one-shot-dist-dil} for this notion in the channel resource theory.) In this sense, the realization of an arbitrary channel box in the above way is trivial because it requires an infinite number of bits of asymmetric distinguishability in order to do so. The concept of environment-parametrized channel boxes becomes non-trivial when the background environment states have finite distinguishability, when measured according to some divergence, so that the channel box can be realized starting from a finite number of bits of asymmetric distinguishability.

\textit{Environment-seizable channel boxes} are defined to be
environment-parametrized with associated environment states $\rho_E$ and $\sigma_E$ and additionally have the property that it is
possible to find a common pre- and post-processing of the channel box
$(\mathcal{N}_{A\rightarrow B},\mathcal{M}_{A\rightarrow B})$ to retrieve the
state box $(\rho_{E},\sigma_{E})$ from it \cite{BHKW18}. That is, for
environment-seizable channels, there exists a common input state $\tau_{RA}$
and a common post-processing channel $\mathcal{D}_{RB\rightarrow E}$ such that%
\begin{align}
\mathcal{D}_{RB\rightarrow E}(\mathcal{N}_{A\rightarrow B}(\tau_{RA}))  &
=\rho_{E},\\
\mathcal{D}_{RB\rightarrow E}(\mathcal{M}_{A\rightarrow B}(\tau_{RA}))  &
=\sigma_{E}.
\end{align}

In this way, we have the following equivalence for environment-seizable
channels:%
\begin{equation}
(\mathcal{N}_{A\rightarrow B},\mathcal{M}_{A\rightarrow B})\leftrightarrow
(\rho_{E},\sigma_{E}), \label{eq:env-seize-ch-boxes}%
\end{equation}
with the direction $\leftarrow$\ of the equivalence following from
\eqref{eq:post-process-env-param} and the other direction $\rightarrow
$\ following from the seizable property. Thus, environment-seizable channel
boxes represent a broader generalization of state boxes than do channel boxes
consisting of replacer channels. Furthermore, environment-seizable channel boxes are fully identified with the background environment states $\rho_E$ and $\sigma_E$ in the above sense. As we show later, and as observed in earlier
work \cite{TW2016,BHKW18}, the equivalence in
\eqref{eq:env-seize-ch-boxes}\ simplifies the resource theory of asymmetric
distinguishability significantly for environment-seizable channel boxes. Finally, several examples of environment-seizable channel boxes were presented in \cite{BHKW18}, and the notion of environment-seizable channel boxes is related to the notion of resource-seizable channels from~\cite{Wilde2018}.

\subsection{Superchannels as transformations of channel boxes}

\label{sec:superchannel}The most general physical transformation allowed on a channel
box is a superchannel $\Theta$, which is a quantum physical transformation of
channels \cite{CDP08}. That is, a superchannel is a linear map that preserves
the set of quantum channels, even when the quantum channel is an arbitrary
bipartite channel with external input and output systems that are arbitrarily
large. In this sense, superchannels are completely CPTP\ preserving. Note that the terminology ``superchannel'' was introduced in~\cite{G18}.

\begin{figure}[ptb]
\begin{center}
\includegraphics[
width=\linewidth
]{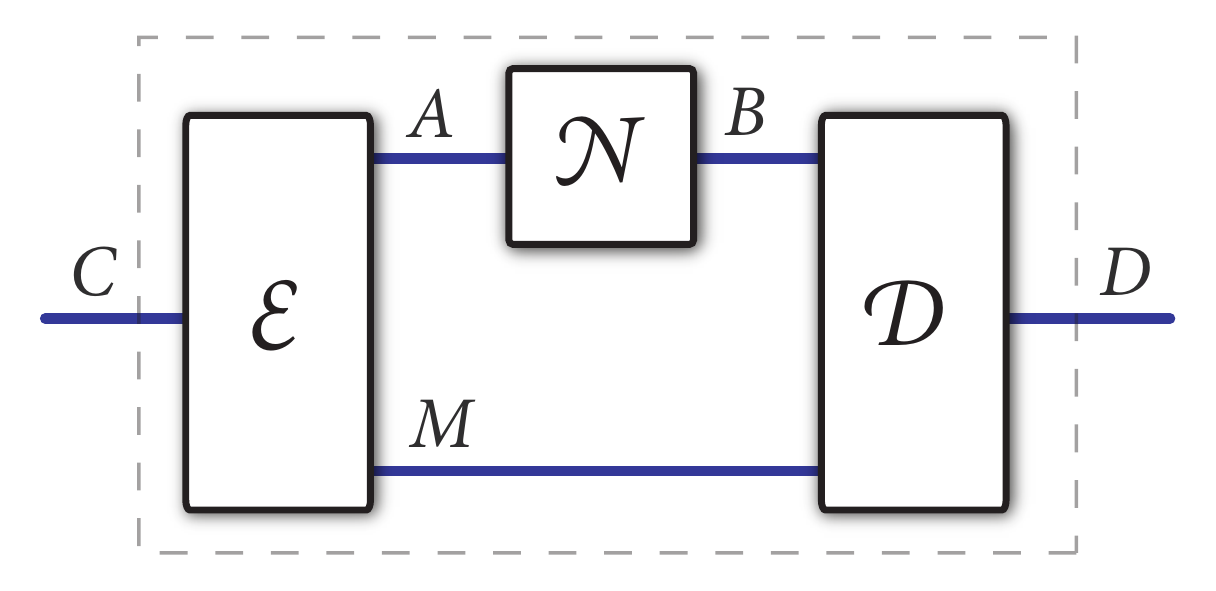}
\end{center}
\caption{The figure depicts the transformation of a channel $\mathcal{N}_{A\to B}$ by a superchannel $\Theta_{(A\to B) \to (C\to D)}$, the latter of which consists of the pre-processing channel $\mathcal{E}_{C\to AM}$ and the post-processing channel $\mathcal{D}_{BM\to D}$.}%
\label{fig:superchannel}%
\end{figure}

To see this, a superchannel $\Theta_{\left(  A\rightarrow B\right)
\rightarrow\left(  C\rightarrow D\right)  }$ takes as input a quantum channel
$\mathcal{N}_{A\rightarrow B}$ and outputs a quantum channel $\mathcal{K}%
_{C\rightarrow D}$, which we denote by%
\begin{equation}
\Theta_{\left(  A\rightarrow B\right)  \rightarrow\left(  C\rightarrow
D\right)  }(\mathcal{N}_{A\rightarrow B})=\mathcal{K}_{C\rightarrow D}.
\end{equation}
The superchannel $\Theta_{\left(  A\rightarrow B\right)  \rightarrow\left(
C\rightarrow D\right)  }$ is completely CPTP\ preserving in the sense that the
following output channel%
\begin{equation}
\left(  \operatorname{id}_{\left(  R\right)  \rightarrow\left(  R\right)
}\otimes\Theta_{\left(  A\rightarrow B\right)  \rightarrow\left(  C\rightarrow
D\right)  }\right)  (\mathcal{M}_{RA\rightarrow RB}),
\end{equation}
is a CPTP map for all input quantum channels $\mathcal{M}_{RA\rightarrow RB}$,
where $\operatorname{id}_{\left(  R\right)  \rightarrow\left(  R\right)  }$
denotes the identity superchannel \cite{CDP08}.

One of the fundamental theorems of superchannels is that each superchannel
$\Theta_{\left(  A\rightarrow B\right)  \rightarrow\left(  C\rightarrow
D\right)  }$ has a physical realization as a pre- and post-processing of the
channel $\mathcal{N}_{A\rightarrow B}$ along with a quantum memory system:%
\begin{multline}
\Theta_{\left(  A\rightarrow B\right)  \rightarrow\left(  C\rightarrow
D\right)  }(\mathcal{N}_{A\rightarrow B})\label{eq:pre-post-proc-superch}\\
=\mathcal{D}_{BM\rightarrow D}\circ\mathcal{N}_{A\rightarrow B}\circ
\mathcal{E}_{C\rightarrow AM},
\end{multline}
where $\mathcal{E}_{C\rightarrow AM}$ and $\mathcal{D}_{BM\rightarrow D}$ are
pre- and post-processing quantum channels, respectively \cite{CDP08}. 
This transformation is depicted in Figure~\ref{fig:superchannel}.

\section{General channel box transformation problem}

\label{sec:ch-box-tr-problem}We can now state one main problem for the
resource theory of asymmetric distinguishability for quantum channels, which
we call the channel box transformation problem. The goal of this problem is to
determine, for an input channel box $(\mathcal{N}_{A\rightarrow B}%
,\mathcal{M}_{A\rightarrow B})$ and an output channel box $(\mathcal{K}%
_{C\rightarrow D},\mathcal{L}_{C\rightarrow D})$, whether there exists a
superchannel $\Theta_{\left(  A\rightarrow B\right)  \rightarrow\left(
C\rightarrow D\right)  }$ such that the following transformation is possible:%
\begin{equation}
(\mathcal{N}_{A\rightarrow B},\mathcal{M}_{A\rightarrow B})~\underrightarrow
{\Theta}~(\mathcal{K}_{C\rightarrow D},\mathcal{L}_{C\rightarrow D}),
\label{eq:ch-box-tr-prob}%
\end{equation}
where the notation means that the following equations should be satisfied%
\begin{align}
\Theta_{\left(  A\rightarrow B\right)  \rightarrow\left(  C\rightarrow
D\right)  }(\mathcal{N}_{A\rightarrow B})  &  =\mathcal{K}_{C\rightarrow D},\\
\Theta_{\left(  A\rightarrow B\right)  \rightarrow\left(  C\rightarrow
D\right)  }(\mathcal{M}_{A\rightarrow B})  &  =\mathcal{L}_{C\rightarrow D}.
\end{align}
This problem was introduced and solved in \cite{G18}, in the sense that the answer to this
question can be determined by means of a semi-definite program or by employing
the extended conditional min-entropy and a quantum dynamic generalization of majorization. The problem there was called ``comparison of quantum channels.'' Note that the simpler problem regarding transformation of state boxes via a common quantum channel has a long history, having been considered extensively both in classical and quantum information theory \cite{B53,AU80,CJW04,MOA11,Buscemi2012,HJRW12,BDS14,BHNOW15,Ren16,BD16,Buscemi2016,GJBDM18,B17,BG17}.

In many cases of interest, the transformation in \eqref{eq:ch-box-tr-prob} is
simply not possible. Thus, it is sensible to modify the problem to allow for
approximation, and the way that we do so is consistent with how we did so for the related problem in
the resource theory of asymmetric distinguishability for states
\cite{WW19states}. Namely, we allow for an approximation error in the
transformation of the first channel in the box, but we demand that the second
channel be simulated exactly (hence the descriptor ``asymmetric'' in ``resource theory of asymmetric distinguishability''). Mathematically, this corresponds to the
following optimization problem:%
\begin{multline}
\varepsilon((\mathcal{N},\mathcal{M})\rightarrow(\mathcal{K},\mathcal{L}%
)):=\label{eq:approx-box-trans-prob}\\
\inf_{\Theta\in\mathrm{SC}}\left\{  \varepsilon\in\left[  0,1\right]
:\Theta(\mathcal{N})\approx_{\varepsilon}\mathcal{K},\ \Theta(\mathcal{M}%
)=\mathcal{L}\right\}  ,
\end{multline}
where SC\ denotes the set of superchannels and the shorthand $\mathcal{N}%
^{1}\approx_{\varepsilon}\mathcal{N}^{2}$ for channels $\mathcal{N}^{1}$ and
$\mathcal{N}^{2}$ is defined as follows:%
\begin{equation}
\mathcal{N}^{1}\approx_{\varepsilon}\mathcal{N}^{2}\qquad\Longleftrightarrow
\qquad\frac{1}{2}\left\Vert \mathcal{N}^{1}-\mathcal{N}^{2}\right\Vert
_{\diamond}\leq\varepsilon. \label{eq:approx-channels-sense}%
\end{equation}
In the above, $\left\Vert \mathcal{P}_{A\rightarrow B}\right\Vert _{\diamond}$
denotes the diamond norm \cite{Kitaev1997}\ of a Hermiticity-preserving map
$\mathcal{P}_{A\rightarrow B}$, defined as%
\begin{equation}
\left\Vert \mathcal{P}_{A\rightarrow B}\right\Vert _{\diamond}:=\sup
_{\rho_{RA}}\left\Vert \mathcal{P}_{A\rightarrow B}(\rho_{RA})\right\Vert
_{1},
\end{equation}
where the optimization is with respect to quantum states $\rho_{RA}$ and the
reference system $R$ can be arbitrarily large. However, note that the
following significant simplification holds%
\begin{equation}
\left\Vert \mathcal{P}_{A\rightarrow B}\right\Vert _{\diamond}:=\sup
_{\psi_{RA}}\left\Vert \mathcal{P}_{A\rightarrow B}(\psi_{RA})\right\Vert
_{1},
\end{equation}
where the optimization is with respect to pure-state inputs $\psi_{RA}$ with
the reference system $R$ isomorphic to the input system $A$.

Why do we adopt the diamond norm to measure the distance between two quantum
channels $\mathcal{N}_{A\rightarrow B}$ and $\mathcal{M}_{A\rightarrow B}$?
Related, how should we assess the performance of a quantum information
processing protocol in which the ideal channel to be simulated is
$\mathcal{N}_{A\rightarrow B}$ but the channel realized in practice is
$\mathcal{M}_{A\rightarrow B}$? Suppose that a third party is trying to assess
how distinguishable the actual channel $\mathcal{M}_{A\rightarrow B}$ is from
the ideal channel $\mathcal{N}_{A\rightarrow B}$. Such an individual has
access to both the input and output ports of the channel, and so the most
general strategy for the distinguisher to employ is to prepare a state
$\rho_{RA}$ of a reference system $R$ and the channel input system $A$. The
distinguisher transmits the $A$ system of $\rho_{RA}$ into the unknown
channel. After that, the distinguisher receives the channel output system $B$
and then performs a measurement described by the POVM $\{\Lambda_{RB}%
^{x}\}_{x}$ on the reference system $R$ and the channel output system~$B$. The
probability of obtaining a particular outcome $\Lambda_{RB}^{x}$ is given by
the Born rule. In the case that the unknown channel is $\mathcal{N}%
_{A\rightarrow B}$, this probability is $\operatorname{Tr}[\Lambda_{RB}%
^{x}\mathcal{N}_{A\rightarrow B}(\rho_{RA})]$, and in the case that the
unknown channel is $\mathcal{M}_{A\rightarrow B}$, this probability is
$\operatorname{Tr}[\Lambda_{RB}^{x}\mathcal{M}_{A\rightarrow B}(\rho_{RA})]$.
What we demand is that the deviation between the two probabilities
$\operatorname{Tr}[\Lambda_{RB}^{x}\mathcal{N}_{A\rightarrow B}(\rho_{RA})]$
and $\operatorname{Tr}[\Lambda_{RB}^{x}\mathcal{M}_{A\rightarrow B}(\rho
_{RA})]$ is no larger than some tolerance $\varepsilon$. Since this should be
the case for all possible input states and measurement outcomes, what we
demand mathematically is that
\begin{equation}
\sup_{\substack{\rho_{RA},\,\\\Lambda_{RB}}}|\operatorname{Tr}[\Lambda
_{RB}\mathcal{N}_{A\rightarrow B}(\rho_{RA})]-\operatorname{Tr}[\Lambda
_{RB}\mathcal{M}_{A\rightarrow B}(\rho_{RA})]|\leq\varepsilon,
\end{equation}
where $\rho_{RA}\geq0$, $\operatorname{Tr}[\rho_{RA}]=1$, and $0\leq
\Lambda_{RB}\leq I_{RB}$. As a consequence of a well known characterization of
trace distance from \cite{H69,Hel76}, we have that%
\begin{align}
&  \sup_{\rho_{RA},\Lambda_{RB}}|\operatorname{Tr}[\Lambda_{RB}(\mathcal{N}%
_{A\rightarrow B}-\mathcal{M}_{A\rightarrow B})(\rho_{RA})]|\nonumber\\
&  =\sup_{\rho_{RA}}\frac{1}{2}\left\Vert \mathcal{N}_{A\rightarrow B}%
(\rho_{RA})-\mathcal{M}_{A\rightarrow B}(\rho_{RA})\right\Vert _{1}\\
&  =\frac{1}{2}\left\Vert \mathcal{N}-\mathcal{M}\right\Vert _{\diamond},
\label{eq-diamond_dist_meas}%
\end{align}
where $\frac{1}{2}\left\Vert \mathcal{N}-\mathcal{M}\right\Vert _{\diamond}$
is  the \textit{normalized diamond distance} between
$\mathcal{N}$ and $\mathcal{M}$. This indicates that if $\frac{1}{2}\left\Vert
\mathcal{N}-\mathcal{M}\right\Vert _{\diamond}\leq\varepsilon$, then the
deviation between probabilities for any possible input state and measurement
operator never exceeds $\varepsilon$, so that the approximation between
quantum channels $\mathcal{N}_{A\rightarrow B}$ and $\mathcal{M}_{A\rightarrow
B}$ is naturally quantified by the normalized diamond distance $\frac{1}{2}\left\Vert
\mathcal{N}-\mathcal{M}\right\Vert _{\diamond}$. We note that related
interpretations of the diamond distance of channels have been given in
\cite{KWerner04,RW05,GLN04}.

%\XW{
As we indicated above, the approximate channel box transformation problem is fundamental to the resource theory of asymmetric distinguishability, indicating exactly how well one can convert channel boxes. It captures distinguishability in a fundamental way: as pointed out in \cite{G18}, a necessary condition for a transformation to be possible exactly is if the two channels in one channel box are more distinguishable than the two channels in the target channel box, as quantified by a channel divergence \cite{LKDW18}. Thinking along the lines of \cite{BHNOW15}, these kinds of limitations from channel divergences can be interpreted as ``second laws'' for distinguishability that draw the line between the possible and impossible. As these lines might be too sharp for practical purposes (i.e., if the transformation were to be possible with small error), then it is sensible to consider the relaxation  presented in \eqref{eq:approx-box-trans-prob}. Furthermore, generalizations of the approximate box transformation problem will have applications in other resource theories of channels, such as entanglement, thermodynamics, purity, magic, etc., and therein can also be interpreted as second laws or approximate second laws.
%}

In Appendix~\ref{app:gen-ch-box-trans-SDP}, we show that the optimization in
\eqref{eq:approx-box-trans-prob}\ for the approximate channel box
transformation problem can be calculated by a semi-definite program, and thus
can be efficiently solved, where the complexity of the problem is polynomial
in the dimension of the inputs and outputs of the channels $(\mathcal{N}%
_{A\rightarrow B},\mathcal{M}_{A\rightarrow B})$ and $(\mathcal{K}%
_{C\rightarrow D},\mathcal{L}_{C\rightarrow D})$. This result generalizes the
recent finding in \cite{G18} mentioned after \eqref{eq:ch-box-tr-prob} above.

\section{One-shot distillation and dilution of quantum channel boxes}

\label{sec:one-shot-dist-dil}

Another way of addressing the general approximate channel box transformation
problem, which is helpful for considering asymptotic versions of the problem,
is to break it into two steps, as was done in \cite{WW19states} for the case
of states. Namely, one can first distill a standard channel box from the
original one, and then dilute this standard channel box to the final target
one. In this work, we take the standard channel box to be the following one:%
\begin{equation}
(\mathcal{R}_{|0\rangle\langle0|},\mathcal{R}_{\pi_{M}}),
\label{eq:standard-channel-box}%
\end{equation}
where $\mathcal{R}_{\sigma}$ denotes a replacer channel, which has the following
action on an arbitrary input $\rho$:%
\begin{equation}
\mathcal{R}_{\sigma}(\rho)=\operatorname{Tr}[\rho]\sigma,
\end{equation}
which is simply to discard the input $\rho$ and replace it with a state
$\sigma$. Also, the state $\pi_{M}$ is defined as%
\begin{equation}
\pi_{M}:=\frac{1}{M}|0\rangle\langle0|+\left(  1-\frac{1}{M}\right)
|1\rangle\langle1|,
\end{equation}
for $M\geq1$. Our interpretation of the channel box in
\eqref{eq:standard-channel-box} is that it contains $\log_{2}M$ bits of
asymmetric distinguishability. Since the replacer channel box in
\eqref{eq:standard-channel-box} is equivalent to the state channel box
$(|0\rangle\langle0|,\pi_{M})$, this interpretation is consistent with the
interpretation given in \cite{WW19states}.

\subsection{Exact one-shot distillation and dilution of quantum channel boxes}

\label{sec:exact-1shot-dilution-distillation}A primary goal in this setting is
the task of exact distillation of as many bits of asymmetric
distinguishability as possible, which is similar to the task for states
considered in \cite{WW19states}, but instead we allow for the most general
processing of the channel box according to a superchannel. Mathematically, we
can phrase this problem as the following optimization:%
\begin{multline}
D_{d}^{0}(\mathcal{N},\mathcal{M}):=\label{eq:exact-dist-channels}\\
\log_{2}\sup_{\Theta\in\mathrm{SC}}\{M:\Theta(\mathcal{N})=\mathcal{R}%
_{|0\rangle\langle0|},\ \Theta(\mathcal{M})=\mathcal{R}_{\pi_{M}}\}.
\end{multline}

We also consider exact dilution of the channel box, starting from as few bits
of asymmetric distinguishability as possible. The requirement here is to
convert bits of asymmetric distinguishability by the action of a common
superchannel to the channel box $(\mathcal{N},\mathcal{M})$ exactly, in such a
way that the number of bits $\log_{2}M$ of asymmetric distinguishability is as
small as possible. Mathematically, this corresponds to the following
optimization problem:%
\begin{multline}
D_{c}^{0}(\mathcal{N},\mathcal{M}):=\label{eq:exact-cost-channels}\\
\log_{2}\inf_{\Theta\in\mathrm{SC}}\{M:\Theta(\mathcal{R}_{|0\rangle\langle
0|})=\mathcal{N},\ \Theta(\mathcal{R}_{\pi_{M}})=\mathcal{M}\}.
\end{multline}

Let $D_{\min}(\mathcal{N}\Vert\mathcal{M})$ denote the channel min-relative
entropy, defined as%
\begin{equation}
D_{\min}(\mathcal{N}\Vert\mathcal{M}):=\sup_{\psi_{RA}}D_{\min}(\mathcal{N}%
_{A\rightarrow B}(\psi_{RA})\Vert\mathcal{M}_{A\rightarrow B}(\psi_{RA})),
\label{eq:ch-min-rel-ent}%
\end{equation}
with the min-relative entropy of states $\rho$ and $\sigma$ defined as
\cite{D09}%
\begin{equation}
D_{\min}(\rho\Vert\sigma):=-\log_2\operatorname{Tr}[\Pi_{\rho}\sigma],
\end{equation}
and $\Pi_{\rho}$ is the projection onto the support of $\rho$. Note that the
min-relative entropy of states is also equal to the Petz--R\'enyi relative
entropy of order zero \cite{P85,P86}, as observed in \cite{D09}.

Let $D_{\max}(\mathcal{N}\Vert\mathcal{M})$ denote the channel max-relative
entropy \cite{Cooney2016,LKDW18,GFWRSCW18}, defined as%
\begin{align}
&  D_{\max}(\mathcal{N}\Vert\mathcal{M})\nonumber\\
&  :=\sup_{\psi_{RA}}D_{\max}(\mathcal{N}_{A\rightarrow B}(\psi_{RA}%
)\Vert\mathcal{M}_{A\rightarrow B}(\psi_{RA}))\label{eq:ch-max-rel-ent}\\
&  =D_{\max}(\mathcal{N}_{A\rightarrow B}(\Phi_{RA})\Vert\mathcal{M}%
_{A\rightarrow B}(\Phi_{RA})), \label{eq:dmax-simplify}%
\end{align}
with the maximally entangled state $\Phi_{RA}$ of Schmidt rank $d$ defined as%
\begin{equation}
\Phi_{RA}:=\frac{1}{d}\sum_{i,j}|i\rangle\langle j|_{R}\otimes|i\rangle\langle
j|_{A},
\end{equation}
and the max-relative entropy of states $\rho$ and $\sigma$ defined as
\cite{D09}%
\begin{equation}
D_{\max}(\rho\Vert\sigma):=\inf\left\{  \lambda:\rho\leq2^{\lambda}%
\sigma\right\}  .
\end{equation}
The equality in \eqref{eq:dmax-simplify} was proved in \cite{GFWRSCW18,BHKW18}.

We then have the following fundamental result for exact distillation and dilution:

\begin{theorem}
\label{thm:exact-distill-cost}The exact one-shot distillable
distinguishability of the channel box $(\mathcal{N},\mathcal{M})$ is equal to
the channel min-relative entropy:%
\begin{equation}
D_{d}^{0}(\mathcal{N},\mathcal{M})=D_{\min}(\mathcal{N}\Vert\mathcal{M}),
\label{eq:exact-distill-dist-min}%
\end{equation}
and the exact one-shot distinguishability cost is equal to the channel
max-relative entropy:%
\begin{equation}
D_{c}^{0}(\mathcal{N},\mathcal{M})=D_{\max}(\mathcal{N}\Vert\mathcal{M}).
\label{eq:exact-dist-cost-max}%
\end{equation}

\end{theorem}

The equality in \eqref{eq:exact-distill-dist-min} is proved in
Appendix~\ref{app:ch-min-rel-ent-exact-distillable-dist}, and the equality in
\eqref{eq:exact-dist-cost-max} is proved in
Appendix~\ref{app:ch-max-rel-ent-exact-dist-cost}.

We remark that it is appealing that the exact one-shot distinguishability cost
of a channel box has a simple characterization in terms of the Choi states of
the channels $\mathcal{N}$ and $\mathcal{M}$, as indicated by the equality in \eqref{eq:dmax-simplify}.

\subsection{Approximate one-shot distillation and dilution of quantum channel
boxes}

\label{sec:approx-1shot-dilution-distillation}We also consider approximate
versions of these tasks. The goal of approximate distillation is to transform
the channel box $(\mathcal{N},\mathcal{M})$ into as many $\varepsilon
$-approximate bits of asymmetric distinguishability as possible.
Mathematically, this corresponds to the following optimization:%
\begin{multline}
D_{d}^{\varepsilon}(\mathcal{N},\mathcal{M}):=\\
\log_{2}\sup_{\Theta\in\mathrm{SC}}\{M:\Theta(\mathcal{N})\approx
_{\varepsilon}\mathcal{R}_{|0\rangle\langle0|},\ \Theta(\mathcal{M}%
)=\mathcal{R}_{\pi_{M}}\}.
\end{multline}

The goal of approximate dilution is to transform as few bits of asymmetric
distinguishability into a channel box $(\widetilde{\mathcal{N}}%
,\mathcal{M})$, such that $\widetilde{\mathcal{N}}\approx
_{\varepsilon}\mathcal{N}$. Mathematically, this corresponds to the following
optimization:%
\begin{multline}
D_{c}^{\varepsilon}(\mathcal{N},\mathcal{M}):=\\
\log_{2}\inf_{\Theta\in\mathrm{SC}}\{M:\Theta(\mathcal{R}_{|0\rangle\langle
0|})\approx_{\varepsilon}\mathcal{N},\ \Theta(\mathcal{R}_{\pi_{M}%
})=\mathcal{M}\}.
\end{multline}

Let $D_{\min}^{\varepsilon}(\mathcal{N}\Vert\mathcal{M})$ denote the 
smooth channel min-relative entropy from \cite{Cooney2016}, defined as%
\begin{equation}
D_{\min}^{\varepsilon}(\mathcal{N}\Vert\mathcal{M}):=\sup_{\psi_{RA}}D_{\min
}^{\varepsilon}(\mathcal{N}_{A\rightarrow B}(\psi_{RA})\Vert\mathcal{M}%
_{A\rightarrow B}(\psi_{RA})), \label{eq:ch-HT-rel-ent}%
\end{equation}
with the optimization being with respect to all pure states $\psi_{RA}$ with
system $R$ isomorphic to the channel input system $A$. The smooth min-relative entropy of states $\rho$ and $\sigma$ is defined as \cite{BD10,BD11,WR12}
\begin{equation}
D_{\min}^{\varepsilon}(\rho\Vert\sigma):=-\log_2 \inf_{\Lambda \geq 0} \{ \operatorname{Tr}[\Lambda \sigma] : 
\Lambda \leq I, \operatorname{Tr}[\Lambda \rho]\geq 1 - \varepsilon \}.
\end{equation}
The quantity $D_{\min}^{\varepsilon}(\rho\Vert\sigma)$ is also known as the hypothesis testing relative entropy \cite{WR12},
and $D_{\min
}^{\varepsilon}(\mathcal{N}\Vert\mathcal{M})$\ is also known as the channel
hypothesis testing relative entropy~\cite{Cooney2016}.

Let $D_{\max}^{\varepsilon}(\mathcal{N}\Vert\mathcal{M})$ denote the smooth channel
 max-relative entropy \cite[Definition~19]{GFWRSCW18}, defined as%
\begin{equation}
D_{\max}^{\varepsilon}(\mathcal{N}\Vert\mathcal{M}):=\inf_{\widetilde
{\mathcal{N}}\approx_{\varepsilon}\mathcal{N}}D_{\max}(\widetilde{\mathcal{N}%
}\Vert\mathcal{M}), \label{eq:ch-smooth-max-rel-ent}%
\end{equation}
 with the optimization being with respect to all quantum channels
$\widetilde{\mathcal{N}}$ satisfying $\widetilde{\mathcal{N}}\approx
_{\varepsilon}\mathcal{N}$, in the sense of \eqref{eq:approx-channels-sense}.
We note here that the smooth channel max-relative entropy has been studied
extensively in \cite{LW19}, in the context of resource erasure.

In Appendix~\ref{app:SDPs-channel-smooth-entropies}, we prove that the 
smooth channel min- and max-relative entropies can be calculated by semi-definite
programs. It follows from these characterizations that the non-smooth
quantities can be as well.

We then have the following result, endowing both the smooth channel min- and
max-relative entropies with fundamental operational meanings in the context of
the resource theory of asymmetric distinguishability:

\begin{theorem}
\label{thm:approx-distill-cost}The approximate one-shot distillable
distinguishability of the channel box $(\mathcal{N},\mathcal{M})$ is equal to
the smooth channel min-relative entropy:%
\begin{equation}
D_{d}^{\varepsilon}(\mathcal{N},\mathcal{M})=D_{\min}^{\varepsilon
}(\mathcal{N}\Vert\mathcal{M}), \label{eq:approx-distill-dist-smooth-min}%
\end{equation}
and the approximate one-shot distinguishability cost is equal to the 
smooth channel max-relative entropy:%
\begin{equation}
D_{c}^{\varepsilon}(\mathcal{N},\mathcal{M})=D_{\max}^{\varepsilon
}(\mathcal{N}\Vert\mathcal{M}). \label{eq:approx-dist-cost-smooth-max}%
\end{equation}

\end{theorem}

The equality in \eqref{eq:approx-distill-dist-smooth-min} is proved in
Appendix~\ref{app:ch-smooth-min-rel-ent-approx-distillable-dist}, and the
equality in \eqref{eq:approx-dist-cost-smooth-max} is proved in
Appendix~\ref{app:ch-smooth-max-rel-ent-approx-dist-cost}.

As a consequence of Theorems~\ref{thm:exact-distill-cost}\ and
\ref{thm:approx-distill-cost}, and the facts that%
\begin{align}
\lim_{\varepsilon\rightarrow0}D_{d}^{\varepsilon}(\mathcal{N},\mathcal{M})  &
=D_{d}^{0}(\mathcal{N},\mathcal{M}),\\
\lim_{\varepsilon\rightarrow0}D_{c}^{\varepsilon}(\mathcal{N},\mathcal{M})  &
=D_{c}^{0}(\mathcal{N},\mathcal{M}),
\end{align}
we conclude the following limits:%
\begin{align}
\lim_{\varepsilon\rightarrow0}D_{\min}^{\varepsilon}(\mathcal{N}%
\Vert\mathcal{M})  &  =D_{\min}(\mathcal{N}\Vert\mathcal{M}%
),\label{eq:smooth-to-exact-limits}\\
\lim_{\varepsilon\rightarrow0}D_{\max}^{\varepsilon}(\mathcal{N}%
\Vert\mathcal{M})  &  =D_{\max}(\mathcal{N}\Vert\mathcal{M}).
\label{eq:smooth-to-exact-limits-2}%
\end{align}
We give alternative proofs of these limits in
Appendix~\ref{app:limits-smooth-min-max-to-exact}.

As an application of the operational approach taken here, we arrive at the
following bound relating $D_{\min}^{\varepsilon_{1}}$ and~$D_{\max
}^{\varepsilon_{2}}$:%
\begin{equation}
D_{\min}^{\varepsilon_{1}}(\mathcal{N}\Vert\mathcal{M})\leq D_{\max
}^{\varepsilon_{2}}(\mathcal{N}\Vert\mathcal{M})+\log_{2}\!\left(  \frac
{1}{1-\varepsilon_{1}-\varepsilon_{2}}\right)  ,
\label{eq:smooth-min-max-bnd-ineq}%
\end{equation}
where $\varepsilon_{1},\varepsilon_{2}\geq0$ and $\varepsilon_{1}%
+\varepsilon_{2}<1$. This bound represents a generalization of a related bound
for quantum states in \cite{WW19states}, and it in fact reduces to it when the
channel box $(\mathcal{N},\mathcal{M})$ is environment seizable.

The main idea for arriving at the bound in
\eqref{eq:smooth-min-max-bnd-ineq}\ can be understood as a channel
generalization of the operational argument from \cite{WW19states}. As shown in \cite{WW19states}, any
approximate distillation protocol performed on the state box $(|0\rangle\langle0|,\pi_{M})$ that leads to the state box $(\widetilde{0}_{\varepsilon},\pi_{K})$, for
$\varepsilon\in\lbrack0,1)$ and $\widetilde{0}_{\varepsilon
}$  a state such that $\widetilde{0}_{\varepsilon
} \approx_{\varepsilon} |0\rangle \langle 0|$, is required to obey the bound%
\begin{equation}
\log_{2}K\leq\log_{2}M+\log_{2}(1/\left[  1-\varepsilon\right]  ).
\label{eq:fundamental-limit-2-step}
\end{equation}
One way to realize the full transformation%
\begin{equation}
(|0\rangle\langle0|,\pi_{M})
\rightarrow(\widetilde{0}_{\varepsilon},\pi_{K})
\end{equation}
is to proceed in two steps: use the equivalence $(|0\rangle\langle0|,\pi_{M}) \leftrightarrow (\mathcal{R}_{C\rightarrow D}^{|0\rangle\langle0|},\mathcal{R}_{C\rightarrow
D}^{\pi_{M}})$, first perform an optimal dilution protocol
$(\mathcal{R}_{C\rightarrow D}^{|0\rangle\langle0|},\mathcal{R}_{C\rightarrow
D}^{\pi_{M}})\rightarrow(\widetilde{\mathcal{N}},\mathcal{M})$, where
$\widetilde{\mathcal{N}}$ is a channel satisfying $\widetilde{\mathcal{N}%
}\approx_{\varepsilon_{2}}\mathcal{N}$ such that $\log_{2}M=D_{\max
}^{\varepsilon_{2}}(\mathcal{N}\Vert\mathcal{M})$ and then perform an optimal
distillation protocol $(\mathcal{N},\mathcal{M})\rightarrow(\widetilde{\mathcal{R}}%
_{C\rightarrow D}^{|0\rangle\langle 0|},\mathcal{R}_{C\rightarrow
D}^{\pi_{K}})$ such that $\log_{2}K=D_{\min}^{\varepsilon_{1}}(\mathcal{N}%
\Vert\mathcal{M})$. Finally, we realize the transformation $(\widetilde{\mathcal{R}}%
_{C\rightarrow D}^{|0\rangle\langle 0|},\mathcal{R}_{C\rightarrow
D}^{\pi_{K}})\rightarrow (\widetilde{0}_{\varepsilon_1},\pi_{K})$ by inputting any state to the final channel box. By employing the triangle inequality for the diamond
distance, the error of the overall transformation is no larger than
$\varepsilon_{1}+\varepsilon_{2}$. Since the fundamental limitation in
\eqref{eq:fundamental-limit-2-step}\ applies to any protocol, the bound
in \eqref{eq:smooth-min-max-bnd-ineq} follows.

\section{Parallel $n$-shot distillation and dilution of quantum channel boxes}

\label{sec:parallel-asymptotic-tasks}An important case to consider in the
resource theory of asymmetric distinguishability for channels is the case of
parallel tasks. In particular, we are interested in $n$-shot \textit{parallel}
distillation and dilution of channel boxes, which essentially amounts to the
replacement $(\mathcal{N},\mathcal{M})\rightarrow(\mathcal{N}^{\otimes
n},\mathcal{M}^{\otimes n})$ in our previous one-shot results from Section~\ref{sec:one-shot-dist-dil}. However, here
we are interested in optimal \textit{rates} at which one can distill or dilute
bits of asymmetric distinguishability from or to a channel box, respectively,
both in the exact and approximate cases.

\subsection{Exact case:\ distillable distinguishability}

\label{sec:exact-parallel-distill}We define the $n$-shot, parallel, exact
distillable distinguishability of a channel box $(\mathcal{N},\mathcal{M}%
)$\ as follows:%
\begin{equation}
\frac{1}{n}D_{d}^{0}(\mathcal{N}^{\otimes n},\mathcal{M}^{\otimes n})=\frac
{1}{n}D_{\min}(\mathcal{N}^{\otimes n}\Vert\mathcal{M}^{\otimes n}),
\end{equation}
noting that it is equal to the optimal rate at which one can distill exact
bits of asymmetric distinguishability for fixed $n\geq1$.\ The equality above
is a direct consequence of \eqref{eq:exact-distill-dist-min}.

The asymptotic parallel exact distillable distinguishability is then defined as%
\begin{align}
D_{d}^{0,p}(\mathcal{N},\mathcal{M})  &  :=\lim_{n\rightarrow\infty}\frac
{1}{n}D_{d}^{0}(\mathcal{N}^{\otimes n},\mathcal{M}^{\otimes n})\\
&  =\lim_{n\rightarrow\infty}\frac{1}{n}D_{\min}(\mathcal{N}^{\otimes n}%
\Vert\mathcal{M}^{\otimes n}), \label{eq:asymptotic-exact-dist-distingui}%
\end{align}
where the equality is again a direct consequence of \eqref{eq:exact-distill-dist-min}.

We note that the regularization in \eqref{eq:asymptotic-exact-dist-distingui}
seems to be necessary in general, due to the fact that $D_{\min}$ for channels
can be non-additive. As an example, suppose that $\mathcal{N}$ is the identity
channel and $\mathcal{M}$ is a unitary channel characterized by a unitary
operator $U$. Then it follows that%
\begin{align}
D_{\min}(\mathcal{N}\Vert\mathcal{M})  &  =-\log_2\inf_{|\psi\rangle_{RA}%
}\left\vert \langle\psi|_{RA}\left(  I_{R}\otimes U_{A}\right)  |\psi
\rangle_{RA}\right\vert ^{2}\nonumber\\
&  :=-\log_2 F(I,U).
\end{align}
It is known from \cite{Acin01}\ that there are unitaries for which
$F(I,U)\in(0,1)$ but $F(I^{\otimes n},U^{\otimes n})=0$ for some finite $n$.
Turning this around, we conclude that there are channels for which%
\begin{equation}
D_{\min}(\mathcal{N}\Vert\mathcal{M})<\infty,
\end{equation}
but%
\begin{equation}
D_{\min}(\mathcal{N}^{\otimes n}\Vert\mathcal{M}^{\otimes n})=\infty,
\end{equation}
for some finite $n$, indicating that the channel min-relative entropy exhibits
an extreme form of non-additivity.

A special case for which the exact distillable distinguishability simplifies
is for environment-seizable channels. As a consequence of the observation in
\eqref{eq:env-seize-ch-boxes}, an immediate conclusion is the following
equality:%
\begin{align}
\frac{1}{n}D_{d}^{0}(\mathcal{N}^{\otimes n},\mathcal{M}^{\otimes n})  &
=\frac{1}{n}D_{\min}(\rho_{E}^{\otimes n}\Vert\sigma_{E}^{\otimes n})\\
&  =D_{\min}(\rho_{E}\Vert\sigma_{E}),
\end{align}
which holds for any channel box $(\mathcal{N},\mathcal{M})$ that is
environment seizable in the sense of \eqref{eq:env-seize-ch-boxes}. The first
equality follows from \eqref{eq:env-seize-ch-boxes}, and the second follows
from the additivity of the min-relative entropy for states. We thus conclude
that the asymptotic exact parallel distillable distinguishability has the
following single-letter formula for the case of environment-seizable channel
boxes:%
\begin{equation}
D_{d}^{0,p}(\mathcal{N},\mathcal{M})=D_{\min}(\rho_{E}\Vert\sigma_{E}).
\end{equation}

\subsection{Exact case:\ distinguishability cost}

\label{sec:exact-parallel-cost}We define the $n$-shot, parallel, exact
distinguishability cost of a channel box $(\mathcal{N},\mathcal{M})$\ as
follows:%
\begin{equation}
\frac{1}{n}D_{c}^{0}(\mathcal{N}^{\otimes n},\mathcal{M}^{\otimes n})=D_{\max
}(\mathcal{N}\Vert\mathcal{M}),
\end{equation}
noting that it is equal to the optimal rate at which one can dilute exact
bits of asymmetric distinguishability to the channel box $(\mathcal{N}^{\otimes n},\mathcal{M}^{\otimes n})$ for fixed $n\geq1$.\ The equality above
is a direct consequence of \eqref{eq:exact-dist-cost-max} and the additivity
of the max-relative entropy of channels, due to the fact that \eqref{eq:dmax-simplify}\ holds.

The asymptotic exact distinguishability cost is then defined as%
\begin{align}
D_{c}^{0,p}(\mathcal{N},\mathcal{M})  &  :=\lim_{n\rightarrow\infty}\frac
{1}{n}D_{c}^{0}(\mathcal{N}^{\otimes n},\mathcal{M}^{\otimes n})\\
&  =D_{\max}(\mathcal{N}\Vert\mathcal{M}), \label{eq:par-exact-cost-dmax}%
\end{align}
where the equality is again a direct consequence of \eqref{eq:exact-dist-cost-max}.

Thus, exact distinguishability dilution in the parallel case is rather
different from exact distinguishability distillation, given that we have a simple
single-letter formula characterizing all channel boxes for the former case but not for the latter.

\subsection{Approximate case:\ distillable distinguishability}

\label{sec:approx-parallel-distill}

We define the $n$-shot, parallel,
$\varepsilon$-approximate distillable distinguishability as follows:%
\begin{equation}
\frac{1}{n}D_{d}^{\varepsilon}(\mathcal{N}^{\otimes n},\mathcal{M}^{\otimes
n}),
\end{equation}
noting that it is equal to the optimal rate at which one can distill
approximate bits of asymmetric distinguishability for fixed $n\geq1$ and
$\varepsilon\in(0,1)$. The asymptotic parallel distillable distinguishability
of the channel box $(\mathcal{N},\mathcal{M})$ is then defined as the
following limit of the above formula:%
\begin{align}
D_{d}^{p}(\mathcal{N},\mathcal{M})  &  :=\lim_{\varepsilon\rightarrow0}%
\lim_{n\rightarrow\infty}\frac{1}{n}D_{d}^{\varepsilon}(\mathcal{N}^{\otimes
n},\mathcal{M}^{\otimes n}) \label{eq:parallel-distill-dist-asymp}\\
&  =\lim_{\varepsilon\rightarrow0}\lim_{n\rightarrow\infty}\frac{1}{n}D_{\min
}^{\varepsilon}(\mathcal{N}^{\otimes n}\Vert\mathcal{M}^{\otimes n}),
\end{align}
where the latter equality follows from \eqref{eq:approx-distill-dist-smooth-min}.

Note that the quantity in \eqref{eq:parallel-distill-dist-asymp} is equal to the optimal exponent in Stein's lemma for the case of parallel  quantum channel discrimination \cite{Cooney2016}. The following theorem gives a formal expression for this quantity in terms of the regularized channel relative entropy.

\begin{theorem}
\label{thm:parallel-distillable-dist-regularized-rel-ent}The parallel
distillable distinguishability of the channel box $(\mathcal{N},\mathcal{M}%
)$\ is equal to the regularized channel relative entropy:%
\begin{equation}
D_{d}^{p}(\mathcal{N},\mathcal{M})=\lim_{m\rightarrow\infty}\frac{1}%
{m}D(\mathcal{N}^{\otimes m}\Vert\mathcal{M}^{\otimes m}),
\end{equation}
and it is finite if and only if $D_{\max}(\mathcal{N}\Vert\mathcal{M})<\infty$.
\end{theorem}

\begin{proof}
By exploiting the following bound for states $\rho$ and $\sigma$
\cite{WR12,MW12,KW17},%
\begin{equation}
D_{\min}^{\varepsilon}(\rho\Vert\sigma)\leq\frac{1}{1-\varepsilon}\left[
D(\rho\Vert\sigma)+h_{2}(\varepsilon)\right]  ,
\end{equation}
where $h_{2}(\varepsilon):=-\varepsilon\log_{2}\varepsilon-\left(
1-\varepsilon\right)  \log_{2}(1-\varepsilon)$ is the binary entropy, we
conclude the following bound for channels after an optimization:%
\begin{equation}
D_{\min}^{\varepsilon}(\mathcal{N}\Vert\mathcal{M})\leq\frac{1}{1-\varepsilon
}\left[  D(\mathcal{N}\Vert\mathcal{M})+h_{2}(\varepsilon)\right]  .
\end{equation}
By making the substitution $(\mathcal{N},\mathcal{M})\rightarrow
(\mathcal{N}^{\otimes m},\mathcal{M}^{\otimes m})$, dividing by $m$, and taking the limit
$m\rightarrow\infty$ followed by $\varepsilon\rightarrow0$, we conclude that%
\begin{equation}
D_{d}^{p}(\mathcal{N},\mathcal{M})\leq\lim_{m\rightarrow\infty}\frac{1}%
{m}D(\mathcal{N}^{\otimes m}\Vert\mathcal{M}^{\otimes m}).
\label{eq:distill-disting-regularized-up-bnd}%
\end{equation}

Also, note that the following lower bound holds as a consequence of the lower
bound from \cite{li14,tomamichel2013hierarchy}:%
\begin{multline}
D_{\min}^{\varepsilon}(\mathcal{N}^{\otimes n}\Vert\mathcal{M}^{\otimes
n})\geq nD(\mathcal{N}\Vert\mathcal{M}%
)\label{eq:iid-case-distill-parallel-2nd-order}\\
+\sqrt{nV_{\varepsilon}(\mathcal{N}\Vert\mathcal{M})}\Phi^{-1}(\varepsilon
)+O(\log n),
\end{multline}
where $\Phi^{-1}$ is the inverse of the cumulative standard normal
distribution function, $V_{\varepsilon}(\mathcal{N}\Vert\mathcal{M})$ is the
channel relative entropy variance, defined as%
\begin{multline}
V_{\varepsilon}(\mathcal{N}\Vert\mathcal{M}):=\\
\left\{
\begin{array}
[c]{cc}%
\inf_{\psi_{RA}\in\Pi}V(\mathcal{N}_{A\rightarrow B}(\psi_{RA})\Vert
\mathcal{M}_{A\rightarrow B}(\psi_{RA})) & \text{if\ }\varepsilon<1/2\\
\sup_{\psi_{RA}\in\Pi}V(\mathcal{N}_{A\rightarrow B}(\psi_{RA})\Vert
\mathcal{M}_{A\rightarrow B}(\psi_{RA})) & \text{else}%
\end{array}
\right.  ,
\end{multline}
with $\Pi\,$the set of all bipartite pure states achieving the optimal value of
$D(\mathcal{N}\Vert\mathcal{M})$ and the relative entropy variance
$V(\rho\Vert\sigma)$ of states $\rho$ and $\sigma$ defined as%
\begin{equation}
V(\rho\Vert\sigma):=\operatorname{Tr}[\rho\left(  \log_{2}\rho-\log_{2}%
\sigma-D(\rho\Vert\sigma)\right)  ^{2}].
\end{equation}
Taking the limit as $n\rightarrow\infty$ and $\varepsilon\rightarrow0$, we
find that%
\begin{equation}
\lim_{\varepsilon\rightarrow0}\lim_{n\rightarrow\infty}\frac{1}{n}D_{\min
}^{\varepsilon}(\mathcal{N}^{\otimes n}\Vert\mathcal{M}^{\otimes n})\geq
D(\mathcal{N}\Vert\mathcal{M}).
\end{equation}
However, we can also conclude the following bound%
\begin{multline}
D_{\min}^{\varepsilon}(\mathcal{N}^{\otimes nm}\Vert\mathcal{M}^{\otimes
nm})\geq nD(\mathcal{N}^{\otimes m}\Vert\mathcal{M}^{\otimes m})\\
+\sqrt{nV_{\varepsilon}(\mathcal{N}^{\otimes m}\Vert\mathcal{M}^{\otimes m}%
)}\Phi^{-1}(\varepsilon)+O(\log n),
\end{multline}
by making the substitution $(\mathcal{N},\mathcal{M})\rightarrow
(\mathcal{N}^{\otimes m},\mathcal{M}^{\otimes m})$ in
\eqref{eq:iid-case-distill-parallel-2nd-order}, from which we conclude that%
\begin{equation}
D_{d}^{p}(\mathcal{N},\mathcal{M})\geq\frac{1}{m}D(\mathcal{N}^{\otimes
m}\Vert\mathcal{M}^{\otimes m}),
\end{equation}
for all $m\geq1$. Since this bound holds for all $m$, we can take the limit,
and when combining with \eqref{eq:distill-disting-regularized-up-bnd}, we
conclude that%
\begin{equation}
D_{d}^{p}(\mathcal{N},\mathcal{M})=\lim_{m\rightarrow\infty}\frac{1}%
{m}D(\mathcal{N}^{\otimes m}\Vert\mathcal{M}^{\otimes m}).
\end{equation}

As observed in \cite[Remark~19]{BHKW18}, the regularized channel relative
entropy on the right-hand side is finite if and only if $D_{\max}%
(\mathcal{N}\Vert\mathcal{M})<\infty$.
\end{proof}

\bigskip

An important case in which the situation simplifies considerably is for
environment-seizable channel boxes, as identified in \cite{BHKW18}. As a
consequence of the observation in \eqref{eq:env-seize-ch-boxes}, an immediate
conclusion is the following equality%
\begin{equation}
\frac{1}{n}D_{d}^{\varepsilon}(\mathcal{N}^{\otimes n},\mathcal{M}^{\otimes
n})=\frac{1}{n}D_{\min}^{\varepsilon}(\rho_{E}^{\otimes n}\Vert\sigma
_{E}^{\otimes n}),
\end{equation}
for any channel box $(\mathcal{N},\mathcal{M})$ that is environment seizable
in the sense of \eqref{eq:env-seize-ch-boxes}. For such channels, we can even
conclude the following expansion:%
\begin{multline}
\frac{1}{n}D_{d}^{\varepsilon}(\mathcal{N}^{\otimes n},\mathcal{M}^{\otimes
n})=D(\rho_{E}\Vert\sigma_{E})\\
+\sqrt{\frac{1}{n}V(\rho_{E}\Vert\sigma_{E})}\Phi^{-1}(\varepsilon)+O(\log n),
\end{multline}
so that%
\begin{equation}
D_{d}^{p}(\mathcal{N},\mathcal{M})=D(\rho_{E}\Vert\sigma_{E})
\end{equation}
for such environment-seizable channel boxes.

Another important case for which we have a handle on the distillable
distinguishability is classical--quantum channel boxes, defined as%
\begin{align}
\mathcal{N}_{X\rightarrow B}(\omega_{X})  &  :=\sum_{x}\langle x|_{X}%
\omega_{X}|x\rangle_{X}\rho_{B}^{x},\label{eq:cq-ch-def-1}\\
\mathcal{M}_{X\rightarrow B}(\omega_{X})  &  :=\sum_{x}\langle x|_{X}%
\omega_{X}|x\rangle_{X}\sigma_{B}^{x}, \label{eq:cq-ch-def-2}%
\end{align}
where $\omega_{X}$ is an arbitrary input state, $\{|x\rangle_{X}\}_{x}$ is an
orthonormal basis, and $\left\{  \rho_{B}^{x}\right\}  _{x}$ and $\left\{
\sigma_{B}^{x}\right\}  _{x}$ are sets of states. An immediate consequence of
\cite[Corollary~28]{BHKW18}\ is the following equality for classical--quantum
channel boxes:%
\begin{equation}
D_{d}^{p}(\mathcal{N}_{X\rightarrow B},\mathcal{M}_{X\rightarrow B})=\sup
_{x}D(\rho_{B}^{x}\Vert\sigma_{B}^{x}).
\label{eq:cq-parallel-distill}
\end{equation}
%\XW{
Eq.~\eqref{eq:cq-parallel-distill} indicates that the asymptotic parallel distillable distinguishability  of a classical-quantum channel box  depends only on the maximum quantum relative entropy that can be realized by the input of a single classical state to the channels.
%}

\subsection{Approximate case:\ distinguishability cost}

\label{sec:approx-parallel-cost}We define the $n$-shot, parallel,
$\varepsilon$-approximate distinguishability cost as follows:%
\begin{equation}
\frac{1}{n}D_{c}^{\varepsilon}(\mathcal{N}^{\otimes n},\mathcal{M}^{\otimes
n}),
\end{equation}
noting that it is equal to the optimal rate at which one can dilute the
channel box $(\mathcal{N}^{\otimes m},\mathcal{M}^{\otimes m})$ approximately
from bits of asymmetric distinguishability for fixed $n\geq1$ and
$\varepsilon\in(0,1)$. The asymptotic parallel distinguishability cost of the
channel box $(\mathcal{N},\mathcal{M})$ is then defined as the following limit
of the above formula:%
\begin{align}
D_{c}^{p}(\mathcal{N},\mathcal{M})  &  :=\lim_{\varepsilon\rightarrow0}%
\lim_{n\rightarrow\infty}\frac{1}{n}D_{c}^{\varepsilon}(\mathcal{N}^{\otimes
n},\mathcal{M}^{\otimes n})\\
&  =\lim_{\varepsilon\rightarrow0}\lim_{n\rightarrow\infty}\frac{1}{n}D_{\max
}^{\varepsilon}(\mathcal{N}^{\otimes n}\Vert\mathcal{M}^{\otimes n}),
\end{align}
where the latter equality follows from \eqref{eq:approx-dist-cost-smooth-max}.

As a direct consequence of the inequality in
\eqref{eq:smooth-min-max-bnd-ineq} and
Theorem~\ref{thm:parallel-distillable-dist-regularized-rel-ent}, we find that%
\begin{equation}
D_{c}^{p}(\mathcal{N},\mathcal{M})\geq\lim_{m\rightarrow\infty}\frac{1}%
{m}D(\mathcal{N}^{\otimes m}\Vert\mathcal{M}^{\otimes m}).
\label{eq:lower-bnd-asymp-cost-regularized-rel-ent}%
\end{equation}
We note that an inequality similar to the above one, which does not include
regularization, has been reported as \cite[Theorem~11]{LW19}. Whether the
lower bound in \eqref{eq:lower-bnd-asymp-cost-regularized-rel-ent} is also an upper bound remains an open question. However, the
following upper bound holds as a consequence of definitions and the fact that
the channel max-relative entropy is single-letter:%
\begin{equation}
D_{c}^{p}(\mathcal{N},\mathcal{M})\leq D_{\max}(\mathcal{N}\Vert\mathcal{M}).
\end{equation}
Furthermore, from this upper bound and \cite[Remark~19]{BHKW18}, we conclude
that the asymptotic parallel distinguishability cost is finite if and only if
$D_{\max}(\mathcal{N}\Vert\mathcal{M})$ is.

Although we have not been able to solve the asymptotic parallel
distinguishability cost in general, we can do so for some interesting special
cases. First, for any channel box $(\mathcal{N},\mathcal{M})$ that is
environment seizable, in the sense of \eqref{eq:env-seize-ch-boxes}, an
immediate conclusion is the following equality:%
\begin{equation}
\frac{1}{n}D_{c}^{\varepsilon}(\mathcal{N}^{\otimes n},\mathcal{M}^{\otimes
n})=\frac{1}{n}D_{\max}^{\varepsilon}(\rho_{E}^{\otimes n}\Vert\sigma
_{E}^{\otimes n}). \label{eq:n-shot-cost-env-seize}%
\end{equation}
Then as a consequence of the asymptotic equipartition property for states \cite{TCR09}, by
taking the limit $n\rightarrow\infty$ of \eqref{eq:n-shot-cost-env-seize}, it
follows that%
\begin{equation}
D_{c}^{p}(\mathcal{N},\mathcal{M})=D(\rho_{E}\Vert\sigma_{E}),
\end{equation}
thus demonstrating a complete understanding of the asymptotic cost for these
channel boxes. As in \cite{WW19states}, one can make refined statements (for second-order expansions) of $\frac{1}{n}D_{c}^{\varepsilon}(\mathcal{N}^{\otimes n},\mathcal{M}^{\otimes
n})$ for such channels.

Another important case for which we have a handle on the distinguishability
cost are classical--quantum channel boxes $(\mathcal{N}_{X\rightarrow
B},\mathcal{M}_{X\rightarrow B})$, with a common classical input alphabet and
output Hilbert space, defined as in \eqref{eq:cq-ch-def-1}--\eqref{eq:cq-ch-def-2}:

\begin{proposition}
\label{prop:cq-cost-parallel}
Let $(\mathcal{N}_{X\rightarrow B},\mathcal{M}_{X\rightarrow B})$ be a
classical--quantum channel box as in \eqref{eq:cq-ch-def-1}--\eqref{eq:cq-ch-def-2}. Then the asymptotic parallel distinguishability cost is equal to the channel relative entropy:
\begin{equation}
D_{c}^{p}(\mathcal{N}_{X\rightarrow B},\mathcal{M}_{X\rightarrow B})=\sup
_{x}D(\rho_{B}^{x}\Vert\sigma_{B}^{x}).
\end{equation}

\end{proposition}

\begin{proof}
It is known from \cite{BHKW18} that the following identity holds for
classical--quantum channel boxes:%
\begin{equation}
D(\mathcal{N}_{X\rightarrow B}\Vert\mathcal{M}_{X\rightarrow B})=\sup
_{x}D(\rho_{B}^{x}\Vert\sigma_{B}^{x}). \label{eq:cq-rel-ent}%
\end{equation}
Thus, the lower bound%
\begin{equation}
D_{c}^{p}(\mathcal{N}_{X\rightarrow B},\mathcal{M}_{X\rightarrow B})\geq
\sup_{x}D(\rho_{B}^{x}\Vert\sigma_{B}^{x})
\end{equation}
is a direct consequence of \eqref{eq:lower-bnd-asymp-cost-regularized-rel-ent}
and \eqref{eq:cq-rel-ent}.

To establish the upper bound, we make use of
Proposition~\ref{prop:cq-smooth-dmax-up-bnd}\ from
Appendix~\ref{app:smooth-max-rel-ent-bnd-cq-ch}, which states that the
following inequality holds for all $\alpha>1$ and $\varepsilon\in(0,1)$:%
\begin{multline}
D_{\max}^{\varepsilon}(\mathcal{N}_{X\rightarrow B}\Vert\mathcal{M}%
_{X\rightarrow B})\leq\widetilde{D}_{\alpha}(\mathcal{N}_{X\rightarrow B}%
\Vert\mathcal{M}_{X\rightarrow B})\\
+\frac{1}{\alpha-1}\log_{2}\!\left(  \frac{1}{\varepsilon^{2}}\right)
+\log_{2}\!\left(  \frac{1}{1-\varepsilon^{2}}\right)  .
\end{multline}
As such, we apply this inequality to the channel box $(\mathcal{N}%
_{X\rightarrow B}^{\otimes n},\mathcal{M}_{X\rightarrow B}^{\otimes n})$, as
well as \cite[Lemma~25]{BHKW18}, to find that the following inequality holds for all $\alpha>1$ and
$\varepsilon\in(0,1)$:%
\begin{multline}
\frac{1}{n}D_{\max}^{\varepsilon}(\mathcal{N}^{\otimes n}_{X\rightarrow B}\Vert
\mathcal{M}^{\otimes n}_{X\rightarrow B})\leq\widetilde{D}_{\alpha}(\mathcal{N}%
_{X\rightarrow B}\Vert\mathcal{M}_{X\rightarrow B})\\
+\frac{1}{n\left(  \alpha-1\right)  }\log_{2}\!\left(  \frac{1}{\varepsilon
^{2}}\right)  +\frac{1}{n}\log_{2}\!\left(  \frac{1}{1-\varepsilon^{2}%
}\right)  .
\end{multline}
Taking the limit as $n\rightarrow\infty$, we find that the following
inequality holds for all $\alpha>1$:%
\begin{multline}
\limsup_{n\rightarrow\infty}\frac{1}{n}D_{\max}^{\varepsilon}(\mathcal{N}%
_{X\rightarrow B}\Vert\mathcal{M}_{X\rightarrow B})\\
\leq\widetilde{D}_{\alpha}(\mathcal{N}_{X\rightarrow B}\Vert\mathcal{M}%
_{X\rightarrow B}).
\end{multline}
Now taking the limit as $\alpha\rightarrow1$, we conclude that%
\begin{align}
&  \limsup_{n\rightarrow\infty}\frac{1}{n}D_{\max}^{\varepsilon}%
(\mathcal{N}_{X\rightarrow B}\Vert\mathcal{M}_{X\rightarrow B})\nonumber\\
&  \leq D(\mathcal{N}_{X\rightarrow B}\Vert\mathcal{M}_{X\rightarrow B})\\
&  =\sup_{x}D(\rho_{B}^{x}\Vert\sigma_{B}^{x}).
\end{align}
This concludes the proof.
\end{proof}

\bigskip
%\XW{
Proposition~\ref{prop:cq-cost-parallel} indicates that the asymptotic parallel distinguishability cost of a classical-quantum channel box  depends only on the maximum quantum relative entropy that can be realized by the input of a single classical state to the channels. As such, when combined with the result from \eqref{eq:cq-parallel-distill}, we conclude that the resource theory of asymmetric distinguishability is reversible in the asymptotic setting of parallel channel box transformations when restricted to classical--quantum channel boxes, meaning that one can convert between such channel boxes without any loss. We provide further related remarks about this observation in the next section.
%}

\section{General channel box transformation: Parallel case}

\label{sec:gen-box-trans-parallel}We can now address the general channel box
transformation problem for the parallel case. Before doing so, let us
formalize the problem. Let $n,m\in\mathbb{Z}^{+}$ and $\varepsilon\in
\lbrack0,1]$. An $(n,m,\varepsilon)$ parallel channel box transformation
protocol for the channel boxes $(\mathcal{N},\mathcal{M})$ and $(\mathcal{K}%
,\mathcal{L})$ consists of a superchannel $\Theta^{(n)}$ such that
\begin{align}
\Theta^{(n)}(\mathcal{N}^{\otimes n})
& \approx_{\varepsilon}\mathcal{K}
^{\otimes m} , \\ \Theta^{(n)}(\mathcal{M}^{\otimes n}) & =\mathcal{L}^{\otimes
m}.
\end{align}

A rate $R$ is achievable if for all $\varepsilon\in(0,1]$, $\delta>0$, and
sufficiently large $n$, there exists an $(n,n\left[  R-\delta\right]
,\varepsilon)$ parallel channel box transformation protocol. The optimal
parallel channel box transformation rate $R^p((\mathcal{N},\mathcal{M}%
)\rightarrow(\mathcal{K},\mathcal{L}))$ is equal to the supremum of all
achievable rates.

On the other hand, a rate $R$ is a strong converse rate if for all
$\varepsilon\in\lbrack0,1)$, $\delta>0$, and sufficiently large $n$, there
does not exist an $(n,n\left[  R+\delta\right]  ,\varepsilon)$ parallel
channel box transformation protocol. The strong converse parallel channel box
transformation rate $\widetilde{R}^p((\mathcal{N},\mathcal{M})\rightarrow
(\mathcal{K},\mathcal{L}))$ is equal to the infimum of all strong converse rates.

Note that the following inequality is a consequence of the definitions:%
\begin{equation}
R^p((\mathcal{N},\mathcal{M})\rightarrow(\mathcal{K},\mathcal{L}))\leq
\widetilde{R}^p((\mathcal{N},\mathcal{M})\rightarrow(\mathcal{K},\mathcal{L})).
\end{equation}

An important result is that if the channel boxes $(\mathcal{N},\mathcal{M})$
and $(\mathcal{K},\mathcal{L})$ are either classical--quantum or
environment-seizable, then the following equality holds%
\begin{align}
R^p((\mathcal{N},\mathcal{M})\rightarrow(\mathcal{K},\mathcal{L}))  &
=\widetilde{R}^p((\mathcal{N},\mathcal{M})\rightarrow(\mathcal{K},\mathcal{L}%
))\\
&  =\frac{D(\mathcal{N}\Vert\mathcal{M})}{D(\mathcal{K}\Vert\mathcal{L})}, \label{eq:main-result-single-letter-parallel-trans}
\end{align}
indicating that the channel relative entropy plays a central role as the
optimal conversion rate between these kinds of channel boxes.
Appendix~\ref{app:gen-box-trans-bnds}\ provides detailed proofs of converse bounds that justify the claim in \eqref{eq:main-result-single-letter-parallel-trans}, by starting with converse bounds for generic one-shot channel box
transformation protocols and then applying them to the parallel case of interest (see also Appendix~\ref{app:smooth-max-lower-bnds} for how to translate some of these bounds to lower bounds on the smooth channel max-relative entropy). The achievability part follows from combining a distillation protocol with a dilution protocol (as was done for states in \cite{WW19states}) and the fact that these tasks have simple characterizations for these channel boxes.

\section{General box transformation:\ Sequential channels and quantum
strategies}

\label{sec:seq-case-gen-trans}We now move on to consider another variant of
the general channel box transformation problem corresponding to the sequential
case. This case is more involved than the parallel case considered above
because it cannot be reduced to the one-shot case. That is, it is
fundamentally a multi-shot problem, and the theory relies upon key
developments from \cite{CDP09}. As such, we develop the theory more generally
for quantum strategies \cite{GW07} or quantum combs \cite{CDP09} and then apply it to sequential
channel boxes, which are a special case of quantum strategies. A quantum strategy consists of a sequence of quantum channels, each of which has an accessible input and output, while passing along an internal memory system that can vary in size \cite{GW07}.  We remark here that there are various terms to refer to this same physical object, including quantum memory channels \cite{PhysRevA.72.062323,RevModPhys.86.1203}, quantum strategies \cite{GW07,G09,G12,Gutoski2018fidelityofquantum}, and quantum combs \cite{CDP09}, and there are even earlier works where similar notions appear \cite{PhysRevA.64.052309,ESW02}. Here we adopt the terminology ``quantum strategy'' to refer to such an object.

%\XW{
The main reason for considering the more complicated quantum strategies  is that doing so leads to a better understanding and simplification of the analysis of sequential channel boxes, while at the same time providing a significant generalization of the theory. Indeed, regarding this latter point, one might think of generalizing the theory even further by considering physical transformations of quantum strategies and even an infinite hierarchy of this sort, just as we generalized the resource theory of states to channels by considering physical transformations of channels in the form of superchannels. However, a key insight of \cite{CDP09} is that quantum strategies are the end of the line: physical transformations of quantum strategies are simply quantum strategies, so that the hierarchy ends with quantum strategies. Thus, the theory developed here in this sense is a rather general resource theory of asymmetric distinguishability.
%}

\subsection{Quantum strategies and sequential channel boxes}

The basic object to manipulate in this setting is a quantum strategy box or a
sequential channel box $(\mathcal{N}^{(n)},\mathcal{M}^{(n)})$. A sequential
channel box is a special case of a quantum strategy box, and since it is
simpler, we discuss it briefly first. For a sequential channel box, the
notation $\mathcal{N}^{(n)}$ indicates $n$ sequential uses of the channel
$\mathcal{N}_{A\rightarrow B}$ and $\mathcal{M}^{(n)}$ indicates $n$
sequential uses of the channel $\mathcal{M}_{A\rightarrow B}$. Sequential
channel boxes have been considered implicitly in previous work on sequential
quantum channel discrimination \cite{CDP08a,CDP09,Duan09,Harrow10,Cooney2016,BHKW18}.

More generally, a quantum strategy $\mathcal{N}^{(n)}$ consists of a sequence
of channels $\mathcal{N}_{A_{1}\rightarrow M_{1}B_{1}}^{1}$, $\mathcal{N}%
_{M_{1}A_{2}\rightarrow M_{2}B_{2}}^{2}$, \ldots, $\mathcal{N}_{M_{n-2}%
A_{n-1}\rightarrow M_{n-1}B_{n-1}}^{n-1}$, and $\mathcal{N}_{M_{n-1}%
A_{n}\rightarrow B_{n}}^{n}$, and the quantum strategy $\mathcal{M}^{(n)}$
consists of a sequence of channels $\mathcal{M}_{A_{1}\rightarrow M_{1}B_{1}%
}^{1}$, $\mathcal{M}_{M_{1}A_{2}\rightarrow M_{2}B_{2}}^{2}$, \ldots,
$\mathcal{M}_{M_{n-2}A_{n-1}\rightarrow M_{n-1}B_{n-1}}^{n-1}$, and
$\mathcal{M}_{M_{n-1}A_{n}\rightarrow B_{n}}^{n}$. As indicated above, quantum
strategies are in one-to-one correspondence with quantum combs
\cite{GW07,CDP08a,CDP09}. In order to have a uniform notation, we sometimes
write%
\begin{align}
\mathcal{N}^{(n)}  &  =(\mathcal{N}_{M_{i-1}A_{i}\rightarrow M_{i}B_{i}}%
^{i})_{i=1}^{n},\label{eq:q-strat-shorthand}\\
\mathcal{M}^{(n)}  &  =(\mathcal{M}_{M_{i-1}A_{i}\rightarrow M_{i}B_{i}}%
^{i})_{i=1}^{n},\label{eq:q-strat-shorthand-2}
\end{align}
where $M_{0}$ and $M_{n}$ are trivial registers.

It is straightforward to see that a quantum strategy box generalizes a
sequential channel box discussed above, with each element of a sequential
channel box being a sequence of the same channel without any memory. That is,
the sequential channel box is a special case of
\eqref{eq:q-strat-shorthand} and \eqref{eq:q-strat-shorthand-2} with $\mathcal{N}_{M_{i-1}A_{i}\rightarrow
M_{i}B_{i}}^{i}=\mathcal{N}_{A_{i}\rightarrow B_{i}}$ and
$\mathcal{M}_{M_{i-1}A_{i}\rightarrow
M_{i}B_{i}}^{i}=\mathcal{M}_{A_{i}\rightarrow B_{i}}$ for all $i\in\left\{
1,\ldots,n\right\}  $.

A quantum co-strategy \cite{GW07} (or tester \cite{CDP08a,CDP09}) for distinguishing two quantum strategies
consists of an input state $\rho_{R_{1}A_{1}}$ and a set of testing channels
$\{\mathcal{A}_{R_{i}B_{i}\rightarrow R_{i+1}A_{i+1}}^{i}\}_{i=1}^{n-1}$, such
that the final state when processing the first quantum strategy
$\mathcal{N}^{(n)}$ is given by%
\begin{multline}
\rho_{R_{n}B_{n}}:=\mathcal{N}_{M_{n-1}A_{n}\rightarrow B_{n}}^{n}%
\circ\label{eq:state-rho-seq-ch-1}\\
(\bigcirc_{i=1}^{n-1}\mathcal{A}_{R_{i}B_{i}\rightarrow R_{i+1}A_{i+1}}%
^{i}\circ\mathcal{N}_{M_{i-1}A_{i}\rightarrow M_{i}B_{i}}^{i})(\rho
_{R_{1}A_{1}})
\end{multline}
and the final state when processing the second quantum strategy
$\mathcal{M}^{(n)}$ is given by%
\begin{multline}
\sigma_{R_{n}B_{n}}:=\mathcal{M}_{M_{n-1}A_{n}\rightarrow B_{n}}^{n}%
\circ\label{eq:state-sig-seq-ch-2}\\
(\bigcirc_{i=1}^{n-1}\mathcal{A}_{R_{i}B_{i}\rightarrow R_{i+1}A_{i+1}}%
^{i}\circ\mathcal{M}_{M_{i-1}A_{i}\rightarrow M_{i}B_{i}}^{i})(\rho
_{R_{1}A_{1}}).
\end{multline}
Figure~\ref{fig:strat-with-co-strat} depicts the state $\rho_{R_n B_n}$ in \eqref{eq:state-rho-seq-ch-1} when $n=3$.
\begin{figure}[ptb]
\begin{center}
\includegraphics[
width=\linewidth
]{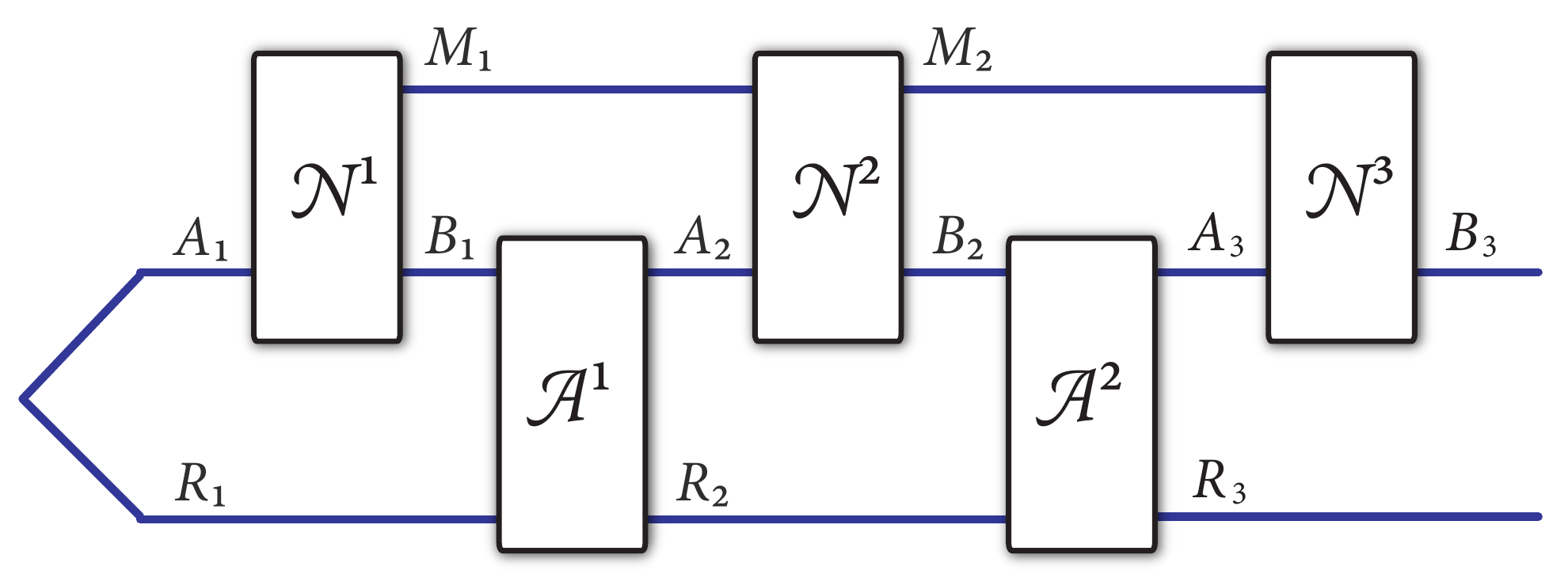}
\end{center}
\caption{A depiction of the state $\rho_{R_n B_n}$ in \eqref{eq:state-rho-seq-ch-1} with $n=3$, which results from the interaction of a three-round quantum strategy
$\mathcal{N}^{(3)}$ with a co-strategy.}%
\label{fig:strat-with-co-strat}%
\end{figure}

For our developments in this and the next section, it is helpful to define a generalized
quantum strategy divergence as an abstract measure of how distinguishable two
quantum strategies are.

\begin{definition}
[Generalized q.~strategy divergence]
\label{def:gen-strat-div}
The generalized quantum strategy
divergence of a quantum strategy box $(\mathcal{N}^{(n)},\mathcal{M}^{(n)})$
is defined as%
\begin{multline}
\mathbf{D}(\mathcal{N}^{(n)}\Vert\mathcal{M}^{(n)}):=\\
\sup_{\rho_{R_{1}A_{1}},\{\mathcal{A}_{R_{i}B_{i}\rightarrow R_{i+1}A_{i+1}%
}^{i}\}_{i=1}^{n-1}}\mathbf{D}(\rho_{R_{n}B_{n}}\Vert\sigma_{R_{n}B_{n}}),
\end{multline}
where the generalized divergence $\mathbf{D}$ for states is defined by \eqref{eq:DP-gen-div}, the states $\rho_{R_{n}B_{n}}$ and $\sigma_{R_{n}B_{n}}$ are defined in
\eqref{eq:state-rho-seq-ch-1} and \eqref{eq:state-sig-seq-ch-2}, respectively,
and\ the optimization is with respect to all quantum co-strategies or testers
that could be used to distinguish the quantum strategies $\mathcal{N}^{(n)}$
and~$\mathcal{M}^{(n)}$.
\end{definition}

Note that this quantity generalizes the quantum strategy distance and quantum
strategy fidelity of \cite{CDP08a,CDP09,G12,Gutoski2018fidelityofquantum}, as well as the strategy max-relative entropy of \cite{Chiribella_2016}, to
arbitrary divergences. Those quantities employ trace distance, fidelity, and max-relative entropy as the underlying divergences, respectively, but in what follows, we make extensive use of the generality afforded by Definition~\ref{def:gen-strat-div}.

\subsection{Physical transformations of quantum strategy boxes and data
processing}

\label{sec:phys-trans-strats}

Just as quantum channels model physical transformations of quantum states and
superchannels model physical transformations of quantum channels, we can also
consider physical transformations of quantum strategies. Given a quantum
strategy $\mathcal{N}^{(n)}$, we consider a general linear and completely
positive transformation $\Theta^{(n\rightarrow m)}$ of it, which takes as input
an $n$-round quantum strategy and outputs an $m$-round quantum strategy. A
fundamental result of \cite{CDP09} is that such a physical transformation
$\Theta^{(n\rightarrow m)}$ of a quantum strategy $\mathcal{N}^{(n)}$ is in
turn described by an $(n+m)$-round quantum strategy that interconnects with
$\mathcal{N}^{(n)}$ to generate an output, $m$-round quantum strategy.

Due to various choices of time ordering involved, there is not a unique way to
describe this physical transformation \cite{CDP09}, but here we adopt the choice that the
physical transformation $\Theta^{(n\rightarrow m)}$ first processes all
channels involved in the quantum strategy $\mathcal{N}^{(n)}$, and then it
generates the output $m$-round strategy $\Theta^{(n\rightarrow m)}%
(\mathcal{N}^{(n)})$. As such, the physical transformation $\Theta
^{(n\rightarrow m)}$ consists of $n+m$ channels $\mathcal{F}^{i}$ for
$i\in\{1,\ldots,n+m\}$, and the output quantum strategy $\mathcal{K}%
^{(m)}=\Theta^{(n\rightarrow m)}(\mathcal{N}^{(n)})$ then consists of the
following $m$ channels:
\begin{multline}
\mathcal{K}_{C_{1}\rightarrow M_{1}^{\prime}D_{1}}^{1}:=\\
\left(  \bigcirc_{i=1}^{n}\mathcal{F}_{R_{i}B_{i}\rightarrow R_{i+1}A_{i+1}%
}^{i+1}\circ\mathcal{N}_{M_{i-1}A_{i}\rightarrow M_{i}B_{i}}^{i}\right)
\circ\mathcal{F}_{C_{1}\rightarrow R_{1}A_{1}}^{1},
\end{multline}%
\begin{align}
\mathcal{K}_{M_{j-1}^{\prime}C_{j}\rightarrow M_{j}^{\prime}D_{j}}^{j}  &
:=\mathcal{F}_{R_{n+j}C_{j}\rightarrow R_{n+1+j}D_{j}}^{n+j},\\
\mathcal{K}_{M_{m-1}^{\prime}C_{m}\rightarrow D_{m}}^{m}  &  :=\mathcal{F}%
_{R_{n+m}C_{m}\rightarrow D_{m}}^{n+m},
\end{align}
for $j\in\left\{  2,\ldots,m-1\right\}  $, where we identify the memory
systems for the output strategy $\mathcal{K}^{(m)}$\ as $M_{k}^{\prime}\equiv
R_{n+k}$ for $k\in\left\{  1,\ldots,m\right\}  $.
Figure~\ref{fig:phys-trans-strat} depicts the transformation of a three-round
quantum strategy $\mathcal{N}^{(3)}$ to a three-round quantum strategy
$\mathcal{K}^{(3)}$ by a physical transformation $\Theta^{(3\rightarrow3)}$
consisting of the channels $\mathcal{F}^{1}$, \ldots, $\mathcal{F}^{6}$, along
with the pairing of the transformed strategy with a quantum co-strategy.

\begin{figure*}[ptb]
\begin{center}
\includegraphics[
width=\linewidth
]{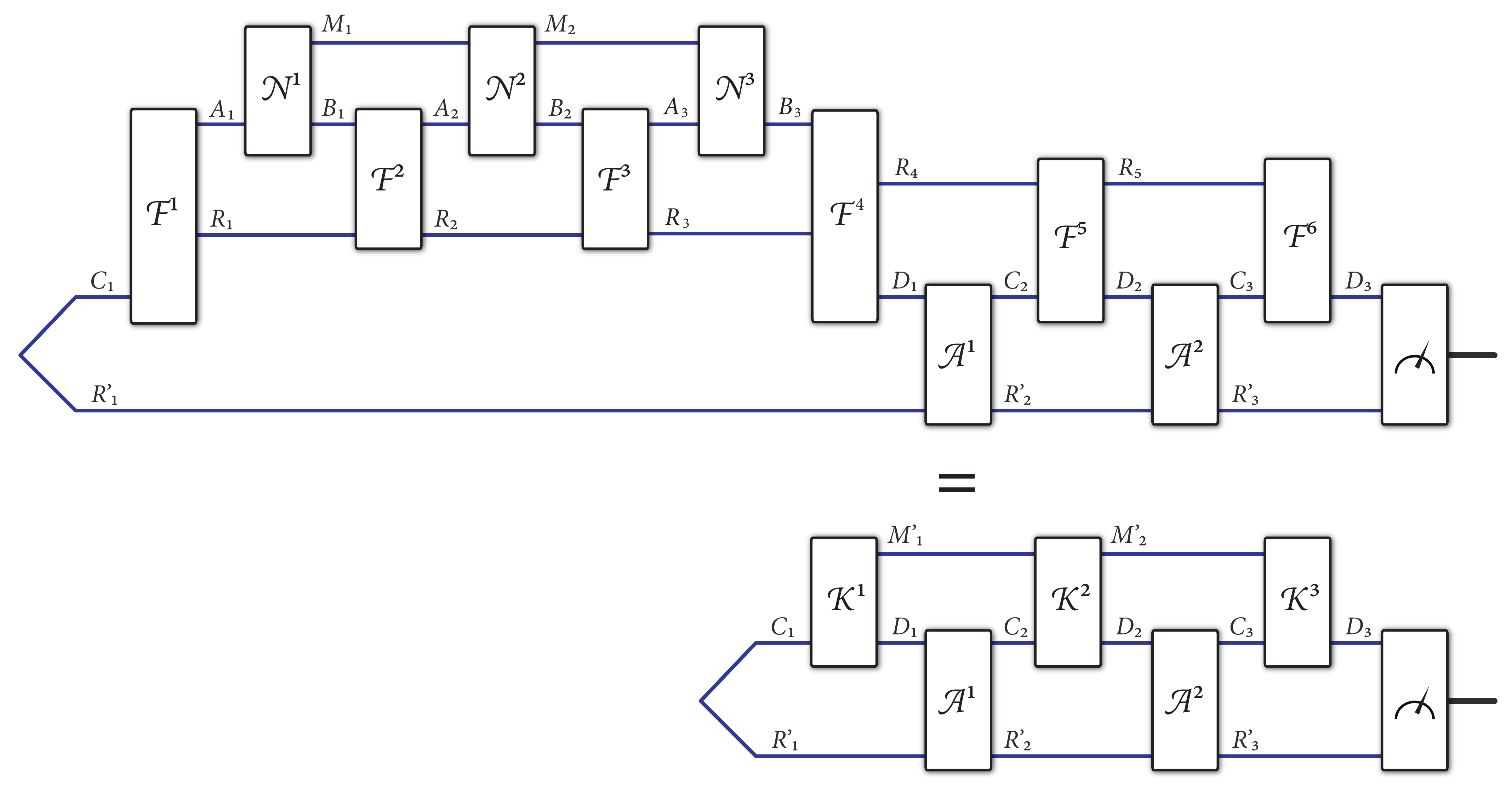}
\end{center}
\caption{A depiction of the transformation of a three-round quantum strategy
$\mathcal{N}^{(3)}$ to a three-round quantum strategy $\mathcal{K}^{(3)}$ by a
physical transformation $\Theta^{(3\rightarrow3)}$ consisting of the channels
$\mathcal{F}^{1}$, \ldots, $\mathcal{F}^{6}$, along with the pairing of the
transformed strategy with a quantum co-strategy.}%
\label{fig:phys-trans-strat}%
\end{figure*}

The following data processing inequality for the generalized strategy
divergence is a direct consequence of the definition and the fact that the
underlying generalized divergence $\mathbf{D}$ obeys data processing. This key
property allows for establishing bounds on the general strategy box
transformation problem. Also, it generalizes the data processing inequality for strategy distance and strategy fidelity from \cite{Gutoski2018fidelityofquantum}, but we require physical transformations in order to establish it.

\begin{theorem}
\label{thm:DP-strat-div}
Let $\mathcal{N}^{(n)}$ and $\mathcal{M}^{(n)}$ be
$n$-round quantum strategies, and let $\Theta^{(n\rightarrow m)}$ be a
physical transformation of them, of the form discussed above, that leads to
$m$-round quantum strategies $\Theta^{(n\rightarrow m)}(\mathcal{N}^{(n)})$
and $\Theta^{(n\rightarrow m)}(\mathcal{M}^{(n)})$. Then the following data
processing inequality holds for the generalized quantum strategy divergence:%
\begin{equation}
\mathbf{D}(\mathcal{N}^{(n)}\Vert\mathcal{M}^{(n)})\geq\mathbf{D}%
(\Theta^{(n\rightarrow m)}(\mathcal{N}^{(n)})\Vert\Theta^{(n\rightarrow
m)}(\mathcal{M}^{(n)})). \label{eq:DP-strategy-div}%
\end{equation}

\end{theorem}

\begin{proof}
The physical transformation $\Theta^{(n\rightarrow m)}$ consists of the
channels $\mathcal{F}^{i}$ for $i\in\{1,\ldots,n+m\}$. Set $\mathcal{K}%
^{(m)}=\Theta^{(n\rightarrow m)}(\mathcal{N}^{(n)})$ and $\mathcal{L}%
^{(m)}=\Theta^{(n\rightarrow m)}(\mathcal{M}^{(n)})$. Also, let us consider a
quantum co-strategy for $\mathcal{K}^{(m)}$ and $\mathcal{L}^{(m)}$, which
consists of a state $\rho_{R_{1}^{\prime}C_{1}}$ and a set $\{\mathcal{A}%
_{R_{i}^{\prime}D_{i}\rightarrow R_{i+1}^{\prime}C_{i+1}}^{i}\}_{i=1}^{m-1}$
of channels:%
\begin{equation}
\mathcal{T}:=\left\{  \rho_{R^{\prime}C_{1}},\{\mathcal{A}_{R_{i}^{\prime
}D_{i}\rightarrow R_{i+1}^{\prime}C_{i+1}}^{i}\}_{i=1}^{m-1}\right\}  .
\end{equation}

Suppose first that the physical transformation $\Theta^{(n\rightarrow m)}$
acts on the quantum strategy $\mathcal{N}^{(n)}$. In this case, the first
channel $\mathcal{F}_{C_{1}\rightarrow R_{1}A_{1}}^{1}$\ acts on the state
$\rho_{R_{1}^{\prime}C_{1}}$ and outputs systems $A_{1}$ and $R_{1}$. Then the
channel $\mathcal{N}^1_{A_{1}\rightarrow M_1 B_{1}}$ is applied, and the second
channel $\mathcal{F}_{R_{1}B_{1}\rightarrow R_{2}A_{2}}^{2}$ is applied. This
repeats $n-1$ more times, and the resulting state is as follows:%
\begin{multline}
\omega_{R_{1}^{\prime}R_{n}B_{n}}:=\mathcal{N}_{M_{n-1}A_{n}\rightarrow B_{n}%
}^{n}\\
\circ\left(  \bigcirc_{i=1}^{n-1}\mathcal{F}_{R_{i}B_{i}\rightarrow
R_{i+1}A_{i+1}}^{i+1}\circ\mathcal{N}_{M_{i-1}A_{i}\rightarrow M_{i}B_{i}}%
^{i}\right) \\
\circ\mathcal{F}_{C_{1}\rightarrow R_{1}A_{1}}^{1}(\rho_{R_{1}^{\prime}C_{1}%
}),
\end{multline}%
\begin{equation}
\omega_{R_{1}^{\prime}R_{n+1}D_{1}}:=\mathcal{F}_{R_{n}B_{n}\rightarrow
R_{n+1}D_{1}}^{n+1}(\omega_{R_{1}^{\prime}R_{n}B_{n}}).
\end{equation}

At this point, the other elements of the co-strategy and the remainder of the
transformation $\Theta^{(n\rightarrow m)}$\ are applied, which consists of the
co-strategy channels $\{\mathcal{A}_{R_{i}^{\prime}D_{i}\rightarrow R_{i+1}%
^{\prime}C_{i+1}}^{i}\}_{i=1}^{m-1}$ interleaved by the transformation
channels $\mathcal{F}^{n+2}$, \ldots, $\mathcal{F}^{m}$. The resulting state
is then%
\begin{equation}
\omega_{R_{m}^{\prime}D_{m}}:=\mathcal{P}_{R_{1}^{\prime}L_{1}D_{1}\rightarrow
R_{m}^{\prime}D_{m}}(\omega_{R_{1}^{\prime}R_{n+1}D_{1}}),
\end{equation}
where%
\begin{multline}
\mathcal{P}_{R_{1}^{\prime}R_{n+1}D_{1}\rightarrow R_{m}^{\prime}D_{m}}:=\\
\mathcal{F}_{R_{n+m-1}C_{m}\rightarrow D_{m}}^{n+m}\circ\mathcal{A}%
_{R_{m-1}^{\prime}D_{m-1}\rightarrow R_{m}^{\prime}C_{m}}^{m-1}\circ\\
\bigcirc_{i=1}^{m-2}\mathcal{F}_{R_{n+i}C_{i+1}\rightarrow R_{n+1+i}D_{i+1}%
}^{n+1+i}\circ\mathcal{A}_{R_{i}^{\prime}D_{i}\rightarrow R_{i+1}^{\prime
}C_{i+1}}^{i}.
\end{multline}

We also define the following states for the quantum strategy $\mathcal{M}%
^{(n)}$:%
\begin{multline}
\xi_{R_{1}^{\prime}R_{n}B_{n}}:=\mathcal{M}_{M_{n-1}A_{n}\rightarrow B_{n}%
}^{n}\\
\circ\left(  \bigcirc_{i=1}^{n-1}\mathcal{F}_{R_{i}B_{i}\rightarrow
R_{i+1}A_{i+1}}^{i+1}\circ\mathcal{M}_{M_{i-1}A_{i}\rightarrow M_{i}B_{i}}%
^{i}\right) \\
\circ\mathcal{F}_{C_{1}\rightarrow R_{1}A_{1}}^{1}(\rho_{R_{1}^{\prime}C_{1}%
}),
\end{multline}%
\begin{align}
\xi_{R_{1}^{\prime}R_{n+1}D_{1}}  &  :=\mathcal{F}_{R_{n}B_{n}\rightarrow
R_{n+1}D_{1}}^{n+1}(\xi_{R_{1}^{\prime}R_{n}B_{n}}),\\
\xi_{R_{m}^{\prime}D_{m}}  &  :=\mathcal{P}_{R_{1}^{\prime}L_{1}%
D_{1}\rightarrow R_{m}^{\prime}D_{m}}(\xi_{R_{1}^{\prime}R_{n+1}D_{1}}).
\end{align}

Then consider that%
\begin{align}
\mathbf{D}(\mathcal{N}^{(n)}\Vert\mathcal{M}^{(n)})  &  \geq\mathbf{D}%
(\omega_{R_{1}^{\prime}R_{n}B_{n}}\Vert\xi_{R_{1}^{\prime}R_{n}B_{n}})\\
&  \geq\mathbf{D}(\omega_{R_{m}^{\prime}D_{m}}\Vert\xi_{R_{m}^{\prime}D_{m}}).
\end{align}
The first inequality follows because the state $\rho_{R_{1}^{\prime}C_{1}}$
and the channels $\mathcal{F}^{i}$ for $i\in\{1,\ldots,n\}$ constitute a
particular co-strategy for discriminating $\mathcal{N}^{(n)}$ from
$\mathcal{M}^{(n)}$. The next inequality is a consequence of quantum data
processing for the underlying generalized divergence, given that
$\mathcal{P}_{R_{1}^{\prime}L_{1}D_{1}\rightarrow R_{m}^{\prime}D_{m}}$ is a
quantum channel. Since the inequality holds for all possible co-strategies
$\mathcal{T}$ that could be used to distinguish $\mathcal{K}^{(m)}$ from
$\mathcal{L}^{(m)}$, we conclude \eqref{eq:DP-strategy-div}.
\end{proof}

\begin{remark}
We note that the data processing inequality in \eqref{eq:DP-strategy-div} holds more generally for physical transformations of quantum strategy boxes that do not necessarily proceed in the order that we have fixed (i.e., it holds for other time orderings of physical transformations of strategy boxes). The main idea for establishing it is to use the data processing inequality for the underlying generalized divergence and that a co-strategy for a physically transformed strategy is a special kind of co-strategy for the original strategy.
\end{remark}

\subsection{Quantum strategy box transformation problem}

The goal of this setting is to convert the quantum strategy box $(\mathcal{N}%
^{(n)},\mathcal{M}^{(n)})$ to the strategy box $(\mathcal{K}^{(m)}%
,\mathcal{L}^{(m)})$ by means of common physical transformation $\Theta
^{(n\rightarrow m)}$, subject to the constraint that $\mathcal{K}^{(m)}$ is
realized approximately from $\mathcal{N}^{(n)}$, i.e.,%
\begin{equation}
\Theta^{(n\rightarrow m)}(\mathcal{N}^{(n)})\approx_{\varepsilon}%
\mathcal{K}^{(m)}, \label{eq:approx-notion-sequential}%
\end{equation}
while $\mathcal{L}^{(m)}$ is realized perfectly from $\mathcal{M}^{(n)}$ by
the protocol $\Theta^{(n\rightarrow m)}$, i.e.,%
\begin{equation}
\Theta^{(n\rightarrow m)}(\mathcal{M}^{(n)})=\mathcal{L}^{(m)},
\label{eq:second-exact-conversion-sequential}%
\end{equation}
just as is the case with all of the other transformations that we have
considered in the resource theory of asymmetric distinguishability. The common
physical transformation $\Theta^{(n\rightarrow m)}$ that we consider is as we
discussed in Section~\ref{sec:phys-trans-strats} and is depicted in Figures~\ref{fig:ch1-seque} and
\ref{fig:ch2-seque-exact}. It consists of a general physical processing of the
strategy box $(\mathcal{N}^{(n)},\mathcal{M}^{(n)})$ to convert it
approximately to the strategy box $(\mathcal{K}^{(m)},\mathcal{L}%
^{(m)})$, in the sense given in \eqref{eq:approx-notion-sequential}--\eqref{eq:second-exact-conversion-sequential}.

The notion of approximation that we employ in
\eqref{eq:approx-notion-sequential} is the normalized strategy distance of
\cite{CDP08a,CDP09,G12}, which generalizes the normalized diamond distance to
the  setting of interest here. This quantity is a special case of
the generalized strategy divergence from Definition~\ref{def:gen-strat-div}, with the underlying
divergence set to be the normalized trace distance $\frac{1}{2}\left\Vert
\cdot\right\Vert _{1}$. The motivation for employing the normalized strategy
distance is the same as that which we gave for normalized diamond
distance:\ it quantifies the worst-case statistical error (absolute deviation)
that one could make when trying to distinguish the simulation $\Theta
^{(n\rightarrow m)}(\mathcal{N}^{(n)})$ from the ideal output strategy $\mathcal{K}^{(m)}$ by
any quantum-physical experiment.

\begin{figure*}[ptb]
\begin{center}
\includegraphics[
width=\linewidth
]{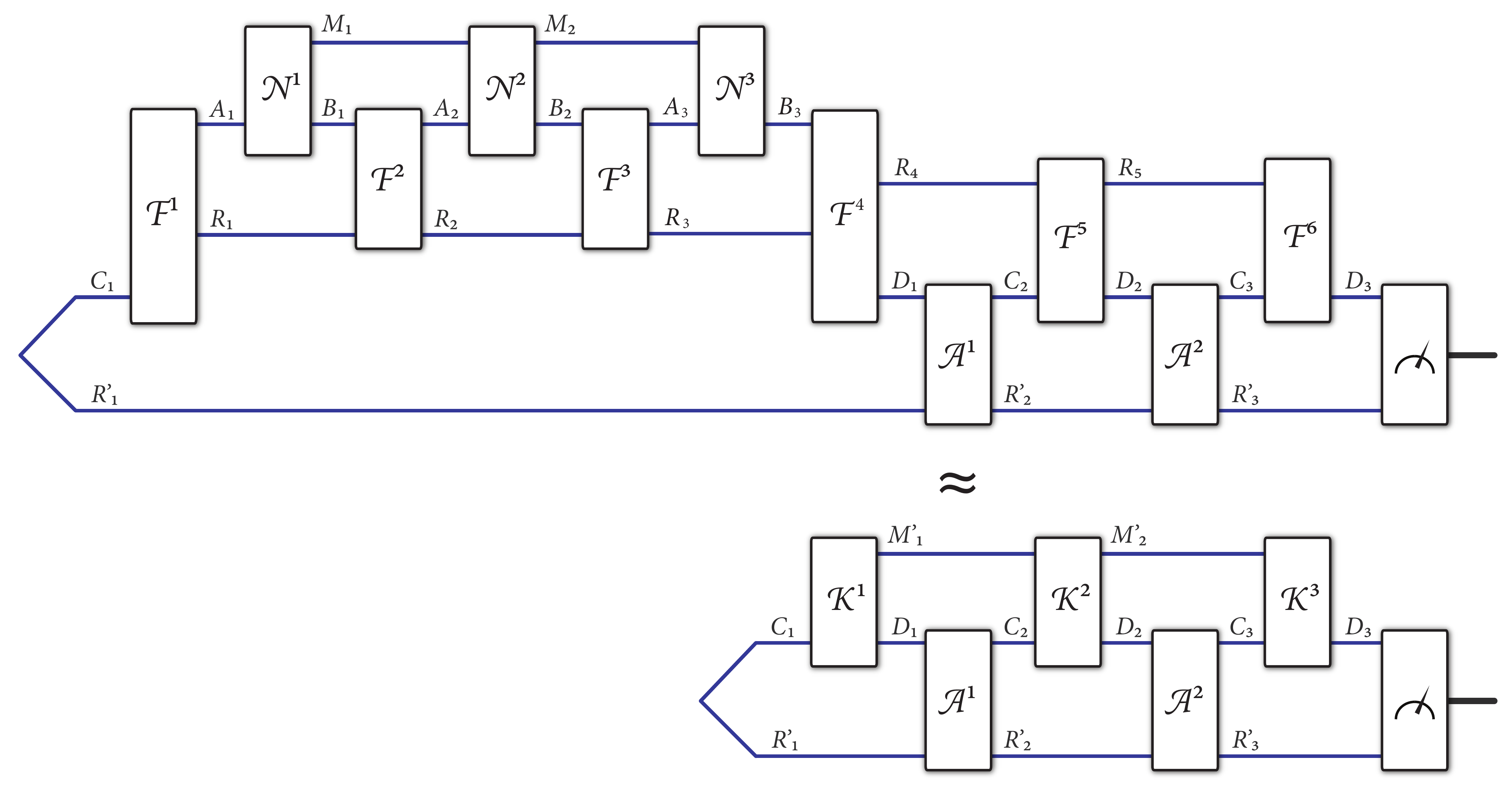}
\end{center}
\caption{Depiction of the physical transformation $\Theta^{(n\rightarrow m)}$
that converts the three-round strategy $\mathcal{N}^{(3)}$ to the three-round strategy $\mathcal{K}^{(3)}$. The physical
transformation $\Theta^{(n\rightarrow m)}$\ consists of the channels
$\mathcal{F}^{1}$, \ldots, $\mathcal{F}^{6}$. A discriminator could in
principle then perform a co-strategy to distinguish the simulation in the top
part of the figure from the ideal implementation of the strategy
$\mathcal{K}^{(3)}$ in the bottom part of the figure. We demand
that the absolute deviation in probability between any measurement outcome in
the top part be no larger than $\varepsilon$ when compared to the same from
the bottom part, i.e., that the normalized strategy distance be no larger than
$\varepsilon$.}%
\label{fig:ch1-seque}%
\end{figure*}

\begin{figure*}[ptb]
\begin{center}
\includegraphics[
width=\linewidth
]
{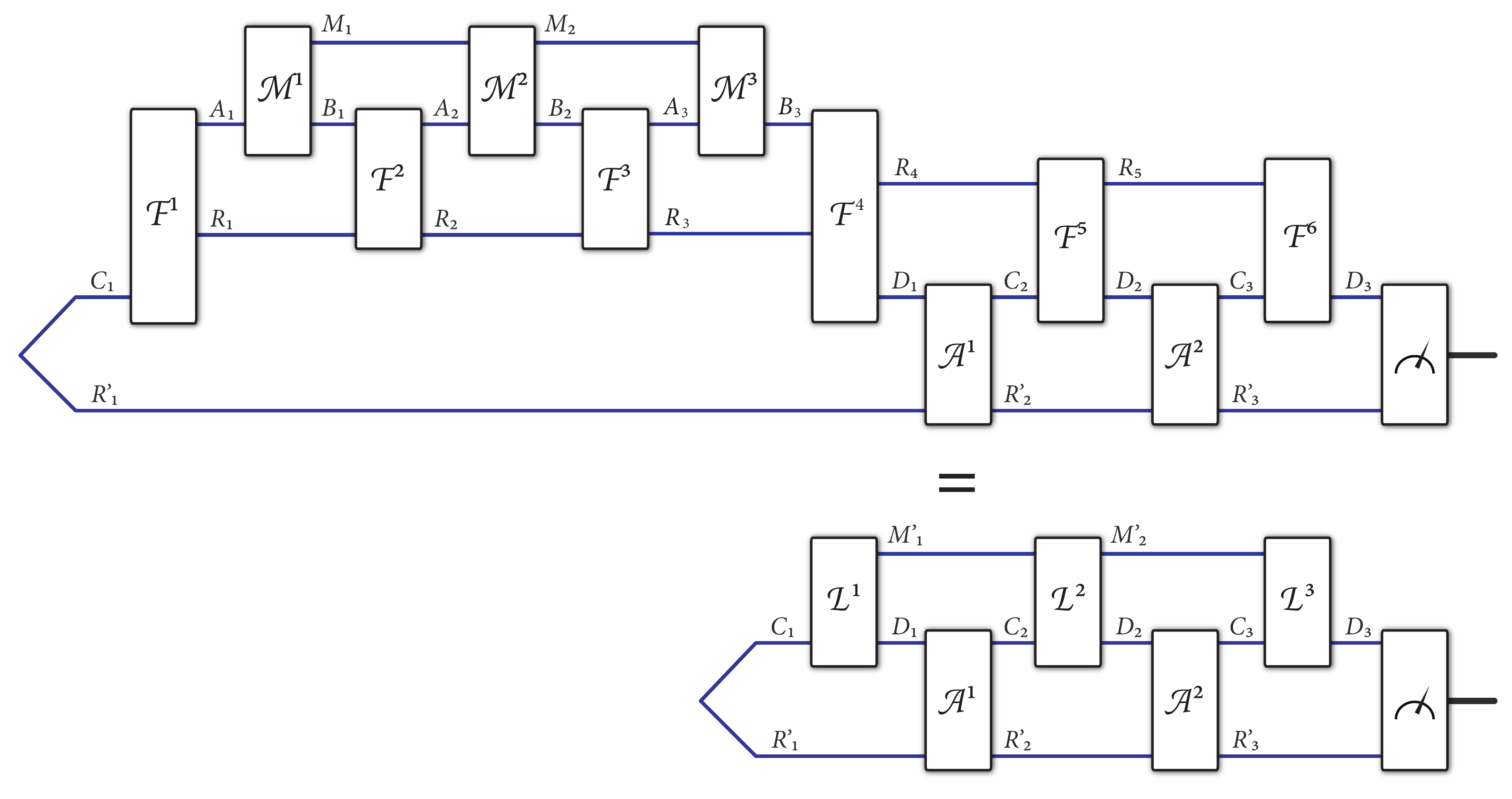}
\end{center}
\caption{Depiction of the physical transformation $\Theta^{(n\rightarrow m)}$
that converts the three-round strategy $\mathcal{M}^{(3)}$ to the three-round strategy $\mathcal{L}^{(3)}$. The physical
transformation $\Theta^{(n\rightarrow m)}$\ is the same as that given in
Figure~\ref{fig:ch1-seque} and consists of the channels $\mathcal{F}^{1}$,
\ldots, $\mathcal{F}^{6}$. A discriminator could in principle then perform a
strategy to distinguish the simulation in the top part of the figure from the
ideal implementation of the three-round strategy $\mathcal{L}^{(3)}$ in the bottom part of the figure. We demand that the absolute deviation in
probability between any measurement outcome in the top part be exactly equal to zero
when compared to the same from the bottom part, i.e., that the normalized
strategy distance be equal to zero, so that the simulation is perfect in this
case.}%
\label{fig:ch2-seque-exact}%
\end{figure*}

We now describe the above in more detail. The general physical transformation
$\Theta^{(n\rightarrow m)}$\ of the first strategy box $(\mathcal{N}%
^{(n)},\mathcal{M}^{(n)})$\ consists of $n+m$ channels, denoted by
$\mathcal{F}^{i}$ for $i\in\left\{  1,\ldots,n+m\right\}  $. To assess the
performance of the transformation%
\begin{equation}
(\mathcal{N}^{(n)},\mathcal{M}^{(n)})\rightarrow(\Theta^{(n\rightarrow
m)}(\mathcal{N}^{(n)}),\Theta^{(n\rightarrow m)}(\mathcal{M}^{(n)}))
\end{equation}
in simulating the strategy $\mathcal{K}^{(m)}$, the resulting quantum strategy
$\Theta^{(n\rightarrow m)}(\mathcal{N}^{(n)})$ is paired up with a quantum
co-strategy $\mathcal{T}$ \cite{GW07} (or tester \cite{CDP08a,CDP09}), which
consists of a state $\rho_{R^{\prime}C_{1}}$, a set $\{\mathcal{A}%
_{R_{i}^{\prime}D_{i}\rightarrow R_{i+1}^{\prime}C_{i+1}}^{i}\}_{i=1}^{m-1}$
of channels, and a final measurement $\{Q_{R_{m}^{\prime}D_{m}},I_{R_{m}%
^{\prime}D_{m}}-Q_{R_{m}^{\prime}D_{m}}\}$:%
\begin{equation}
\mathcal{T}:=\left\{  \rho_{R^{\prime}C_{1}},\{\mathcal{A}_{R_{i}^{\prime
}D_{i}\rightarrow R_{i+1}^{\prime}C_{i+1}}^{i}\}_{i=1}^{m-1},Q_{R_{m}^{\prime
}D_{m}}\right\}  . \label{eq:q-tester}%
\end{equation}

Suppose first that the transformation $\Theta^{(n\rightarrow m)}$ acts on the
quantum strategy $\mathcal{N}^{(n)}$. In this case, the first channel
$\mathcal{F}_{C_{1}\rightarrow R_{1}A_{1}}^{1}$\ acts on the state
$\rho_{R_{1}^{\prime}C_{1}}$ and outputs systems $A_{1}$ and $R_{1}$. Then the
channel $\mathcal{N}_{A_{1}\rightarrow B_{1}}$ is applied, and the second
channel $\mathcal{F}_{R_{1}B_{1}\rightarrow R_{2}A_{2}}^{2}$ is applied. This
repeats $n-1$ more times, and the resulting state is as follows:%
\begin{multline}
\omega_{R_{1}^{\prime}R_{n+1}D_{1}}:=\mathcal{F}_{R_{n}B_{n}\rightarrow
R_{n+1}D_{1}}^{n+1}\circ\mathcal{N}_{M_{n-1}A_{n}\rightarrow B_{n}}\\
\circ\left(  \bigcirc_{i=1}^{n-1}\mathcal{F}_{R_{i}B_{i}\rightarrow
R_{i+1}A_{i+1}}^{i+1}\circ\mathcal{N}_{M_{i-1}A_{i}\rightarrow M_{i}B_{i}%
}\right) \\
\circ\mathcal{F}_{C_{1}\rightarrow R_{1}A_{1}}^{1}(\rho_{R_{1}^{\prime}C_{1}%
}).
\end{multline}
At this point, the other elements of the co-strategy and the remainder of the
simulation are applied, which consists of the testing channels $\{\mathcal{A}%
_{R_{i}^{\prime}D_{i}\rightarrow R_{i+1}^{\prime}C_{i+1}}^{i}\}_{i=1}^{m-1}$
interleaved by the simulation channels $\mathcal{F}^{n+2}$, \ldots,
$\mathcal{F}^{m}$. The resulting state is then%
\begin{equation}
\omega_{R_{m}^{\prime}D_{m}}:=\mathcal{P}_{R_{1}^{\prime}L_{1}D_{1}\rightarrow
R_{m}^{\prime}D_{m}}(\omega_{R_{1}^{\prime}R_{n+1}D_{1}}),
\end{equation}
where%
\begin{multline}
\mathcal{P}_{R_{1}^{\prime}R_{n+1}D_{1}\rightarrow R_{m}^{\prime}D_{m}}:=\\
\mathcal{F}_{R_{n+m-1}C_{m}\rightarrow D_{m}}^{n+m}\circ\mathcal{A}%
_{R_{m-1}^{\prime}D_{m-1}\rightarrow R_{m}^{\prime}C_{m}}^{m-1}\circ\\
\bigcirc_{i=1}^{m-2}\mathcal{F}_{R_{n+i}C_{i+1}\rightarrow R_{n+1+i}D_{i+1}%
}^{n+1+i}\circ\mathcal{A}_{R_{i}^{\prime}D_{i}\rightarrow R_{i+1}^{\prime
}C_{i+1}}^{i}%
\end{multline}

The final state above is then compared with the following state, which results
from the application of the quantum co-strategy $\mathcal{T}$ to the ideal strategy
$\mathcal{K}^{(m)}$:%
\begin{multline}
\tau_{R_{m}^{\prime}D_{m}}:=\mathcal{K}_{M_{m-1}^{\prime}C_{m}\rightarrow
D_{m}}\circ\\
(\bigcirc_{i=1}^{m-1}\mathcal{A}_{R_{i}^{\prime}D_{i}\rightarrow
R_{i+1}^{\prime}C_{i+1}}^{i}\circ\mathcal{K}_{M_{i-1}^{\prime}C_{i}\rightarrow
M_{i}^{\prime}D_{i}})(\rho_{R_{1}^{\prime}C_{1}})
\end{multline}
See Figure~\ref{fig:ch1-seque} for a depiction of these two scenarios.

The simulation has $\varepsilon$ error if the following inequality holds%
\begin{equation}
\sup_{\rho,\{\mathcal{A}^{i}\}_{i=1}^{m-1},Q}\left\vert \operatorname{Tr}%
[Q_{R_{m}^{\prime}D_{m}}\left(  \omega_{R_{m}^{\prime}D_{m}}-\tau
_{R_{m}^{\prime}D_{m}}\right)  ]\right\vert \leq\varepsilon,
\label{eq:sequential-condition-quality}%
\end{equation}
where the optimization is with respect to all quantum co-strategies $\mathcal{T}$ as
defined in \eqref{eq:q-tester}. The expression on the left-hand side above is
in fact equal to the $m$-round normalized quantum strategy distance considered
in \cite{CDP08a,CDP09,G12,Gutoski2018fidelityofquantum}, so that we can write
\eqref{eq:sequential-condition-quality}\ equivalently as%
\begin{equation}
\frac{1}{2}\left\Vert \Theta^{(n\rightarrow m)}(\mathcal{N}^{(n)}%
)-\mathcal{K}^{(m)}\right\Vert _{\diamond m}\leq\varepsilon.
\label{eq:strategy-dist-approx}%
\end{equation}
As a shorthand for the inequality in \eqref{eq:strategy-dist-approx}, we
employ the notation%
\begin{equation}
\Theta^{(n\rightarrow m)}(\mathcal{N}^{(n)})\approx_{\varepsilon}%
\mathcal{K}^{(m)}. \label{eq:strategy-dist-approx-shorthand}%
\end{equation}

It is also demanded that the transformation $\Theta^{(n\rightarrow m)}$\ be
such that $\Theta^{(n\rightarrow m)}(\mathcal{M}^{(n)})=\mathcal{L}^{(m)}$,
which is the same \cite{G12} as demanding that%
\begin{equation}
\frac{1}{2}\left\Vert \Theta^{(n\rightarrow m)}(\mathcal{M}^{(n)}%
)-\mathcal{L}^{(m)}\right\Vert _{\diamond m}=0.
\end{equation}
This is consistent with our prior error criteria in the simpler scenarios for
the resource theory of asymmetric distinguishability.

Thus, the general strategy box transformation problem can be phrased as the
following optimization problem, which is a function of $n,m\in\mathbb{Z}^{+}$
and channels $\mathcal{N}$, $\mathcal{M}$, $\mathcal{K}$, and $\mathcal{L}$:%
\begin{equation}
\inf_{\Theta^{(n\rightarrow m)}}\left\{
\begin{array}
[c]{c}%
\varepsilon\in\left[  0,1\right]  :\Theta^{(n\rightarrow m)}(\mathcal{N}%
^{(n)})\approx_{\varepsilon}\mathcal{K}^{(m)},\\
\Theta^{(n\rightarrow m)}(\mathcal{M}^{(n)})=\mathcal{L}^{(m)}%
\end{array}
\right\}  , \label{eq:gen-seq-ch-box-trans-prob}%
\end{equation}
where the infimum is with respect to physical transformations $\Theta
^{(n\rightarrow m)}$.

We assert here that the optimization problem in
\eqref{eq:gen-seq-ch-box-trans-prob} can be cast as a semi-definite program,
by employing the facts that the quantum strategy distance can be calculated by
a semi-definite program\ and one can write down Choi operators for
$\Theta^{(n\rightarrow m)}$, $\mathcal{N}^{(n)}$, $\mathcal{M}^{(n)}$,
$\mathcal{K}^{(m)}$, and $\mathcal{L}^{(m)}$ \cite{CDP08a,CDP09,G12} along
with various non-signaling constraints to denote the time-orderings involved.
However, we do not elaborate on the details here.

In Appendix~\ref{app:q-strat-trans}, Proposition~\ref{prop:n-to-m-seq-ch-box-trans-conv-bnd} states converse bounds that apply to arbitrary
protocols that transform the $n$-round strategy box $(\mathcal{N}%
^{(n)},\mathcal{M}^{(n)})$ to the $m$-round strategy box $(\mathcal{K}%
^{(m)},\mathcal{L}^{(m)})$ while satisfying $\Theta^{(n\rightarrow
m)}(\mathcal{N}^{(n)})\approx_{\varepsilon}\mathcal{K}^{(m)}$ and
$\Theta^{(n\rightarrow m)}(\mathcal{M}^{(n)})=\mathcal{L}^{(m)}$. The bounds are expressed in terms of strategy R\'enyi divergences, which are defined as special cases of Definition~\ref{def:gen-strat-div} with the underlying divergence fixed to be the 
R\'enyi divergences.

\subsection{Asymptotic setting for sequential channel box transformations}

It does not seem sensible to consider an asymptotic version of the general
strategy box transformation problem, as in general there is no regular
structure associated with arbitrary strategy boxes. However, if we impose some
structure, then it is sensible to do so.

The simplest structure that we can impose is that each strategy box is
actually a sequential channel box, involving sequential uses of the same
quantum channels. Then we can phrase the sequential channel box transformation
problem in an asymptotic, Shannon-theoretic way, similar to how we did for the
parallel channel box transformation problem in
Section~\ref{sec:gen-box-trans-parallel}.

Let $n,m\in\mathbb{Z}^{+}$ and $\varepsilon\in\left[  0,1\right]  $. An
$(n,m,\varepsilon)$ sequential channel box transformation protocol for the
channel boxes $(\mathcal{N},\mathcal{M})$ and $(\mathcal{K},\mathcal{L})$
consists of a physical transformation $\Theta^{(n\rightarrow m)}$, as
described in Section~\ref{sec:phys-trans-strats}, such that 
\begin{align}
\Theta^{(n\rightarrow m)}(\mathcal{N}%
^{(n)})
& \approx_{\varepsilon}\mathcal{K}^{(m)} , \\\Theta^{(n\rightarrow
m)}(\mathcal{M}^{(n)}) & =\mathcal{L}^{(m)},
\end{align}
where $\mathcal{N}^{(n)}$,
$\mathcal{M}^{(n)}$, $\mathcal{K}^{(m)}$, and $\mathcal{L}^{(m)}$ are the
sequential channels corresponding to the channels $\mathcal{N}$, $\mathcal{M}%
$, $\mathcal{K}$, and $\mathcal{L}$, respectively. For clarity, Figure~\ref{fig:seq-ch-box-trans} depicts an example of a sequential channel box transformation protocol.
\begin{figure*}[ptb]
\begin{center}
\includegraphics[
width=\linewidth
]
{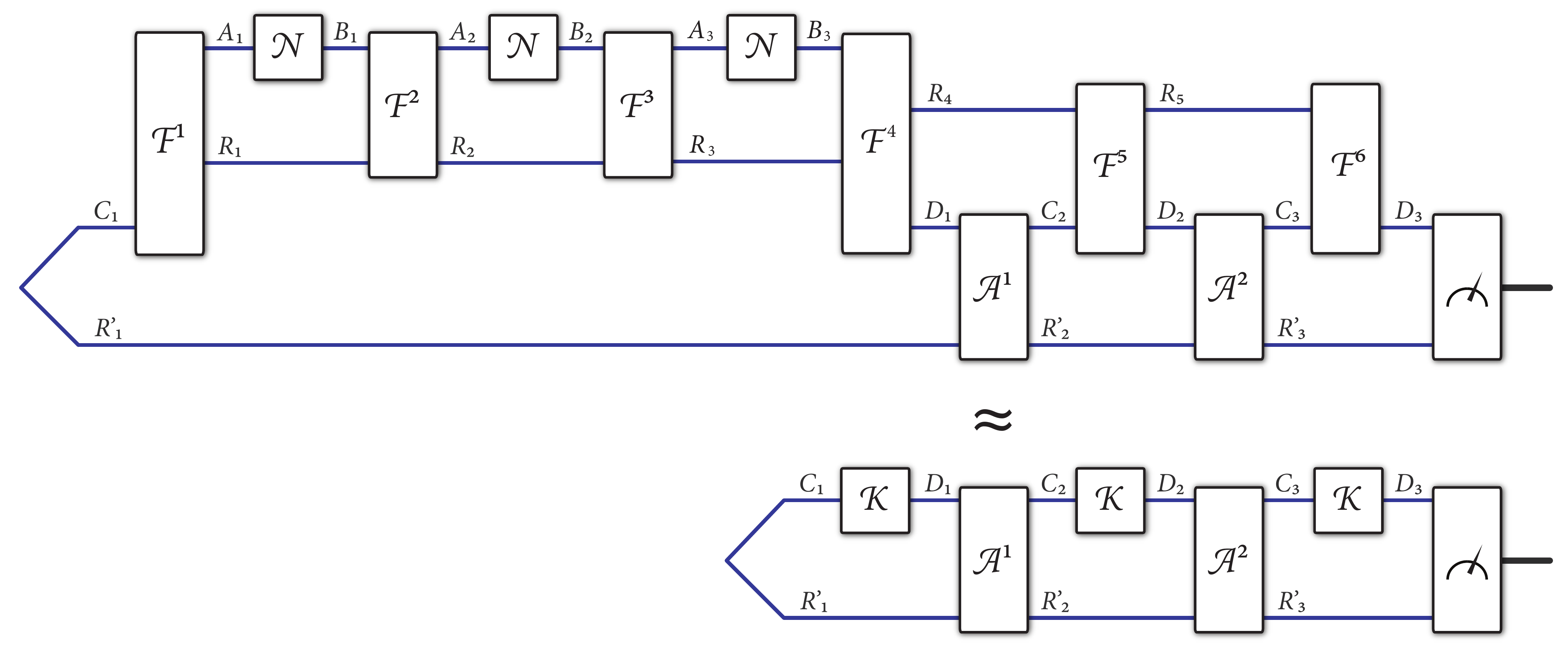}
\includegraphics[
width=\linewidth
]
{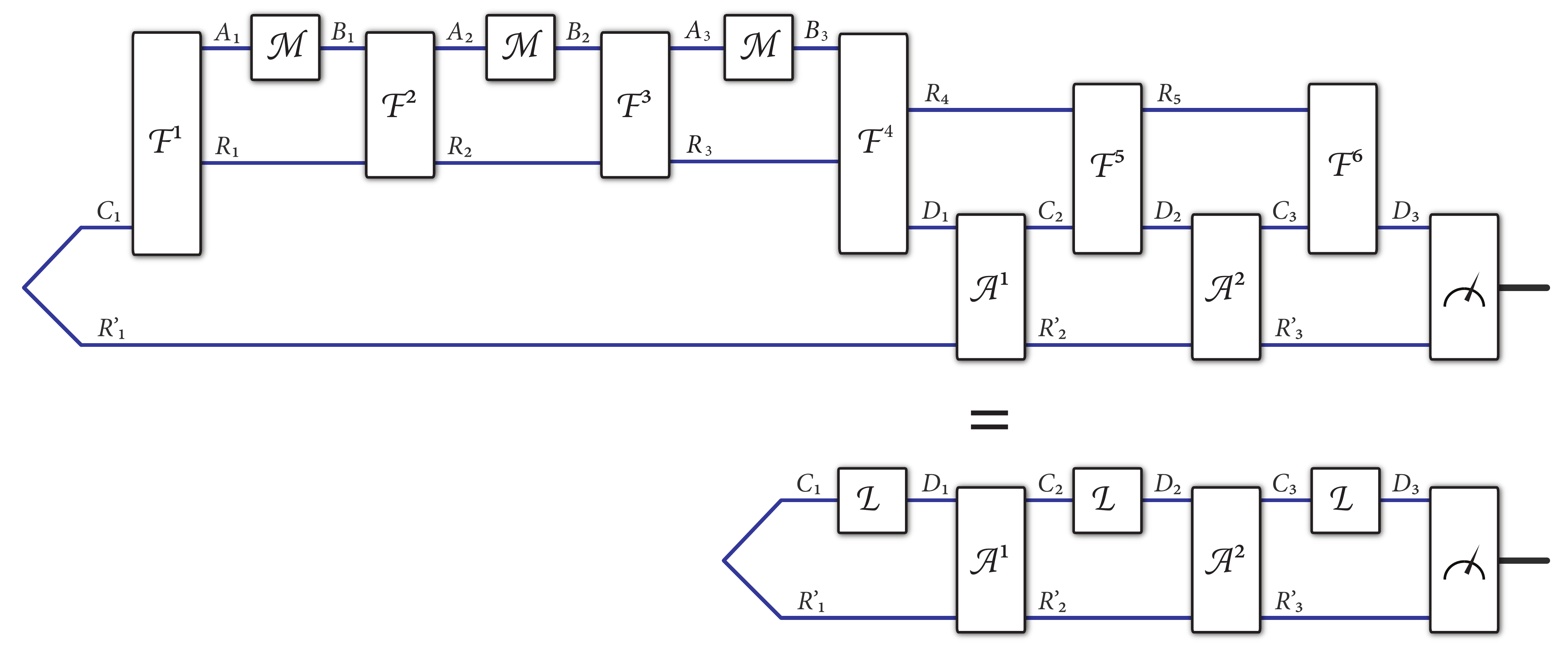}
\end{center}
\caption{Depiction of a sequential channel box transformation protocol. Three sequential uses of the channel $\mathcal{N}$ are converted approximately to three sequential uses of the channel $\mathcal{K}$, while three sequential uses of the channel $\mathcal{M}$ are converted exactly to three sequential uses of the channel $\mathcal{L}$. This is a special case of a strategy box transformation protocol, as depicted in Figures~\ref{fig:ch1-seque} and \ref{fig:ch2-seque-exact}.}%
\label{fig:seq-ch-box-trans}%
\end{figure*}

A rate $R$ is achievable if for all $\varepsilon\in(0,1]$, $\delta>0$, and
sufficiently large $n$, there exists an $\left(  n,n\left[  R-\delta\right]
,\varepsilon\right)  $ sequential channel box transformation protocol. The
optimal sequential channel box transformation rate $R((\mathcal{N}%
,\mathcal{M})\rightarrow(\mathcal{K},\mathcal{L}))$ is equal to the supremum
of all achievable rates.

On the other hand, a rate $R$ is a strong converse rate if for all
$\varepsilon\in\lbrack0,1)$, $\delta>0$, and sufficiently large $n$, there
does not exist an $(n,n\left[  R+\delta\right]  ,\varepsilon)$ sequential
channel box transformation protocol. The strong converse sequential channel
box transformation rate $\widetilde{R}((\mathcal{N},\mathcal{M})\rightarrow
(\mathcal{K},\mathcal{L}))$ is equal to the infimum of all strong converse rates.

The following inequality is a direct consequence of definitions:%
\begin{equation}
R((\mathcal{N},\mathcal{M})\rightarrow(\mathcal{K},\mathcal{L}))\leq
\widetilde{R}((\mathcal{N},\mathcal{M})\rightarrow(\mathcal{K},\mathcal{L})).
\label{eq:sequential-rates-gen-trans}%
\end{equation}

Although it is a challenging question in general to determine the optimal
rates in \eqref{eq:sequential-rates-gen-trans} for arbitrary channel boxes,
there are some special cases for which it is possible to determine them.

\begin{enumerate}
\item If the channel boxes $(\mathcal{N},\mathcal{M})$\ and $(\mathcal{K}%
,\mathcal{L})$ are environment-seizable, then our prior results from
\cite{WW19states} and Corollary~\ref{cor:seq-ch-box-conv} from Appendix~\ref{app:bnds-strat-trans} imply that%
\begin{align}
&  R((\mathcal{N},\mathcal{M})\rightarrow(\mathcal{K},\mathcal{L}))\notag \\
&  =\widetilde{R}((\mathcal{N},\mathcal{M})\rightarrow(\mathcal{K}%
,\mathcal{L}))\label{eq:opt-ach-env-seiz-seq-1}\\
&  =\frac{D(\mathcal{N}\Vert\mathcal{M})}{D(\mathcal{K}\Vert\mathcal{L})}.
\label{eq:opt-ach-env-seiz-seq-2}%
\end{align}
The main reason that this simplification occurs is that the channels involved
for environment-seizable pairs are equivalent to states, so that the prior
achievability results for states \cite{WW19states} apply. Also, the converse
bounds from Appendix~\ref{app:bnds-strat-trans} simplify for the same reason.

\item If the channel boxes $(\mathcal{N},\mathcal{M})$\ and $(\mathcal{K}%
,\mathcal{L})$ are classical--quantum, then the following strong converse
bound holds%
\begin{equation}
\widetilde{R}((\mathcal{N},\mathcal{M})\rightarrow(\mathcal{K},\mathcal{L}%
))\leq\frac{D(\mathcal{N}\Vert\mathcal{M})}{D(\mathcal{K}\Vert\mathcal{L})},
\label{eq:str-conv-cq-seq}%
\end{equation}
as a consequence of  \cite[Lemma~26]{BHKW18} and the discussion in Appendix~\ref{app:seq-ch-trans-amortized-div}. It is
reasonable to conjecture that this bound is saturated---what remains is to
show that $D(\mathcal{N}\Vert\mathcal{M})$ is the optimal rate of
distinguishability dilution for classical--quantum channels.

\item If the channel box $(\mathcal{N},\mathcal{M})$ is classical--quantum and
$(\mathcal{K},\mathcal{L})$ is environment seizable, then the equalities in
\eqref{eq:opt-ach-env-seiz-seq-1}--\eqref{eq:opt-ach-env-seiz-seq-2} hold.
This is a consequence of the upper bound in \eqref{eq:str-conv-cq-seq} holding
in this case, while the lower bound $R((\mathcal{N},\mathcal{M})\rightarrow
(\mathcal{K},\mathcal{L}))\geq\frac{D(\mathcal{N}\Vert\mathcal{M}%
)}{D(\mathcal{K}\Vert\mathcal{L})}$ follows because one can first distill bits
of asymmetric distinguishability from $(\mathcal{N},\mathcal{M})$ at the rate
$D(\mathcal{N}\Vert\mathcal{M})$ and then dilute them to $(\mathcal{K}%
,\mathcal{L})$, in a sequential simulation, with the latter simulation being
possible easily by preparing the environment states for $(\mathcal{K}%
,\mathcal{L})$ and then acting with the relevant common channels on demand
when needed.
\end{enumerate}

\section{Distillation and dilution of quantum strategy and sequential channel
boxes}

\label{sec:dist-dil-strat-seq-boxes}

In this section, we present distillation and dilution of quantum strategy boxes.
A special case of this theory involves distillation and dilution of sequential
channel boxes. Here we are interested in not only in the optimal number but
also \textit{rates} at which one can distill or dilute bits of asymmetric
distinguishability from or to a strategy or sequential channel box,
respectively, both in the exact and approximate cases.

All of the basic definitions in this case represent generalizations of what we have presented previously for one-shot tasks regarding quantum channels. As such, we do not delve into as many details as we did before but mainly state the results and provide brief justifications.

\subsection{Exact case: distillable distinguishability}

\label{sec:exact-distill-distingui-1-shot}

Given a strategy box $(\mathcal{N}^{(n)},\mathcal{M}^{(n)})$, the exact
distillable distinguishability is equal to the largest $M$ such that we can
transform $(\mathcal{N}^{(n)},\mathcal{M}^{(n)})$ to the channel box
$(\mathcal{R}_{C\rightarrow D}^{|0\rangle\langle0|},\mathcal{R}_{C\rightarrow
D}^{\pi_{M}})$ exactly by means of a physical transformation $\Theta
^{(n\rightarrow1)}$. Note that the physical transformation
$\Theta
^{(n\rightarrow1)}$ is a special case of those that we discussed previously in Section~\ref{sec:phys-trans-strats}, taking an $n$-round quantum strategy box to a channel box. Mathematically, the exact
distillable distinguishability is defined as the following
optimization problem:%
\begin{multline}
D_{d}^{0}(\mathcal{N}^{(n)},\mathcal{M}^{(n)}):=\\
\sup_{\Theta^{(n\rightarrow1)}}\left\{
\begin{array}
[c]{c}%
\log_{2}M:\Theta^{(n\rightarrow1)}(\mathcal{N}^{(n)})=\mathcal{R}%
_{C\rightarrow D}^{|0\rangle\langle0|},\\
\Theta^{(n\rightarrow1)}(\mathcal{M}^{(n)})=\mathcal{R}_{C\rightarrow D}%
^{\pi_{M}}%
\end{array}
\right\}  .
\end{multline}

Note that this problem is essentially equivalent to $\left\lfloor D_{d}%
^{0}(\mathcal{N}^{(n)},\mathcal{M}^{(n)})\right\rfloor $, which is the largest
$m$ for which a physical transformation $\Theta^{(n\rightarrow m)}$\ exists
such that%
\begin{align}
\Theta^{(n\rightarrow m)}(\mathcal{N}^{(n)})  &  =(\mathcal{R}_{C\rightarrow
D}^{|0\rangle\langle0|})^{(m)},\\
\Theta^{(n\rightarrow m)}(\mathcal{M}^{(n)})  &  =(\mathcal{R}_{C\rightarrow
D}^{\pi})^{(m)},
\end{align}
where the superscript $(m)$ indicates $m$ sequential channel uses. This is
because the channel box $(\mathcal{R}_{C\rightarrow D}^{|0\rangle\langle
0|},\mathcal{R}_{C\rightarrow D}^{\pi_{2^{m}}})$ and the sequential channel box
$((\mathcal{R}_{C\rightarrow D}^{|0\rangle\langle0|})^{(m)},(\mathcal{R}%
_{C\rightarrow D}^{\pi})^{(m)})$ are equivalent to each other by means of
common quantum strategies, due to the fact that the underlying channel pairs
are environment seizable and thus equivalent to state boxes.

By employing reasoning similar to that which we employed  previously to justify \eqref{eq:exact-distill-dist-min}, we conclude
that%
\begin{equation}
D_{d}^{0}(\mathcal{N}^{(n)},\mathcal{M}^{(n)})=D_{\min}(\mathcal{N}^{(n)}%
\Vert\mathcal{M}^{(n)}),
\end{equation}
where $D_{\min}(\mathcal{N}^{(n)}\Vert\mathcal{M}^{(n)})$ is the quantum
strategy divergence from Definition~\ref{def:gen-strat-div}, with $\mathbf{D}$ therein set to
$D_{\min}$. The main reasons that this equality holds are that 1)\ the optimal
co-strategy for $D_{\min}(\mathcal{N}^{(n)}\Vert\mathcal{M}^{(n)})$ leads to a
protocol for distilling bits of asymmetric distinguishability and 2)\ its
optimality follows from the data processing inequality (Theorem~\ref{thm:DP-strat-div}) for $D_{\min
}(\mathcal{N}^{(n)}\Vert\mathcal{M}^{(n)})$ with respect to an arbitrary
physical transformation $\Theta^{(n\rightarrow1)}$ that produces the channel
box $(\mathcal{R}_{C\rightarrow D}^{|0\rangle\langle0|},\mathcal{R}%
_{C\rightarrow D}^{\pi_{M}})$ exactly.

If the strategy box $(\mathcal{N}^{(n)},\mathcal{M}^{(n)})$ is in fact a
sequential channel box for all $n$, with corresponding channels $\mathcal{N}$
and $\mathcal{M}$, then we define the exact sequential distillable
distinguishability as%
\begin{equation}
D_{d}^{0}(\mathcal{N},\mathcal{M}):=\lim_{n\rightarrow\infty}\frac{1}{n}%
D_{d}^{0}(\mathcal{N}^{(n)},\mathcal{M}^{(n)}).
\end{equation}
Just as with the parallel case discussed in Section~\ref{sec:exact-distill-distingui-1-shot}, the underlying quantity
$D_{d}^{0}(\mathcal{N}^{(n)},\mathcal{M}^{(n)})$ can jump from zero to
$\infty$ as $n$ increases. In fact, this jump can occur in the simplest case
when $n$ goes from one to two \cite{Harrow10}. By the general bound from
\cite{BHKW18}, we have that%
\begin{equation}
D_{d}^{0}(\mathcal{N},\mathcal{M})\leq D_{\max}(\mathcal{N}\Vert\mathcal{M}).
\end{equation}

\subsection{Exact case: distinguishability cost}

Given a strategy box $(\mathcal{N}^{(n)},\mathcal{M}^{(n)})$, the exact
distinguishability cost is equal to the smallest $M$ such that we can
transform the channel box $(\mathcal{R}_{C\rightarrow D}^{|0\rangle\langle
0|},\mathcal{R}_{C\rightarrow D}^{\pi_{M}})$ to $(\mathcal{N}^{(n)}%
,\mathcal{M}^{(n)})$ exactly by means of a physical transformation
$\Theta^{(1\rightarrow n)}$.
Note that the physical transformation
$\Theta
^{(1\rightarrow n)}$ is a special case of those that we discussed previously in Section~\ref{sec:phys-trans-strats}, taking a channel box to an $n$-round quantum strategy box.
Mathematically, the exact
distinguishability cost is defined as the following
optimization problem:%
\begin{multline}
D_{c}^{0}(\mathcal{N}^{(n)},\mathcal{M}^{(n)}):=\\
\inf_{\Theta^{(1\rightarrow n)}}\left\{
\begin{array}
[c]{c}%
\log_{2}M:\mathcal{N}^{(n)}=\Theta^{(1\rightarrow n)}(\mathcal{R}%
_{C\rightarrow D}^{|0\rangle\langle0|}),\\
\mathcal{M}^{(n)}=\Theta^{(1\rightarrow n)}(\mathcal{R}_{C\rightarrow D}%
^{\pi_{M}})
\end{array}
\right\}  .
\end{multline}

For similar reasons stated in the previous section, this problem is
essentially equivalent to $\left\lceil D_{c}^{0}(\mathcal{N}^{(n)}%
,\mathcal{M}^{(n)})\right\rceil $, which is the smallest $m$ for which a
physical transformation $\Theta^{(m\rightarrow n)}$\ exists such that%
\begin{align}
\mathcal{N}^{(n)}  &  =\Theta^{(m\rightarrow n)}((\mathcal{R}_{C\rightarrow
D}^{|0\rangle\langle0|})^{(m)}),\\
\mathcal{M}^{(n)}  &  =\Theta^{(m\rightarrow n)}((\mathcal{R}_{C\rightarrow
D}^{\pi})^{(m)}),
\end{align}
where the superscript $(m)$ again indicates $m$ sequential channel uses.

By employing reasoning similar to that which we used previously to justify \eqref{eq:exact-dist-cost-max}, we conclude
that%
\begin{equation}
D_{c}^{0}(\mathcal{N}^{(n)},\mathcal{M}^{(n)})=D_{\max}(\mathcal{N}^{(n)}%
\Vert\mathcal{M}^{(n)}), \label{eq:max-div-strategy}%
\end{equation}
where $D_{\max}(\mathcal{N}^{(n)}\Vert\mathcal{M}^{(n)})$ is the quantum
strategy divergence from Definition~\ref{def:gen-strat-div}, with $\mathbf{D}$ therein set to
$D_{\max}$. The quantity $D_{\max}(\mathcal{N}^{(n)}\Vert\mathcal{M}^{(n)})$ has already been defined in and studied in \cite{Chiribella_2016}, wherein it was shown that it is equal to the max-relative entropy of the Choi operators of the strategies. Eq.~\eqref{eq:max-div-strategy} gives $D_{\max}(\mathcal{N}^{(n)}\Vert\mathcal{M}^{(n)})$ its fundamental  operational meaning in terms of the exact distinguishability cost of the strategy box $(\mathcal{N}^{(n)},\mathcal{M}^{(n)})$.  

The main reasons that the equality in \eqref{eq:max-div-strategy} holds are that 1)\ an optimal
dilution protocol, generalizing that from Appendix~\ref{app:ch-max-rel-ent-exact-dist-cost}, results from a strategy that outputs strategy $\mathcal{N}^{(n)}$  if $|0\rangle \langle 0 | $ is input and outputs the strategy
\begin{equation}
\frac{2^{D_{\max}(\mathcal{N}^{(n)}\Vert\mathcal{M}^{(n)})}\mathcal{M}^{(n)}-\mathcal{N}^{(n)}}{2^{D_{\max}(\mathcal{N}^{(n)}\Vert\mathcal{M}^{(n)})} - 1}
\end{equation} if $|1\rangle \langle 1 |$ is input  and 2)\ its optimality follows from the data processing inequality (Theorem~\ref{thm:DP-strat-div})
for $D_{\max}(\mathcal{N}^{(n)}\Vert\mathcal{M}^{(n)})$ with respect to an
arbitrary physical transformation $\Theta^{(1\rightarrow n)}$ that produces
the strategy box $(\mathcal{N}^{(n)},\mathcal{M}^{(n)})$\ exactly from the
channel box $(\mathcal{R}_{C\rightarrow D}^{|0\rangle\langle0|},\mathcal{R}%
_{C\rightarrow D}^{\pi_{M}})$.

If the strategy box $(\mathcal{N}^{(n)},\mathcal{M}^{(n)})$ is in fact a
sequential channel box for all $n$, with corresponding channels $\mathcal{N}$
and $\mathcal{M}$, then we define the exact sequential distinguishability cost
as%
\begin{equation}
D_{c}^{0}(\mathcal{N},\mathcal{M}):=\lim_{n\rightarrow\infty}\frac{1}{n}%
D_{c}^{0}(\mathcal{N}^{(n)},\mathcal{M}^{(n)}).
\end{equation}
Note that the following inequality holds%
\begin{equation}
D_{c}^{0}(\mathcal{N},\mathcal{M})\geq D_{c}^{0,p}(\mathcal{N},\mathcal{M}),
\label{eq:seq-par-exact-cost}%
\end{equation}
because a sequential simulation is more stringent than a parallel simulation.
That is, any sequential simulation works as a parallel simulation.

A key result that we have for this problem, strengthening our earlier finding
from \eqref{eq:par-exact-cost-dmax}, is expressed by the following theorem.

\begin{theorem}
For channels $\mathcal{N}$ and $\mathcal{M}$, the exact sequential
distinguishability cost is equal to the channel max-relative entropy:%
\begin{equation}
D_{c}^{0}(\mathcal{N},\mathcal{M})=D_{\max}(\mathcal{N}\Vert\mathcal{M}).
\label{eq:dmax-seq-exact-cost}%
\end{equation}

\end{theorem}

\begin{proof}
The inequality $D_{c}^{0}(\mathcal{N},\mathcal{M})\geq D_{\max}(\mathcal{N}%
\Vert\mathcal{M})$ is a consequence of \eqref{eq:seq-par-exact-cost} and
\eqref{eq:par-exact-cost-dmax}. The other inequality is a consequence of the
following scheme for simulating the sequential channel box $(\mathcal{N}%
^{(n)},\mathcal{M}^{(n)})$, similar to that employed in
\cite{GFWRSCW18,WW19PPT}. In the first round of the sequential simulation, one
starts from the channel box $(\mathcal{R}_{C\rightarrow D}^{|0\rangle
\langle0|},\mathcal{R}_{C\rightarrow D}^{\pi_{M}})$ and simulates the tensor
product channel box $(\mathcal{R}_{C\rightarrow D}^{|0\rangle\langle0|}%
\otimes\mathcal{N},\mathcal{R}_{C\rightarrow D}^{\pi_{M_{1}}}\otimes
\mathcal{M})$. Employing \eqref{eq:exact-dist-cost-max}, the cost for doing so is%
\begin{align}
\log_{2}M  &  =D_{\max}(\mathcal{R}_{C\rightarrow D}^{|0\rangle\langle
0|}\otimes\mathcal{N}\Vert\mathcal{R}_{C\rightarrow D}^{\pi_{M_{1}}}%
\otimes\mathcal{M})\\
&  =D_{\max}(\mathcal{R}_{C\rightarrow D}^{|0\rangle\langle0|}\Vert
\mathcal{R}_{C\rightarrow D}^{\pi_{M_{1}}})+D_{\max}(\mathcal{N}%
\Vert\mathcal{M})\\
&  =\log_{2}M_{1}+D_{\max}(\mathcal{N}\Vert\mathcal{M}).
\end{align}
In the next round, one uses the leftover channel box $(\mathcal{R}%
_{C\rightarrow D}^{|0\rangle\langle0|},\mathcal{R}_{C\rightarrow D}%
^{\pi_{M_{1}}})$ to simulate the channel box $(\mathcal{R}_{C\rightarrow
D}^{|0\rangle\langle0|}\otimes\mathcal{N},\mathcal{R}_{C\rightarrow D}%
^{\pi_{M_{2}}}\otimes\mathcal{M})$. Again employing \eqref{eq:exact-dist-cost-max} and an analysis
similar to the above, the cost for doing so is%
\begin{equation}
\log_{2}M_{1}=\log_{2}M_{2}+D_{\max}(\mathcal{N}\Vert\mathcal{M}).
\end{equation}
This continues until the last round, and adding everything up, the total cost
for the simulation of the sequential channel box $(\mathcal{N}^{(n)}%
,\mathcal{M}^{(n)})$ is $nD_{\max}(\mathcal{N}\Vert\mathcal{M})$. Since this
holds for every $n$, we conclude that $D_{c}^{0}(\mathcal{N},\mathcal{M})\leq
D_{\max}(\mathcal{N}\Vert\mathcal{M})$, and in turn, we conclude \eqref{eq:dmax-seq-exact-cost}.
\end{proof}

\subsection{Approximate case: distillable distinguishability}

Given a strategy box $(\mathcal{N}^{(n)},\mathcal{M}^{(n)})$, the approximate
distillable distinguishability is equal to the largest $M$ such that we can
transform the strategy box $(\mathcal{N}^{(n)},\mathcal{M}^{(n)})$ to the channel box
$(\mathcal{R}_{C\rightarrow D}^{|0\rangle\langle0|},\mathcal{R}_{C\rightarrow
D}^{\pi_{M}})$ approximately by means of a physical transformation
$\Theta^{(n\rightarrow1)}$. Mathematically, it is defined as the following
optimization problem:%
\begin{multline}
D_{d}^{\varepsilon}(\mathcal{N}^{(n)},\mathcal{M}^{(n)}):=\\
\sup_{\Theta^{(n\rightarrow1)}}\left\{
\begin{array}
[c]{c}%
\log_{2}M:\Theta^{(n\rightarrow1)}(\mathcal{N}^{(n)})\approx_{\varepsilon
}\mathcal{R}_{C\rightarrow D}^{|0\rangle\langle0|},\\
\Theta^{(n\rightarrow1)}(\mathcal{M}^{(n)})=\mathcal{R}_{C\rightarrow D}%
^{\pi_{M}}%
\end{array}
\right\}  ,
\end{multline}
where the shorthand $\approx_{\varepsilon}$ is defined in \eqref{eq:strategy-dist-approx}--\eqref{eq:strategy-dist-approx-shorthand} in terms of the normalized strategy distance.

For similar reasons stated in the previous section, this problem is
essentially equivalent to $\left\lfloor D_{d}^{\varepsilon}(\mathcal{N}%
^{(n)},\mathcal{M}^{(n)})\right\rfloor $, which is the largest $m$ for which a
physical transformation $\Theta^{(n\rightarrow m)}$\ exists such that%
\begin{align}
\Theta^{(n\rightarrow m)}(\mathcal{N}^{(n)})  &  \approx_{\varepsilon} (\mathcal{R}_{C\rightarrow
D}^{|0\rangle\langle0|})^{(m)},\\
\Theta^{(n\rightarrow m)}(\mathcal{M}^{(n)})  &  =(\mathcal{R}_{C\rightarrow
D}^{\pi})^{(m)}.
\end{align}

By employing reasoning similar to that which we used previously to justify \eqref{eq:approx-distill-dist-smooth-min}, we conclude
that%
\begin{equation}
D_{d}^{\varepsilon}(\mathcal{N}^{(n)},\mathcal{M}^{(n)})=D_{\min}%
^{\varepsilon}(\mathcal{N}^{(n)}\Vert\mathcal{M}^{(n)}),
\end{equation}
where $D_{\min}^{\varepsilon}(\mathcal{N}^{(n)}\Vert\mathcal{M}^{(n)})$ is the
quantum strategy divergence from Definition~\ref{def:gen-strat-div}, with $\mathbf{D}$ therein set to
$D_{\min}^{\varepsilon}$. The main reasons that this equality holds are that
1)\ the optimal co-strategy for $D_{\min}^{\varepsilon}(\mathcal{N}^{(n)}%
\Vert\mathcal{M}^{(n)})$ leads to a protocol for distilling bits of asymmetric
distinguishability approximately and 2)\ its optimality follows from the data
processing inequality for $D_{\min}^{\varepsilon}(\mathcal{N}^{(n)}%
\Vert\mathcal{M}^{(n)})$ with respect to an arbitrary physical transformation
$\Theta^{(n\rightarrow1)}$ that produces the channel box $(\mathcal{R}%
_{C\rightarrow D}^{|0\rangle\langle0|},\mathcal{R}_{C\rightarrow D}^{\pi_{M}%
})$ approximately.

If the strategy box $(\mathcal{N}^{(n)},\mathcal{M}^{(n)})$ is in fact a
sequential channel box for all $n$, with corresponding channels $\mathcal{N}$
and $\mathcal{M}$, then we define the sequential distillable
distinguishability as%
\begin{equation}
D_{d}(\mathcal{N},\mathcal{M}):=\lim_{\varepsilon\rightarrow0}\lim
_{n\rightarrow\infty}\frac{1}{n}D_{d}^{\varepsilon}(\mathcal{N}^{(n)}%
,\mathcal{M}^{(n)}).
\end{equation}
A key result of our paper is the following formal expression for
$D_{d}(\mathcal{N},\mathcal{M})$ in terms of the amortized channel relative
entropy from \cite{BHKW18}:

\begin{theorem}
\label{thm:seq-distill-distinguish-amortized}For channels $\mathcal{N}$ and
$\mathcal{M}$, the sequential distillable distinguishability is equal to the
amortized channel relative entropy of \cite{BHKW18}:%
\begin{equation}
D_{d}(\mathcal{N},\mathcal{M})=D^{\mathcal{A}}(\mathcal{N}\Vert\mathcal{M}),
\end{equation}
where%
\begin{multline}
D^{\mathcal{A}}(\mathcal{N}\Vert\mathcal{M}):=\\
\sup_{\rho_{RA},\sigma_{RA}}D(\mathcal{N}_{A\rightarrow B}(\rho_{RA}%
)\Vert\mathcal{M}_{A\rightarrow B}(\sigma_{RA}))-D(\rho_{RA}\Vert\sigma_{RA}).
\end{multline}

\end{theorem}

\begin{proof}
The bound
\begin{equation}
D_{d}(\mathcal{N},\mathcal{M})\leq D^{\mathcal{A}}(\mathcal{N}\Vert
\mathcal{M})
\end{equation}
follows from \cite[Proposition~16]{BHKW18}, due to the equivalence between
sequential distillable distinguishability and the optimal rate of the quantum hypothesis testing
problem considered in \cite{BHKW18}. So it remains to establish the opposite inequality.

To do so, here we employ a technique used in the resource theory of coherence
\cite[Theorem~17]{GFWRSCW18}, which was used therein to show that the amortized relative entropy of coherence  
 is equal to the distillable coherence of a quantum channel. A similar technique was also discussed previously in \cite[Section~2.4]{BHLS03}.
 
 Let
$\rho_{RA}$ and $\sigma_{RA}$ be arbitrary quantum states. Let $\psi_{RA}$ be
a state such that
\begin{equation}
D(\mathcal{N}_{A\rightarrow B}(\psi_{RA})\Vert
\mathcal{M}_{A\rightarrow B}(\psi_{RA}))>0.
\end{equation}
(If such a state does not exist,
then $D_{d}(\mathcal{N},\mathcal{M})$ is trivially equal to zero.) The first
step is to send in the tensor-power state $\psi_{RA}^{\otimes m}$ to $m$ parallel calls
of the unknown channel, where%
\begin{equation}
m\gtrsim n\frac{D(\rho_{RA}\Vert\sigma_{RA})+\delta}{D(\mathcal{N}%
_{A\rightarrow B}(\psi_{RA})\Vert\mathcal{M}_{A\rightarrow B}(\psi
_{RA}))-\delta}%
\end{equation}
for $\delta>0$, and distill bits of asymmetric distinguishability at the rate
$D(\mathcal{N}_{A\rightarrow B}(\psi_{RA})\Vert\mathcal{M}_{A\rightarrow
B}(\psi_{RA}))$. Second, we dilute these bits of asymmetric distinguishability
to the state box $(\rho_{RA}^{\otimes n},\sigma_{RA}^{\otimes n})$. Third, we
then send this state box into $n$ uses of the unknown channel, producing the state box
$([\mathcal{N}_{A\rightarrow B}(\rho_{RA})]^{\otimes n},[\mathcal{M}%
_{A\rightarrow B}(\sigma_{RA})]^{\otimes n})$. Fourth, from this state box, we
distill bits of asymmetric distinguishability at the rate $D(\mathcal{N}%
_{A\rightarrow B}(\rho_{RA})\Vert\mathcal{M}_{A\rightarrow B}(\sigma
_{RA}))-\delta$. We output a fraction $R-2\delta$ of these bits, where%
\begin{equation}
R:=D(\mathcal{N}_{A\rightarrow B}(\rho_{RA})\Vert\mathcal{M}_{A\rightarrow
B}(\sigma_{RA}))-D(\rho_{RA}\Vert\sigma_{RA}),
\end{equation}
and then reinvest a fraction $D(\rho_{RA}\Vert\sigma_{RA})+\delta$ for the
next round. We then repeat steps 2) through 4)\ $k$ times. In the last round, a
fraction $R_{f}-\delta$ bits of asymmetric distinguishability are output,
where%
\begin{equation}
R_{f}:=D(\mathcal{N}_{A\rightarrow B}(\rho_{RA})\Vert\mathcal{M}_{A\rightarrow
B}(\sigma_{RA})),
\end{equation}
and no reinvestment is made (because it is the last round). Counting up everything, this protocol calls the
unknown channel $kn+m$ times, while outputting%
\begin{equation}
\left(  k-1\right)  n\left(  R-2\delta\right)  +n\left(  R_{f}-\delta\right)
\end{equation}
bits of asymmetric distinguishability. Thus, the rate of the protocol is given
by%
\begin{equation}
\frac{\left(  k-1\right)  n\left(  R-2\delta\right)  +n\left(  R_{f}%
-\delta\right)  }{kn+m}.
\end{equation}
In the limit as $k\rightarrow\infty$, this rate converges to $R-2\delta$.
Since $\delta>0$ is arbitrary, the rate $R$ is achievable. Note that all of the conversions stated above are approximate, but for large enough $n$ and by employing the triangle inequality, the error vanishes. Finally, since the
states $\rho_{RA}$ and $\sigma_{RA}$ are arbitrary, we can take a supremum over all of them and conclude the inequality%
\begin{equation}
D_{d}(\mathcal{N},\mathcal{M})\geq D^{\mathcal{A}}(\mathcal{N}\Vert
\mathcal{M}),
\end{equation}
thus completing the proof.
\end{proof}

\bigskip

Theorem~\ref{thm:seq-distill-distinguish-amortized} establishes an operational
meaning for the amortized channel relative entropy of \cite{BHKW18}, thus giving it some
distinction in the resource theory of asymmetric distinguishability for
quantum channels. Theorem~\ref{thm:seq-distill-distinguish-amortized} can
alternatively be understood as a formal solution to Stein's lemma for quantum
channels in the sequential setting, thus completing the line of reasoning
put forward in \cite{BHKW18}.

More generally, this result can be used to determine whether a sequential
protocol is truly necessary to attain the optimal distillable
distinguishability. If an amortization collapse occurs for a pair of channels,
so that $D^{\mathcal{A}}(\mathcal{N}\Vert\mathcal{M})=D(\mathcal{N}%
\Vert\mathcal{M})$, then one can conclude that a sequential protocol is not
necessary and one can simply input a tensor-power state $\psi_{RA}^{\otimes
n}$ to distinguish the channels optimally in the asymptotic regime
\cite{BHKW18}. This collapse occurs for both environment-seizable and
classical--quantum channel boxes. It also occurs for channel boxes in which the first channel is arbitrary and the second is a replacer channel \cite{Cooney2016,BHKW18}. What
Theorems~\ref{thm:parallel-distillable-dist-regularized-rel-ent}\ and
\ref{thm:seq-distill-distinguish-amortized}\ add to this story is that the
condition%
\begin{equation}
D^{\mathcal{A}}(\mathcal{N}\Vert\mathcal{M})>\lim_{n\rightarrow\infty}\frac
{1}{n}D(\mathcal{N}^{\otimes n}\Vert\mathcal{M}^{\otimes n})
\label{eq:seq-vs-parallel}
\end{equation}
is necessary and sufficient for an adaptive strategy to have an advantage over
a parallel strategy in the setting of asymmetric channel discrimination, or
equivalently, when distilling bits of asymmetric distinguishability. Determining whether \eqref{eq:seq-vs-parallel} holds for a pair of quantum channels is an interesting and challenging open problem.

It seems that the main idea of \cite[Theorem~17]{GFWRSCW18} (also \cite[Section~2.4]{BHLS03}), as used in the proof of Theorem~\ref{thm:seq-distill-distinguish-amortized}, can be employed for a sequential distillation task in any quantum resource thery for which the static version of the theory (for quantum states) is asymptotically reversible. This is because the interleaving of distillation and dilution plays an essential role in the given protocol, and for an asymptotically reversible resource theory, there is no loss when going back and forth like this. 

\subsection{Approximate case: distinguishability cost}

Given a strategy box $(\mathcal{N}^{(n)},\mathcal{M}^{(n)})$, the approximate
distinguishability cost is equal to the smallest $M$ such that we can
transform the channel box $(\mathcal{R}_{C\rightarrow D}^{|0\rangle\langle
0|},\mathcal{R}_{C\rightarrow D}^{\pi_{M}})$ to $(\mathcal{N}^{(n)}%
,\mathcal{M}^{(n)})$ approximately by means of a physical transformation
$\Theta^{(1\rightarrow n)}$. Mathematically, it is defined as the following
optimization problem:%
\begin{multline}
D_{c}^{\varepsilon}(\mathcal{N}^{(n)},\mathcal{M}^{(n)}):=\\
\inf_{\Theta^{(1\rightarrow n)}}\left\{
\begin{array}
[c]{c}%
\log_{2}M:\mathcal{N}^{(n)}\approx_{\varepsilon}\Theta^{(1\rightarrow
n)}(\mathcal{R}_{C\rightarrow D}^{|0\rangle\langle0|}),\\
\mathcal{M}^{(n)}=\Theta^{(1\rightarrow n)}(\mathcal{R}_{C\rightarrow D}%
^{\pi_{M}})
\end{array}
\right\}  .
\end{multline}

For similar reasons stated previously, this problem is essentially equivalent
to $\left\lceil D_{c}^{\varepsilon}(\mathcal{N}^{(n)},\mathcal{M}%
^{(n)})\right\rceil $, which is the smallest $m$ for which a physical
transformation $\Theta^{(m\rightarrow n)}$\ exists such that%
\begin{align}
\mathcal{N}^{(n)}  &  \approx_{\varepsilon}\Theta^{(m\rightarrow
n)}((\mathcal{R}_{C\rightarrow D}^{|0\rangle\langle0|})^{(m)}),\\
\mathcal{M}^{(n)}  &  =\Theta^{(m\rightarrow n)}((\mathcal{R}_{C\rightarrow
D}^{\pi})^{(m)}),
\end{align}
where the superscript $(m)$ again indicates $m$ sequential channel uses.

By employing reasoning similar to that which we used previously to justify \eqref{eq:approx-dist-cost-smooth-max}, we conclude
that%
\begin{equation}
D_{c}^{\varepsilon}(\mathcal{N}^{(n)},\mathcal{M}^{(n)})=D_{\max}%
^{\varepsilon}(\mathcal{N}^{(n)}\Vert\mathcal{M}^{(n)}),
\end{equation}
where the smooth strategy max-relative entropy is defined as
\begin{multline}
D_{\max}^{\varepsilon}(\mathcal{N}^{(n)}\Vert\mathcal{M}^{(n)}):=\\
\inf_{\frac{1}{2}\left\Vert \widetilde{\mathcal{N}}^{(n)}-\mathcal{N}%
^{(n)}\right\Vert _{\diamond n}\leq\varepsilon}D_{\max}(\widetilde
{\mathcal{N}}^{(n)}\Vert\mathcal{M}^{(n)}),
\label{eq:smooth-dmax-strat}
\end{multline}
and $D_{\max}(\widetilde{\mathcal{N}}^{(n)}\Vert\mathcal{M}^{(n)})$ is defined
in \eqref{eq:max-div-strategy}. The infimum in \eqref{eq:smooth-dmax-strat} is with respect to $n$-round strategies $\widetilde
{\mathcal{N}}^{(n)}$ that are $\varepsilon$-close in normalized strategy distance to the strategy $\mathcal{N}^{(n)}$. The main reasons that this equality holds are
that 1)\ an optimal approximate dilution protocol results from applying an
optimal exact dilution protocol to $\widetilde{\mathcal{N}}^{(n)}$ and
$\mathcal{M}^{(n)}$, where $\widetilde{\mathcal{N}}^{(n)}$ is $\varepsilon
$-close to $\mathcal{N}^{(n)}$ with respect to the normalized strategy
distance and 2)\ its optimality follows from the data processing inequality
for $D_{\max}^{\varepsilon}(\mathcal{N}^{(n)}\Vert\mathcal{M}^{(n)})$ with
respect to an arbitrary physical transformation $\Theta^{(1\rightarrow n)}$
that produces the strategy box $(\mathcal{N}^{(n)},\mathcal{M}^{(n)}%
)$\ approximately from the channel box $(\mathcal{R}_{C\rightarrow
D}^{|0\rangle\langle0|},\mathcal{R}_{C\rightarrow D}^{\pi_{M}})$.

If the strategy box $(\mathcal{N}^{(n)},\mathcal{M}^{(n)})$ is in fact a
sequential channel box for all $n$, with corresponding channels $\mathcal{N}$
and $\mathcal{M}$, then we define the sequential distinguishability cost
as%
\begin{equation}
D_{c}(\mathcal{N},\mathcal{M}):=\lim_{\varepsilon \to 0}\lim_{n\rightarrow\infty}\frac{1}{n}%
D_{c}^{\varepsilon}(\mathcal{N}^{(n)},\mathcal{M}^{(n)}).
\end{equation}
Note that the following inequality holds%
\begin{equation}
D_{c}(\mathcal{N},\mathcal{M})\geq D_{c}^{p}(\mathcal{N},\mathcal{M}),
\label{eq:seq-par-approx-cost}%
\end{equation}
because a sequential simulation is more stringent than a parallel simulation.
That is, any sequential simulation works as a parallel simulation.

As occurred for all other tasks in this paper, the sequential distinguishability cost simplifies for environment-seizable channel boxes.
It remains an interesting open question to understand the sequential
distinguishability cost of quantum channel boxes other than environment-seizable ones.

\section{Conclusion}

In this paper, we generalized the resource theory of asymmetric distinguishability from states \cite{Mats10,M11,WW19states} to channels. In this resource theory, the main constituents are quantum channel boxes that can be manipulated by means of a quantum superchannel, the most general physical transformation that sends quantum channels to quantum channels. Furthermore, the basic units of currency are bits of asymmetric distinguishability \cite{WW19states}.

In the one-shot scenario, we considered the approximate channel box transformation problem and proved that it is characterized by a semi-definite program. As special cases of this, we considered exact and approximate one-shot distillation and dilution of channel boxes, arriving at the following conclusions:
\begin{enumerate}
\item \textit{The exact one-shot distillable distinguishability of a channel box is equal to the channel min-relative entropy.}
\item \textit{The exact one-shot  distinguishability cost of a channel box is equal to the channel max-relative entropy.}
\item \textit{The approximate one-shot distillable distinguishability of a channel box is equal to the smooth channel min-relative entropy.}
\item \textit{The approximate one-shot  distinguishability cost of a channel box is equal to the smooth channel max-relative entropy.}
\end{enumerate}
These results endow these fundamental channel measures of distinguishability with operational interpretations.

We then moved on to consider asymptotic parallel versions of the above tasks, with our key findings here being that the parallel distillable distinguishability is equal to the regularized channel relative entropy and the parallel exact distinguishability cost is equal to the channel max-relative entropy. We solved the asymptotic version of the parallel channel box transformation problem for environment-seizable and classical--quantum channel boxes.

We finally considered the approximate strategy box transformation problem and asserted that it is characterized by a semi-definite program. We introduced the generalized strategy divergence as a way of quantifying distinguishability of quantum strategies and used instantiations of this concept to provide bounds on how well one can convert one strategy box to another. In particular, transformations of sequential channel boxes are a special case of transformations of strategy boxes, so that many of the results for strategy boxes apply directly, and all of the results simplify for environment-seizable or classical--quantum sequential channel boxes.
%Transformations of sequential channel boxes are a special case of transformations of strategy boxes, so that many of the results for strategy boxes apply directly, and all of the results simplify for environment-seizable or classical--quantum sequential channel boxes.

By focusing on distillation and dilution tasks, we proved that the asymptotic sequential distillable distinguishability of a sequential channel box is equal to the amortized channel relative entropy of \cite{BHKW18}, thus endowing this quantity with a fundamental operational meaning. We also proved that the exact sequential distinguishability cost is equal to the channel max-relative entropy.

Going forward from here, there are many open questions for future work. Are there other channel boxes, besides environment-seizable and classical--quantum ones, for which the theory simplifies significantly? Based on the distillation results of \cite{Cooney2016}, and other findings of \cite{FWTB18}, it seems plausible that the channel relative entropy should be the optimal rate for dilution protocols of channel boxes in which the first channel is arbitrary and the second is a replacer channel. Are there examples of channel boxes for which the regularization in the regularized channel relative entropy is necessary? Are there examples of channel boxes for which the amortized channel relative entropy does not collapse to the ordinary channel relative entropy? Answers to these questions would provide insights as to whether general parallel or sequential strategies are helpful in distinguishability distillation. Can we characterize the asymptotic parallel or sequential distinguishability cost, in the case in which the simulation need not be exact but with vanishing error in the asymptotic limit? Is it possible to give a more general theory beyond independent and identically distributed channels, i.e., for memory channels with some structure? These and other questions remain the subject of future investigations.

Note: After our paper appeared online, the preprint \cite{FFRS19} was posted, which has addressed some of the open questions from our paper.

\begin{acknowledgments}
We are grateful to Vishal Katariya for pointing out a problem with our previous formulation of the semi-definite program for the smooth channel max-relative entropy. We are grateful to Andreas Winter for pointing out \cite[Section~2.4]{BHLS03} in the context of Theorem~\ref{thm:seq-distill-distinguish-amortized}.
XW acknowledges support from the Department of Defense, and MMW acknowledges support from the National Science Foundation under Grant No.~1907615.
\end{acknowledgments}

\bibliographystyle{alpha}
\bibliography{Ref}

\newcommand{\etalchar}[1]{$^{#1}$}
\begin{thebibliography}{MLDS{\etalchar{+}}13}

\bibitem[ABJT19]{ABJT19}
Anurag Anshu, Mario Berta, Rahul Jain, and Marco Tomamichel.
\newblock A minimax approach to one-shot entropy inequalities.
\newblock June 2019.
\newblock 1906.00333v1.

\bibitem[Aci01]{Acin01}
Antonio Acin.
\newblock Statistical distinguishability between unitary operations.
\newblock {\em Physical Review Letters}, 87(17):177901, October 2001.
\newblock arXiv:quant-ph/0102064.

\bibitem[AU80]{AU80}
P.~M. Alberti and A.~Uhlmann.
\newblock A problem relating to positive linear maps on matrix algebras.
\newblock {\em Reports on Mathematical Physics}, 18(2):163--176, October 1980.

\bibitem[BaHN{\etalchar{+}}15]{BHNOW15}
Fernando G. S.~L. Brand\~{a}o, Michal Horodecki, Nelly Ng, Jonathan Oppenheim,
  and Stephanie Wehner.
\newblock The second laws of quantum thermodynamics.
\newblock {\em Proceedings of the National Academy of Sciences},
  112(11):3275--3279, March 2015.
\newblock arXiv:1305.5278.

\bibitem[BBCW13]{BBCW13}
Mario Berta, Fernando G. S.~L. Brand{\~a}o, Matthias Christandl, and Stephanie
  Wehner.
\newblock Entanglement cost of quantum channels.
\newblock {\em IEEE Transactions on Information Theory}, 59(10):6779--6795,
  October 2013.
\newblock arXiv:1108.5357.

\bibitem[BCR11]{BCR09}
Mario Berta, Matthias Christandl, and Renato Renner.
\newblock The quantum reverse {Shannon} theorem based on one-shot information
  theory.
\newblock {\em Communications in Mathematical Physics}, 306(3):579--615, August
  2011.
\newblock arXiv:0912.3805.

\bibitem[BD10]{BD10}
Francesco Buscemi and Nilanjana Datta.
\newblock The quantum capacity of channels with arbitrarily correlated noise.
\newblock {\em IEEE Transactions on Information Theory}, 56(3):1447--1460,
  March 2010.
\newblock arXiv:0902.0158.

\bibitem[BD11]{BD11}
Fernando G. S.~L. Brandao and Nilanjana Datta.
\newblock One-shot rates for entanglement manipulation under non-entangling
  maps.
\newblock {\em IEEE Transactions on Information Theory}, 57(3):1754--1760,
  March 2011.
\newblock arXiv:0905.2673.

\bibitem[BD16]{BD16}
Francesco Buscemi and Nilanjana Datta.
\newblock Equivalence between divisibility and monotonic decrease of
  information in classical and quantum stochastic processes.
\newblock {\em Physical Review A}, 93(1):012101, January 2016.
\newblock arXiv:1408.7062.

\bibitem[BDH{\etalchar{+}}14]{BDHSW09}
Charles~H. Bennett, Igor Devetak, Aram~W. Harrow, Peter~W. Shor, and Andreas
  Winter.
\newblock The quantum reverse {Shannon} theorem and resource tradeoffs for
  simulating quantum channels.
\newblock {\em IEEE Transactions on Information Theory}, 60(5):2926--2959, May
  2014.
\newblock arXiv:0912.5537.

\bibitem[BDS14]{BDS14}
Francesco Buscemi, Nilanjana Datta, and Sergii Strelchuk.
\newblock Game-theoretic characterization of antidegradable channels.
\newblock {\em Journal of Mathematical Physics}, 55(9):092202, September 2014.
\newblock arXiv:1404.0277.

\bibitem[BDW18]{BDW18}
Stefan B\"auml, Siddhartha Das, and Mark~M. Wilde.
\newblock Fundamental limits on the capacities of bipartite quantum
  interactions.
\newblock {\em Physical Review Letters}, 121(25):250504, December 2018.
\newblock arXiv:1812.08223.

\bibitem[Bei13]{beigi2013sandwiched}
Salman Beigi.
\newblock Sandwiched {R{\'e}nyi} divergence satisfies data processing
  inequality.
\newblock {\em Journal of Mathematical Physics}, 54(12):122202, December 2013.
\newblock arXiv:1306.5920.

\bibitem[Ben05]{B05}
Charles~H. Bennett.
\newblock Simulated time travel, teleportation without communication, and how
  to conduct a romance with someone who has fallen into a black hole.
\newblock \url{https://www.research.ibm.com/people/b/bennetc/QUPONBshort.pdf},
  May 2005.

\bibitem[Ber13]{B13}
Mario Berta.
\newblock {\em Quantum Side Information: Uncertainty Relations, Extractors,
  Channel Simulations}.
\newblock PhD thesis, ETH Zurich, October 2013.
\newblock arXiv:1310.4581.

\bibitem[BG17]{BG17}
Francesco Buscemi and Gilad Gour.
\newblock Quantum relative {Lorenz} curves.
\newblock {\em Physical Review A}, 95(1):012110, January 2017.
\newblock arXiv:1607.05735.

\bibitem[BGMW17]{BenDana2017}
Khaled {Ben Dana}, Mar{\'{i}}a {Garc{\'{i}}a D{\'{i}}az}, Mohamed Mejatty, and
  Andreas Winter.
\newblock {Resource theory of coherence: Beyond states}.
\newblock {\em Physical Review A}, 95(6):062327, June 2017.
\newblock arXiv:1704.03710.

\bibitem[BGNP01]{PhysRevA.64.052309}
David Beckman, Daniel Gottesman, M.~A. Nielsen, and John Preskill.
\newblock Causal and localizable quantum operations.
\newblock {\em Physical Review A}, 64(5):052309, October 2001.
\newblock arXiv:quant-ph/0102043.

\bibitem[BHKW18]{BHKW18}
Mario Berta, Christoph Hirche, Eneet Kaur, and Mark~M. Wilde.
\newblock Amortized channel divergence for asymptotic quantum channel
  discrimination.
\newblock August 2018.
\newblock arXiv:1808.01498v1.

\bibitem[BHLS03]{BHLS03}
Charles~H. Bennett, Aram~W. Harrow, Debbie~W. Leung, and John~A. Smolin.
\newblock On the capacities of bipartite {Hamiltonians} and unitary gates.
\newblock {\em IEEE Transactions on Information Theory}, 49(8):1895--1911,
  August 2003.
\newblock arXiv:quant-ph/0205057.

\bibitem[Bla53]{B53}
David Blackwell.
\newblock Equivalent comparisons of experiments.
\newblock {\em The Annals of Mathematical Statistics}, 24(2):265--272, June
  1953.

\bibitem[BRW14]{BRW14}
Mario Berta, Joseph~M. Renes, and Mark~M. Wilde.
\newblock Identifying the information gain of a quantum measurement.
\newblock {\em IEEE Transactions on Information Theory}, 60(12):7987--8006,
  December 2014.
\newblock arXiv:1301.1594.

\bibitem[BSST02]{ieee2002bennett}
Charles~H. Bennett, Peter~W. Shor, John~A. Smolin, and Ashish~V. Thapliyal.
\newblock Entanglement-assisted capacity of a quantum channel and the reverse
  {Shannon} theorem.
\newblock {\em IEEE Transactions on Information Theory}, 48(10):2637--2655,
  October 2002.
\newblock arXiv:quant-ph/0106052.

\bibitem[Bus12]{Buscemi2012}
Francesco Buscemi.
\newblock Comparison of quantum statistical models: Equivalent conditions for
  sufficiency.
\newblock {\em Communications in Mathematical Physics}, 310(3):625--647, March
  2012.
\newblock arXiv:1004.3794.

\bibitem[Bus16]{Buscemi2016}
Francesco Buscemi.
\newblock Degradable channels, less noisy channels, and quantum statistical
  morphisms: An equivalence relation.
\newblock {\em Problems of Information Transmission}, 52(3):201--213, July
  2016.
\newblock arXiv:1511.08893.

\bibitem[Bus17]{B17}
Francesco Buscemi.
\newblock Comparison of noisy channels and reverse data-processing theorems.
\newblock {\em 2017 IEEE Information Theory Workshop}, pages 489--493, 2017.
\newblock arXiv:1803.02945.

\bibitem[CDP08a]{CDP08a}
Giulio Chiribella, Giacomo~Mauro D'Ariano, and Paolo Perinotti.
\newblock Memory effects in quantum channel discrimination.
\newblock {\em Physical Review Letters}, 101(18):180501, October 2008.
\newblock arXiv:0803.3237.

\bibitem[CDP08b]{CDP08}
Giulio Chiribella, Giacomo~Mauro D'Ariano, and Paolo Perinotti.
\newblock Transforming quantum operations: Quantum supermaps.
\newblock {\em Europhysics Letters}, 83(3):30004, August 2008.
\newblock arXiv:0804.0180.

\bibitem[CDP09]{CDP09}
Giulio Chiribella, Giacomo~Mauro D'Ariano, and Paolo Perinotti.
\newblock Theoretical framework for quantum networks.
\newblock {\em Physical Review A}, 80(2):022339, August 2009.
\newblock arXiv:0904.4483.

\bibitem[CE16]{Chiribella_2016}
Giulio Chiribella and Daniel Ebler.
\newblock Optimal quantum networks and one-shot entropies.
\newblock {\em New Journal of Physics}, 18(9):093053, September 2016.
\newblock arXiv:1606.02394.

\bibitem[CGLM14]{RevModPhys.86.1203}
Filippo Caruso, Vittorio Giovannetti, Cosmo Lupo, and Stefano Mancini.
\newblock Quantum channels and memory effects.
\newblock {\em Reviews of Modern Physics}, 86(4):1203--1259, December 2014.
\newblock arXiv:1207.5435.

\bibitem[CJW04]{CJW04}
Anthony Chefles, Richard Jozsa, and Andreas Winter.
\newblock On the existence of physical transformations between sets of quantum
  states.
\newblock {\em International Journal of Quantum Information}, 02(01):11--21,
  March 2004.
\newblock arXiv:quant-ph/0307227.

\bibitem[CMW16]{Cooney2016}
Tom Cooney, Mil{\'a}n Mosonyi, and Mark~M. Wilde.
\newblock Strong converse exponents for a quantum channel discrimination
  problem and quantum-feedback-assisted communication.
\newblock {\em Communications in Mathematical Physics}, 344(3):797--829, June
  2016.
\newblock arXiv:1408.3373.

\bibitem[Dat09]{D09}
Nilanjana Datta.
\newblock Min- and max-relative entropies and a new entanglement monotone.
\newblock {\em IEEE Transactions on Information Theory}, 55(6):2816--2826, June
  2009.
\newblock arXiv:0803.2770.

\bibitem[DBW17]{DBW17}
Siddhartha Das, Stefan B\"auml, and Mark~M. Wilde.
\newblock Entanglement and secret-key-agreement capacities of bipartite quantum
  interactions and read-only memory devices.
\newblock December 2017.
\newblock arXiv:1712.00827.

\bibitem[DDM14]{DM14}
Rafal Demkowicz-Dobrza{\'n}ski and Lorenzo Maccone.
\newblock Using entanglement against noise in quantum metrology.
\newblock {\em Physical Review Letters}, 113(25):250801, December 2014.
\newblock arXiv:1407.2934.

\bibitem[DFY09]{Duan09}
Runyao Duan, Yuan Feng, and Mingsheng Ying.
\newblock Perfect distinguishability of quantum operations.
\newblock {\em Physical Review Letters}, 103(21):210501, November 2009.
\newblock arXiv:0908.0119.

\bibitem[DP05]{DP05}
Giacomo~Mauro D'Ariano and Paolo Perinotti.
\newblock Programmable quantum channels and measurements.
\newblock In {\em Workshop on Quantum Information Theory and Quantum
  Statistical Inference}, Tokyo, ERATO Quantum Computation and Information
  Project, November 2005.
\newblock arXiv:quant-ph/0510033.

\bibitem[DW19]{DW17}
Siddhartha Das and Mark~M. Wilde.
\newblock Quantum reading capacity: General definition and bounds.
\newblock {\em IEEE Transactions on Information Theory}, 65(11):7566--7583,
  November 2019.
\newblock arXiv:1703.03706.

\bibitem[ESW02]{ESW02}
T.~Eggeling, D.~Schlingemann, and Reinhard~F. Werner.
\newblock Semicausal operations are semilocalizable.
\newblock {\em Europhysics Letters}, 57(6):782--788, March 2002.
\newblock arXiv:quant-ph/0104027.

\bibitem[Fai19]{Faist2019}
Philippe Faist.
\newblock Unpublished notes.
\newblock April 2019.

\bibitem[FBB19]{FBB18}
Philippe Faist, Mario Berta, and Fernando G. S.~L. Brandao.
\newblock Thermodynamic capacity of quantum processes.
\newblock {\em Physical Review Letters}, 122(20):200601, May 2019.
\newblock arXiv:1807.05610.

\bibitem[FFRS19]{FFRS19}
Kun Fang, Omar Fawzi, Renato Renner, and David Sutter.
\newblock A chain rule for the quantum relative entropy.
\newblock September 2019.
\newblock arXiv:1909.05826.

\bibitem[FL13]{FL13}
Rupert~L. Frank and Elliott~H. Lieb.
\newblock Monotonicity of a relative {R\'enyi} entropy.
\newblock {\em Journal of Mathematical Physics}, 54(12):122201, December 2013.
\newblock arXiv:1306.5358.

\bibitem[FWTB19]{FWTB18}
Kun Fang, Xin Wang, Marco Tomamichel, and Mario Berta.
\newblock Quantum channel simulation and the channel's smooth max-information.
\newblock {\em IEEE Transactions on Information Theory (in press)}, July 2019.
\newblock arXiv:1807.05354.

\bibitem[GFW{\etalchar{+}}18]{GFWRSCW18}
Mar{\'i}a {Garc{\'{\i}}a D{\'{\i}}az}, Kun {Fang}, Xin {Wang}, Matteo {Rosati},
  Michalis {Skotiniotis}, John {Calsamiglia}, and Andreas {Winter}.
\newblock Using and reusing coherence to realize quantum processes.
\newblock {\em Quantum}, 2:100, October 2018.
\newblock arXiv:1805.04045.

\bibitem[GJB{\etalchar{+}}18]{GJBDM18}
Gilad Gour, David Jennings, Francesco Buscemi, Runyao Duan, and Iman Marvian.
\newblock Quantum majorization and a complete set of entropic conditions for
  quantum thermodynamics.
\newblock {\em Nature Communications}, 9(1):5352, December 2018.
\newblock arXiv:1708.04302.

\bibitem[GLN05]{GLN04}
Alexei Gilchrist, Nathan~K. Langford, and Michael~A. Nielsen.
\newblock Distance measures to compare real and ideal quantum processes.
\newblock {\em Physical Review A}, 71(6):062310, June 2005.
\newblock arXiv:quant-ph/0408063.

\bibitem[Gou19]{G18}
Gilad Gour.
\newblock Comparison of quantum channels with superchannels.
\newblock {\em IEEE Transactions on Information Theory}, 65(9):5880--5904,
  September 2019.
\newblock arXiv:1808.02607.

\bibitem[GRS18]{Gutoski2018fidelityofquantum}
Gus Gutoski, Ansis Rosmanis, and Jamie Sikora.
\newblock Fidelity of quantum strategies with applications to cryptography.
\newblock {\em {Quantum}}, 2:89, September 2018.
\newblock arXiv:1704.04033.

\bibitem[Gut09]{G09}
Gus Gutoski.
\newblock {\em Quantum strategies and local operations}.
\newblock PhD thesis, University of Waterloo, 2009.
\newblock arXiv:1003.0038.

\bibitem[Gut12]{G12}
Gus Gutoski.
\newblock On a measure of distance for quantum strategies.
\newblock {\em Journal of Mathematical Physics}, 53(3):032202, March 2012.
\newblock arXiv:1008.4636.

\bibitem[GW07]{GW07}
Gus Gutoski and John Watrous.
\newblock Toward a general theory of quantum games.
\newblock {\em Proceedings of the thirty-ninth annual {ACM} symposium on theory
  of computing}, pages 565--574, 2007.
\newblock arXiv:quant-ph/0611234.

\bibitem[Hay03]{Hay03}
Masahito Hayashi.
\newblock Hypothesis testing approach to quantum information theory.
\newblock In {\em COE Symposium on Quantum Information Theory}, pages 15--16,
  Kyoto, Japan, 2003.

\bibitem[Hay04]{Hay04}
Masahito Hayashi.
\newblock Hypothesis testing approach to quantum information theory.
\newblock In {\em 1st Asia-Pacific Conference on Quantum Information Science},
  National Cheng Kung University, Tainan, Taiwan, 2004.

\bibitem[Hay09]{Hayashi09}
Masahito Hayashi.
\newblock Discrimination of two channels by adaptive methods and its
  application to quantum system.
\newblock {\em IEEE Transactions on Information Theory}, 55(8):3807--3820,
  August 2009.
\newblock arXiv:0804.0686.

\bibitem[Hay17]{Hay17}
Masahito Hayashi.
\newblock Role of hypothesis testing in quantum information theory.
\newblock In {\em Asian Conference on Quantum Information Science (AQIS 17)},
  National University of Singapore, Singapore, September 2017.
\newblock arXiv:1709.07701.

\bibitem[Hel69]{H69}
Carl~W. Helstrom.
\newblock Quantum detection and estimation theory.
\newblock {\em Journal of Statistical Physics}, 1:231--252, 1969.

\bibitem[Hel76]{Hel76}
Carl~W. Helstrom.
\newblock {\em Quantum Detection and Estimation Theory}.
\newblock Academic, New York, 1976.

\bibitem[HHLW10]{Harrow10}
Aram~W. Harrow, Avinatan Hassidim, Debbie~W. Leung, and John Watrous.
\newblock Adaptive versus nonadaptive strategies for quantum channel
  discrimination.
\newblock {\em Physical Review A}, 81(3):032339, March 2010.
\newblock arXiv:0909.0256.

\bibitem[HJRW12]{HJRW12}
Teiko Heinosaari, Maria~A. Jivulescu, David Reeb, and Michael~M. Wolf.
\newblock Extending quantum operations.
\newblock {\em Journal of Mathematical Physics}, 53(10):102208, October 2012.
\newblock arXiv:1205.0641.

\bibitem[HP91]{HP91}
Fumio Hiai and D\'enes Petz.
\newblock The proper formula for relative entropy and its asymptotics in
  quantum probability.
\newblock {\em Communications in Mathematical Physics}, 143(1):99--114,
  December 1991.

\bibitem[JN12]{JN12}
Rahul {Jain} and Ashwin {Nayak}.
\newblock Short proofs of the quantum substate theorem.
\newblock {\em IEEE Transactions on Information Theory}, 58(6):3664--3669, June
  2012.
\newblock arXiv:1103.6067.

\bibitem[JRS{\etalchar{+}}16]{JRSWW16}
Marius {Junge}, Renato {Renner}, David {Sutter}, Mark~M. {Wilde}, and Andreas
  {Winter}.
\newblock Universal recoverability in quantum information.
\newblock In {\em 2016 IEEE International Symposium on Information Theory
  (ISIT)}, pages 2494--2498, July 2016.

\bibitem[JWD{\etalchar{+}}08]{JWDFY08}
Zhengfeng Ji, Guoming Wang, Runyao Duan, Yuan Feng, and Mingsheng Ying.
\newblock Parameter estimation of quantum channels.
\newblock {\em IEEE Transactions on Information Theory}, 54(11):5172--5185,
  November 2008.
\newblock arXiv:quant-ph/0610060.

\bibitem[Kit97]{Kitaev1997}
Alexei~Yu Kitaev.
\newblock {Quantum computations: algorithms and error correction}.
\newblock {\em Russian Mathematical Surveys}, 52(6):1191--1249, December 1997.

\bibitem[KW04]{KWerner04}
Dennis Kretschmann and Reinhard~F. Werner.
\newblock Tema con variazioni: quantum channel capacity.
\newblock {\em New Journal of Physics}, 6(1):26, 2004.
\newblock arXiv:quant-ph/0311037.

\bibitem[KW05]{PhysRevA.72.062323}
Dennis Kretschmann and Reinhard~F. Werner.
\newblock Quantum channels with memory.
\newblock {\em Physical Review A}, 72(6):062323, December 2005.
\newblock arXiv:quant-ph/0502106.

\bibitem[KW17]{KW17}
Eneet Kaur and Mark~M. Wilde.
\newblock Upper bounds on secret key agreement over lossy thermal bosonic
  channels.
\newblock {\em Physical Review A}, 96(6):062318, December 2017.
\newblock arXiv:1706.04590.

\bibitem[LBL18]{LBL18}
Lu~Li, Kaifeng Bu, and Zi-Wen Liu.
\newblock Quantifying the resource content of quantum channels: An operational
  approach.
\newblock December 2018.
\newblock arXiv:1812.02572.

\bibitem[Li14]{li14}
Ke~Li.
\newblock Second order asymptotics for quantum hypothesis testing.
\newblock {\em Annals of Statistics}, 42(1):171--189, February 2014.
\newblock arXiv:1208.1400.

\bibitem[LKDW18]{LKDW18}
Felix Leditzky, Eneet Kaur, Nilanjana Datta, and Mark~M. Wilde.
\newblock Approaches for approximate additivity of the {Holevo} information of
  quantum channels.
\newblock {\em Physical Review A}, 97(1):012332, January 2018.
\newblock arXiv:1709.01111.

\bibitem[LW19]{LW19}
Zi-Wen Liu and Andreas Winter.
\newblock Resource theories of quantum channels and the universal role of
  resource erasure.
\newblock April 2019.
\newblock arXiv:1904.04201v1.

\bibitem[LY19]{LY19}
Yunchao Liu and Xiao Yuan.
\newblock Operational resource theory of quantum channels.
\newblock April 2019.
\newblock arXiv:1904.02680.

\bibitem[Mat10]{Mats10}
Keiji Matsumoto.
\newblock Reverse test and characterization of quantum relative entropy.
\newblock October 2010.
\newblock arXiv:1010.1030.

\bibitem[Mat11]{M11}
Keiji Matsumoto.
\newblock Reverse test and characterization of quantum relative entropy.
\newblock In {\em The Second Nagoya Winter Workshop on Quantum Information,
  Measurement, and Foundations}, 2011.
\newblock Slides available at
  \url{https://sites.google.com/site/nww2011/home/talks-slides/matsumoto.pdf}.

\bibitem[MLDS{\etalchar{+}}13]{muller2013quantum}
Martin M{\"u}ller-Lennert, Fr{\'e}d{\'e}ric Dupuis, Oleg Szehr, Serge Fehr, and
  Marco Tomamichel.
\newblock On quantum {R}{\'e}nyi entropies: a new generalization and some
  properties.
\newblock {\em Journal of Mathematical Physics}, 54(12):122203, December 2013.
\newblock arXiv:1306.3142.

\bibitem[MOA11]{MOA11}
Albert~W. Marshall, Ingram Olkin, and Barry~C. Arnold.
\newblock {\em Inequalities: Theory of Majorization and Its Applications}.
\newblock Springer, second edition, 2011.

\bibitem[MW14]{MW12}
William Matthews and Stephanie Wehner.
\newblock Finite blocklength converse bounds for quantum channels.
\newblock {\em IEEE Transactions on Information Theory}, 60(11):7317--7329,
  November 2014.
\newblock arXiv:1210.4722.

\bibitem[NC97]{NC97}
Michael~A. Nielsen and Isaac~L. Chuang.
\newblock Programmable quantum gate arrays.
\newblock {\em Physical Review Letters}, 79(2):321--324, July 1997.
\newblock arXiv:quant-ph/9703032.

\bibitem[ON00]{ON00}
Tomohiro Ogawa and Hiroshi Nagaoka.
\newblock Strong converse and {Stein's} lemma in quantum hypothesis testing.
\newblock {\em IEEE Transactions on Information Theory}, 46(7):2428--2433,
  November 2000.
\newblock arXiv:quant-ph/9906090.

\bibitem[Pet85]{P85}
D\'enes Petz.
\newblock Quasi-entropies for states of a von {Neumann} algebra.
\newblock {\em Publ. RIMS, Kyoto University}, 21:787--800, 1985.

\bibitem[Pet86]{P86}
D\'enes Petz.
\newblock Quasi-entropies for finite quantum systems.
\newblock {\em Reports in Mathematical Physics}, 23:57--65, 1986.

\bibitem[Ren16]{Ren16}
Joseph~M. Renes.
\newblock Relative submajorization and its use in quantum resource theories.
\newblock {\em Journal of Mathematical Physics}, 57(12):122202, December 2016.
\newblock arXiv:1510.03695.

\bibitem[RW05]{RW05}
Bill Rosgen and John Watrous.
\newblock On the hardness of distinguishing mixed-state quantum computations.
\newblock {\em Proceedings of the 20th IEEE Conference on Computational
  Complexity}, pages 344--354, June 2005.
\newblock arXiv:cs/0407056.

\bibitem[SC19]{Seddon2019}
James~R. Seddon and Earl Campbell.
\newblock {Quantifying magic for multi-qubit operations}.
\newblock {\em Proceedings of the Royal Society A}, 475(2227), July 2019.
\newblock arXiv:1901.03322.

\bibitem[Sio58]{sion58}
M.~Sion.
\newblock On general minimax theorems.
\newblock {\em Pacific Journal of Mathematics}, 8(1):171--176, 1958.

\bibitem[SW12]{SW12}
Naresh Sharma and Naqueeb~Ahmad Warsi.
\newblock On the strong converses for the quantum channel capacity theorems.
\newblock May 2012.
\newblock arXiv:1205.1712.

\bibitem[TCR09]{TCR09}
Marco Tomamichel, Roger Colbeck, and Renato Renner.
\newblock A fully quantum asymptotic equipartition property.
\newblock {\em IEEE Transactions on Information Theory}, 55(12):5840--5847,
  December 2009.
\newblock arXiv:0811.1221.

\bibitem[TEZP19]{TEZP19}
Thomas Theurer, Dario Egloff, Lijian Zhang, and Martin~B. Plenio.
\newblock Quantifying operations with an application to coherence.
\newblock {\em Physical Review Letters}, 122(19):190405, May 2019.
\newblock arXiv:1806.07332.

\bibitem[TH13]{tomamichel2013hierarchy}
Marco Tomamichel and Masahito Hayashi.
\newblock A hierarchy of information quantities for finite block length
  analysis of quantum tasks.
\newblock {\em IEEE Transactions on Information Theory}, 59(11):7693--7710,
  August 2013.
\newblock arXiv:1208.1478.

\bibitem[TR19]{TR19}
Ryuji Takagi and Bartosz Regula.
\newblock General resource theories in quantum mechanics and beyond:
  operational characterization via discrimination tasks.
\newblock January 2019.
\newblock arXiv:1901.08127.

\bibitem[TW16]{TW2016}
Masahiro Takeoka and Mark~M. Wilde.
\newblock Optimal estimation and discrimination of excess noise in thermal and
  amplifier channels.
\newblock November 2016.
\newblock arXiv:1611.09165.

\bibitem[Uhl76]{U76}
Armin Uhlmann.
\newblock The ``transition probability'' in the state space of a *-algebra.
\newblock {\em Reports on Mathematical Physics}, 9(2):273--279, 1976.

\bibitem[Ume62]{U62}
Hisaharu Umegaki.
\newblock Conditional expectations in an operator algebra {IV} (entropy and
  information).
\newblock {\em Kodai Mathematical Seminar Reports}, 14(2):59--85, 1962.

\bibitem[Wat09]{Wat09}
John Watrous.
\newblock Semidefinite programs for completely bounded norms.
\newblock {\em Theory of Computing}, 5(11):217--238, November 2009.
\newblock arXiv:0901.4709.

\bibitem[Wil17]{W17book}
Mark~M. Wilde.
\newblock {\em Quantum Information Theory}.
\newblock Cambridge University Press, March 2017.
\newblock \href{https://arxiv.org/abs/1106.1445}{arXiv:1106.1445v7}.

\bibitem[Wil18]{Wilde2018}
Mark~M. Wilde.
\newblock Entanglement cost and quantum channel simulation.
\newblock {\em Physical Review A}, 98(4):042320, October 2018.
\newblock arXiv:1807.11939.

\bibitem[WR12]{WR12}
Ligong Wang and Renato Renner.
\newblock One-shot classical-quantum capacity and hypothesis testing.
\newblock {\em Physical Review Letters}, 108(20):200501, May 2012.
\newblock arXiv:1007.5456.

\bibitem[WW18]{WW19PPT}
Xin Wang and Mark~M. Wilde.
\newblock Exact entanglement cost of quantum states and channels under
  {PPT}-preserving operations.
\newblock September 2018.
\newblock arXiv:1809.09592.

\bibitem[WW19]{WW19states}
Xin Wang and Mark~M. Wilde.
\newblock Resource theory of asymmetric distinguishability.
\newblock {\em Physical Review Research}, 1(3):033170, December 2019.
\newblock arXiv:1905.11629.

\bibitem[WWS19]{WWS19channels}
Xin Wang, Mark~M. Wilde, and Yuan Su.
\newblock Quantifying the magic of quantum channels.
\newblock {\em New Journal of Physics}, 21(103002), September 2019.
\newblock arXiv:1903.04483.

\bibitem[WWY14]{WWY14}
Mark~M. Wilde, Andreas Winter, and Dong Yang.
\newblock Strong converse for the classical capacity of entanglement-breaking
  and {H}adamard channels via a sandwiched {R}\'enyi relative entropy.
\newblock {\em Communications in Mathematical Physics}, 331(2):593--622,
  October 2014.
\newblock arXiv:1306.1586.

\bibitem[YLZ{\etalchar{+}}19]{YLZRTG19}
Xiao Yuan, Yunchao Liu, Qi~Zhao, Bartosz Regula, Jayne Thompson, and Mile Gu.
\newblock Robustness of quantum memories: An operational resource-theoretic
  approach.
\newblock July 2019.
\newblock arXiv:1907.02521.

\bibitem[Yua19]{Yuan2019}
Xiao Yuan.
\newblock Relative entropies of quantum channels with applications in resource
  theory.
\newblock {\em Physical Review A}, 99(3):032317, March 2019.
\newblock arXiv:1807.05958.

\end{thebibliography}

\appendix

\section{Background}

\subsection{Generalized divergences}

A generalized divergence is a function $\mathbf{D}(\rho\Vert\sigma)$ taking
arbitrary quantum states $\rho$ and $\sigma$ to the non-negative reals and such
that the data processing inequality holds for an arbitrary quantum channel $\mathcal{N}$
\cite{SW12}:%
\begin{equation}
\mathbf{D}(\rho\Vert\sigma)\geq\mathbf{D}(\mathcal{N}(\rho)\Vert
\mathcal{N}(\sigma)). \label{eq:DP-gen-div}%
\end{equation}
Generalized divergences of interest include the trace distance, the negative
logarithm of the fidelity \cite{U76}, the quantum relative entropy \cite{U62},
the Petz--R\'enyi relative entropy \cite{P85,P86}, and the sandwiched R\'enyi
relative entropy \cite{muller2013quantum,WWY14}.

For completeness, we define
the last three quantities now and refer to our companion paper
\cite{WW19states} for further details of their properties. The quantum
relative entropy $D(\rho\Vert\sigma)$ is defined for states $\rho$ and
$\sigma$ as
\begin{equation}
D(\rho\Vert\sigma) := \operatorname{Tr}[\rho(\log_{2} \rho- \log_{2} \sigma)]
\end{equation}
if $\operatorname{supp}(\rho)\subseteq\operatorname{supp}(\sigma)$ and it is
set to $\infty$ otherwise. The Petz--R\'{e}nyi relative entropy is defined for
states $\rho$ and $\sigma$ as \cite{P86}%
\begin{align}
D_{\alpha}(\rho\Vert\sigma)  &  :=\frac{1}{\alpha-1}\log_{2}\operatorname{Tr}%
[\rho^{\alpha}\sigma^{1-\alpha}]
\end{align}
if $\alpha\in(0,1)$ or $\alpha\in(1,\infty)$ and $\operatorname{supp}%
(\rho)\subseteq\operatorname{supp}(\sigma)$. If $\alpha\in(1,\infty)$ and
$\operatorname{supp}(\rho)\not \subseteq \operatorname{supp}(\sigma)$, then
$D_{\alpha}(\rho\Vert\sigma):=\infty$ \cite{TCR09}. The sandwiched R\'enyi
relative entropy is defined for states $\rho$ and $\sigma$ as
\cite{muller2013quantum,WWY14}%
\begin{align}
\widetilde{D}_{\alpha}(\rho\Vert\sigma)  &  :=\frac{1}{\alpha-1}\log
_{2}\operatorname{Tr}[(\sigma^{(1-\alpha)/2\alpha}\rho\sigma^{(1-\alpha
)/2\alpha})^{\alpha}]
\end{align}
if $\alpha\in(0,1)$ or $\alpha\in(1,\infty)$ and $\operatorname{supp}%
(\rho)\subseteq\operatorname{supp}(\sigma)$. If $\alpha\in(1,\infty)$ and
$\operatorname{supp}(\rho)\not \subseteq \operatorname{supp}(\sigma)$, then
$\widetilde{D}_{\alpha}(\rho\Vert\sigma):=\infty$.

A generalized channel divergence is defined from that for states, as presented
above, given by the following function of quantum channels $\mathcal{N}%
_{A\rightarrow B}$ and $\mathcal{M}_{A\rightarrow B}$ \cite{LKDW18}:%
\begin{equation}
\mathbf{D}(\mathcal{N}\Vert\mathcal{M}):=\sup_{\rho_{RA}}\mathbf{D}%
(\mathcal{N}_{A\rightarrow B}(\rho_{RA})\Vert\mathcal{M}_{A\rightarrow B}%
(\rho_{RA})),
\label{eq:gen-ch-div-def}
\end{equation}
where the optimization is with respect to a quantum state $\rho_{RA}$ such
that the reference system is arbitrary. As observed in \cite{LKDW18}, the
following simplification holds%
\begin{equation}
\mathbf{D}(\mathcal{N}\Vert\mathcal{M})=\sup_{\psi_{RA}}\mathbf{D}%
(\mathcal{N}_{A\rightarrow B}(\psi_{RA})\Vert\mathcal{M}_{A\rightarrow B}%
(\psi_{RA})),
\end{equation}
where the optimization is with respect to pure states $\psi_{RA}$ such that
the reference system $R$ is isomorphic to the channel input system $A$.

The data processing inequality holds for the generalized channel divergence,
with respect to a superchannel~$\Theta$:%
\begin{equation}
\mathbf{D}(\mathcal{N}\Vert\mathcal{M})\geq\mathbf{D}(\Theta(\mathcal{N}%
)\Vert\Theta(\mathcal{M})), \label{eq:DP-superch}%
\end{equation}
as proved in \cite{G18}. The inequality in \eqref{eq:DP-superch} follows from
the definition in \eqref{eq:gen-ch-div-def} and the fact that the underlying generalized divergence
$\mathbf{D}$ obeys the data processing inequality in \eqref{eq:DP-gen-div}. Other applications and interpretations of channel divergences were considered in \cite{Yuan2019}.

For an environment-parametrized channel box $(\mathcal{N},\mathcal{M})$ with
environment states $\rho_{E}$ and $\sigma_{E}$, the following inequality holds
\cite{TW2016}%
\begin{equation}
\mathbf{D}(\mathcal{N}\Vert\mathcal{M})\leq\mathbf{D}(\rho_{E}\Vert\sigma
_{E}).
\end{equation}
If the channel box is also environment seizable (see
Section~\ref{sec:env-param-seize}), then the opposite inequality
$\mathbf{D}(\mathcal{N}\Vert\mathcal{M})\geq\mathbf{D}(\rho_{E}\Vert\sigma
_{E})$ holds as well (as a consequence of \eqref{eq:DP-superch}), from which
we conclude the following equality in this case:%
\begin{equation}
\mathbf{D}(\mathcal{N}\Vert\mathcal{M})=\mathbf{D}(\rho_{E}\Vert\sigma_{E}).
\label{eq:env-seizable-gen-div-equal}%
\end{equation}

Particular examples of generalized channel divergences are the channel
min-relative entropy in \eqref{eq:ch-min-rel-ent}, the smooth channel
min-relative entropy in \eqref{eq:ch-HT-rel-ent}, and the channel max-relative
entropy in \eqref{eq:ch-max-rel-ent}. Other examples include those built from
the relative entropy, the Petz--R\'enyi relative entropy, and the sandwiched
R\'enyi relative entropy, as defined in \cite{Cooney2016}. As such, the
inequality in \eqref{eq:DP-superch} holds for all of these channel
divergences, a property that we make extensive use of in what follows.

It is not clear how to write the smooth channel max-relative entropy in
\eqref{eq:ch-smooth-max-rel-ent} as a generalized channel divergence. However,
it does obey the data processing inequality in \eqref{eq:DP-superch}, as the
following simple argument demonstrates. Let $\mathcal{N}$ and $\mathcal{M}$ be
arbitrary channels, and let $\Theta$ be a superchannel. Let $\widetilde
{\mathcal{N}}$ be a channel satisfying $\widetilde{\mathcal{N}}\approx
_{\varepsilon}\mathcal{N}$. Then, from the data processing inequality for the
diamond distance with respect to superchannels, it follows that $\Theta
(\widetilde{\mathcal{N}})\approx_{\varepsilon}\Theta(\mathcal{N})$. We then
have that%
\begin{align}
D_{\max}(\widetilde{\mathcal{N}}\Vert\mathcal{M})  &  \geq D_{\max}%
(\Theta(\widetilde{\mathcal{N}})\Vert\Theta(\mathcal{M}))\\
&  \geq D_{\max}^{\varepsilon}(\Theta(\mathcal{N})\Vert\Theta(\mathcal{M})).
\end{align}
The first inequality follows from the data processing inequality for $D_{\max
}$ of channels, and the second follows from the definition of the 
smooth channel max-relative entropy and the fact that $\Theta(\widetilde{\mathcal{N}%
})\approx_{\varepsilon}\Theta(\mathcal{N})$. Since the inequality holds for
all $\widetilde{\mathcal{N}}$ satisfying $\widetilde{\mathcal{N}}%
\approx_{\varepsilon}\mathcal{N}$, we conclude the desired data processing
inequality:%
\begin{equation}
D_{\max}^{\varepsilon}(\mathcal{N}\Vert\mathcal{M})\geq D_{\max}^{\varepsilon
}(\Theta(\mathcal{N})\Vert\Theta(\mathcal{M})).
\end{equation}

\subsection{Choi isomorphism for quantum channels}

\label{app:choi-iso-ch}The Choi isomorphism is a way of characterizing quantum
channels that is suitable for optimizing over them in semi-definite programs.
For a quantum channel $\mathcal{N}_{A\rightarrow B}$, its Choi operator is
given by%
\begin{equation}
\Gamma_{RB}^{\mathcal{N}}:=\mathcal{N}_{A\rightarrow B}(\Gamma_{RA}),
\end{equation}
where $\Gamma_{RA}=|\Gamma\rangle\langle\Gamma|_{RA}$ and%
\begin{equation}
|\Gamma\rangle_{RA}:=\sum_{i}|i\rangle_{R}|i\rangle_{A},
\end{equation}
with $\{|i\rangle_{R}\}_{i}$ and $\{|i\rangle_{A}\}$ orthonormal bases. The
Choi operator is positive semi-definite $\Gamma_{RB}^{\mathcal{N}}\geq0$,
corresponding to $\mathcal{N}_{A\rightarrow B}$ being completely positive, and
satisfies $\operatorname{Tr}_{B}[\Gamma_{RB}^{\mathcal{N}}]=I_{R}$, the latter
corresponding to $\mathcal{N}_{A\rightarrow B}$ being trace preserving. On the
other hand, given an operator $\Gamma_{RB}^{\mathcal{M}}$ satisfying
$\Gamma_{RB}^{\mathcal{M}}\geq0$ and $\operatorname{Tr}_{B}[\Gamma
_{RB}^{\mathcal{M}}]=I_{R}$, one realizes via postselected teleportation
\cite{B05}\ the following quantum channel:%
\begin{align}
\mathcal{M}_{A\rightarrow B}(\rho_{A})  &  =\langle\Gamma|_{SR}\left(
\rho_{S}\otimes\Gamma_{RB}^{\mathcal{M}}\right)  |\Gamma\rangle_{SR}\\
&  =\operatorname{Tr}_{R}[(T_{R}(\rho_{R})\otimes I_{B})\Gamma_{RB}%
^{\mathcal{M}}], \label{eq:output-from-Choi-oper}%
\end{align}
where systems $S$, $R$, and $A$ are isomorphic and the last line employs the
facts that $\left(  M_{S}\otimes I_{R}\right)  |\Gamma\rangle_{SR}=\left(
I_{S}\otimes T_{R}(M_{R})\right)  |\Gamma\rangle_{SR}$ for $T_{R}$ the
transpose map, defined as%
\begin{equation}
T_{R}(\rho_{R})=\sum_{i,j}|i\rangle\langle j|_{R}\rho_{R}|i\rangle\langle
j|_{R},
\end{equation}
and $\langle\Gamma|_{SR}\left(  I_{S}\otimes X_{RB}\right)  |\Gamma
\rangle_{SR}=\operatorname{Tr}_{R}[X_{RB}]$. We often abbreviate the transpose
map simply as%
\begin{equation}
\rho_{R}^{T}=T_{R}(\rho_{R}).
\end{equation}
Since the constraints $\Gamma_{RB}^{\mathcal{M}}\geq0$ and $\operatorname{Tr}%
_{B}[\Gamma_{RB}^{\mathcal{M}}]=I_{R}$ are semi-definite, this is a useful way
of incorporating optimizations over quantum channels into semi-definite programs.

\subsection{Semi-definite programs for diamond distance}

\label{sec:SDP-diamond-dist}The normalized diamond distance between quantum
channels $\mathcal{N}_{A\rightarrow B}$ and $\mathcal{M}_{A\rightarrow B}$ is
given by the following primal and dual semi-definite programs \cite{Wat09}:%
\begin{align}
&  \frac{1}{2}\left\Vert \mathcal{N}-\mathcal{M}\right\Vert _{\diamond
}\nonumber\\
&  =\sup_{\rho_{R},\Omega_{RB}\geq0}\left\{
\begin{array}
[c]{c}%
\operatorname{Tr}[\Omega_{RB}(\Gamma_{RB}^{\mathcal{N}}-\Gamma_{RB}%
^{\mathcal{M}})]:\\
\Omega_{RB}\leq\rho_{R}\otimes I_{B},\ \operatorname{Tr}[\rho_{R}]=1
\end{array}
\right\} \label{eq:SDP-diamond-dist-primal}\\
&  =\inf_{\mu,Z_{RB}\geq0}\left\{
\begin{array}
[c]{c}%
\mu:Z_{RB}\geq\Gamma_{RB}^{\mathcal{N}}-\Gamma_{RB}^{\mathcal{M}},\\
\mu I_{R}\geq\operatorname{Tr}_{B}[Z_{RB}]
\end{array}
\right\}  . \label{eq:SDP-diamond-dist-dual}%
\end{align}
The latter expression is equal to%
\begin{equation}
\inf_{Z_{RB}\geq0}\left\{  \left\Vert \operatorname{Tr}_{B}[Z_{RB}]\right\Vert
_{\infty}:Z_{RB}\geq\Gamma_{RB}^{\mathcal{N}}-\Gamma_{RB}^{\mathcal{M}%
}\right\}  .
\end{equation}

\subsection{Choi isomorphism for quantum superchannels}

\label{app:Choi-rep-super-ch}Just as there is a Choi isomorphism for quantum
channels, as reviewed in Appendix~\ref{app:choi-iso-ch}, there is a Choi
isomorphism for quantum superchannels \cite{CDP08,G18}. To define it, we can
exploit the known result that a quantum superchannel $\Theta_{\left(
A\rightarrow B\right)  \rightarrow\left(  C\rightarrow D\right)  }$\ is in
one-to-one correspondence with a bipartite channel $\mathcal{L}_{CB\rightarrow
AD}$ that has no-signaling constraints \cite{CDP08,G18}. That is, as stated in
\eqref{eq:pre-post-proc-superch}, every superchannel $\Theta_{\left(
A\rightarrow B\right)  \rightarrow\left(  C\rightarrow D\right)  }$ can be
physically realized by means of pre- and post-processing channels
$\mathcal{E}_{C\rightarrow AM}$ and $\mathcal{D}_{BM\rightarrow D}$,
respectively, such that \eqref{eq:pre-post-proc-superch}\ holds. The bipartite
channel corresponding to $\mathcal{E}_{C\rightarrow AM}$ and $\mathcal{D}%
_{BM\rightarrow D}$ is then given by%
\begin{equation}
\mathcal{L}_{CB\rightarrow AD}=\mathcal{D}_{BM\rightarrow D}\circ
\mathcal{E}_{C\rightarrow AM}, \label{eq:bipartite-ch-no-signaling}%
\end{equation}
i.e., where we do not \textquotedblleft plug in\textquotedblright\ the channel
$\mathcal{N}_{A\rightarrow B}$ to the ports $A$ and $B$, and instead system
$B$ is available as input and system $A$ is available as output. On the other
hand, suppose that $\mathcal{L}_{CB\rightarrow AD}$ is a bipartite channel
with the constraint that it is no-signaling from input system $B$ to output
system $A$. Then there exist channels $\mathcal{E}_{C\rightarrow AM}$ and
$\mathcal{D}_{BM\rightarrow D}$ such that $\mathcal{L}_{CB\rightarrow AD}$ can
be realized as in \eqref{eq:bipartite-ch-no-signaling}, as proved in
\cite{ESW02}, placing superchannels in one-to-one correspondence with
bipartite channels that have a no-signaling constraint.

Using this correspondence, we define the Choi operator of a superchannel
$\Theta_{\left(  A\rightarrow B\right)  \rightarrow\left(  C\rightarrow
D\right)  }$ with corresponding $B\not \rightarrow A$ no-signaling\ bipartite
channel $\mathcal{L}_{CB\rightarrow AD}$ as%
\begin{equation}
\Gamma_{R_{C}R_{B}AD}^{\Theta}:=\mathcal{L}_{CB\rightarrow AD}(\Gamma_{R_{C}%
C}\otimes\Gamma_{R_{B}B}).
\end{equation}
The fact that $\Theta_{\left(  A\rightarrow B\right)  \rightarrow\left(
C\rightarrow D\right)  }$ preserves completely positivity corresponds to the
condition $\Gamma_{R_{C}R_{B}AD}^{\Theta}\geq0$, and the fact that
$\Theta_{\left(  A\rightarrow B\right)  \rightarrow\left(  C\rightarrow
D\right)  }$ preserves trace preservation corresponds to the condition
$\Gamma_{R_{C}R_{B}}^{\Theta}=I_{R_{C}R_{B}}$. The no-signaling condition
corresponds to $\Gamma_{R_{C}R_{B}A}^{\Theta}=\Gamma_{R_{C}A}^{\Theta}%
\otimes\pi_{R_{B}}$, where $\pi_{R_{B}}$ is the maximally mixed state.
Furthermore, as an extension of \eqref{eq:output-from-Choi-oper}, the Choi
operator $\Gamma_{R_{C}D}^{\mathcal{K}}$\ for the output channel
$\mathcal{K}_{C\rightarrow D}$\ of the superchannel $\Theta_{\left(
A\rightarrow B\right)  \rightarrow\left(  C\rightarrow D\right)  }$, when the
input is a channel $\mathcal{N}_{A\rightarrow B}$\ with Choi operator
$\Gamma_{R_{A}B}^{\mathcal{N}}$, is as follows:%
\begin{equation}
\Gamma_{R_{C}D}^{\mathcal{K}}=\operatorname{Tr}_{R_{A}B}[(T_{R_{A}B}%
(\Gamma_{R_{A}B}^{\mathcal{N}})\otimes I_{R_{C}D})\Gamma_{R_{C}R_{B}%
AD}^{\Theta}].
\end{equation}
This kind of formulation of a superchannel allows for incorporating
optimizations over superchannels into semi-definite programs, as we do in
Appendix~\ref{app:gen-ch-box-trans-SDP}.

\section{General channel box transformation problem as a semi-definite
program}

\label{app:gen-ch-box-trans-SDP}Here we prove the statement claimed at the end
of Section~\ref{sec:ch-box-tr-problem}, that the general channel box
transformation problem stated in \eqref{eq:approx-box-trans-prob}\ can be
solved by means of a semi-definite program. By employing the Choi
representation of superchannels from Appendix~\ref{app:Choi-rep-super-ch}, as
well as the semi-definite program for the diamond distance in
Appendix~\ref{sec:SDP-diamond-dist}, we find that
\eqref{eq:approx-box-trans-prob}, as a function of channels $\mathcal{N}%
_{A\rightarrow B}$, $\mathcal{M}_{A\rightarrow B}$, $\mathcal{K}_{C\rightarrow
D}$, and $\mathcal{L}_{C\rightarrow D}$,\ can be written as the following
semi-definite program:%
\begin{equation}
\inf_{Z_{CD},\ \Gamma_{CBAD}^{\Theta}\geq0}\left\Vert \operatorname{Tr}%
_{D}[Z_{CD}]\right\Vert _{\infty}, \label{eq:SDP-primal-box-tr-1}%
\end{equation}
subject to%
\begin{align}
Z_{CD}  &  \geq\Gamma_{CD}^{\mathcal{K}}-\operatorname{Tr}_{AB}[(\Gamma
_{AB}^{\mathcal{N}})^{T}\Gamma_{CBAD}^{\Theta}],\nonumber\\
\Gamma_{CD}^{\mathcal{L}}  &  =\operatorname{Tr}_{AB}[(\Gamma_{AB}%
^{\mathcal{M}})^{T}\Gamma_{CBAD}^{\Theta}],\nonumber\\
\Gamma_{CB}^{\Theta}  &  =I_{CB},\nonumber\\
\Gamma_{CBA}^{\Theta}  &  =\Gamma_{CA}^{\Theta}\otimes I_B/\left\vert B\right\vert ,
\label{eq:SDP-primal-box-tr-2}%
\end{align}
where we employ the shorthand%
\begin{align}
\Gamma_{CD}^{\mathcal{K}}  &  :=\mathcal{K}_{C^{\prime}\rightarrow D}%
(\Gamma_{CC^{\prime}}),\\
\Gamma_{AB}^{\mathcal{N}}  &  :=\mathcal{N}_{A^{\prime}\rightarrow B}%
(\Gamma_{AA^{\prime}}),\\
\Gamma_{AB}^{\mathcal{M}}  &  :=\mathcal{M}_{A^{\prime}\rightarrow B}%
(\Gamma_{AA^{\prime}}),
\end{align}
with system $C^{\prime}$ isomorphic to system $C$ and system $A^{\prime}$
isomorphic to system $A$.

The dual of the semi-definite program in
\eqref{eq:SDP-primal-box-tr-1}--\eqref{eq:SDP-primal-box-tr-2} is given by%
\begin{equation}
\sup\operatorname{Tr}[\Gamma_{CD}^{\mathcal{K}}Y_{CD}]+\operatorname{Tr}%
[\Gamma_{CD}^{\mathcal{L}}W_{CD}]+\operatorname{Tr}[S_{CB}],
\label{eq:SDP-dual-ch-box-trans}%
\end{equation}
subject to%
\begin{align}
Y_{CD},M_{C}  &  \geq0,\\
W_{CD},S_{CB},L_{CBA}  &  \in\text{Herm},
\end{align}
\vspace{-.3in}%
\begin{equation}
\operatorname{Tr}[M_{C}]\leq1,\qquad Y_{CD}\leq M_{C}\otimes I_{D},
\end{equation}
\vspace{-.3in}%
\begin{multline}
I_{BD}\otimes\operatorname{Tr}_{B}[L_{CBA}]/\left\vert B\right\vert \geq
Y_{CD}\otimes(\Gamma_{AB}^{\mathcal{N}})^{T}\\
+W_{CD}\otimes(\Gamma_{AB}^{\mathcal{M}})^{T}+S_{CB}\otimes I_{AD}\\
+L_{CBA}\otimes I_{D}.
\end{multline}
By employing strong duality, it follows that the optimal value of
\eqref{eq:SDP-primal-box-tr-1}--\eqref{eq:SDP-primal-box-tr-2}\ is equal to
the optimal value of \eqref{eq:SDP-dual-ch-box-trans}.

\section{One-shot distillation and dilution of channel boxes}

\subsection{Channel min-relative entropy as exact one-shot distillable
distinguishability}

\label{app:ch-min-rel-ent-exact-distillable-dist}To establish
\eqref{eq:exact-distill-dist-min}, we first prove the inequality%
\begin{equation}
D_{d}^{0}(\mathcal{N},\mathcal{M})\geq D_{\min}(\mathcal{N}\Vert\mathcal{M}).
\label{eq:exact-distill-dist-lower-bnd-dmin}%
\end{equation}
Let $\Theta_{\left(  A\rightarrow B\right)  \rightarrow\left(  C\rightarrow
D\right)  }$ be the superchannel that traces out the input $C$, prepares the
pure state $\psi_{RA}$, transmits $A$ through the unknown channel
$\mathcal{N}$ or $\mathcal{M}$, and then applies the following channel to
systems $R$ and $B$, where $B$ is the output of the unknown channel:%
\begin{multline}
\omega_{RB}\rightarrow\operatorname{Tr}[\Pi_{RB}^{\mathcal{N}(\psi)}%
\omega_{RB}]|0\rangle\langle0|_{D}\label{eq:exact-distill-channel}\\
+\operatorname{Tr}[(I_{RB}-\Pi_{RB}^{\mathcal{N}(\psi)})\omega_{RB}%
]|1\rangle\langle1|_{D},
\end{multline}
and $\Pi_{RB}^{\mathcal{N}(\psi)}$ is the projection onto the support of the
state $\mathcal{N}_{A\rightarrow B}(\psi_{RA})$. By construction, if the
unknown channel is $\mathcal{N}_{A\rightarrow B}$, then the channel realized
by the superchannel delineated above is the replacer channel $\mathcal{R}%
_{C\rightarrow D}^{|0\rangle\langle0|}$. On the other hand, note that if
$\omega_{RB}=\mathcal{M}_{A\rightarrow B}(\psi_{RA})$ is the input to the
channel in \eqref{eq:exact-distill-channel}, then the output is the state
$\pi_{M}$, where%
\begin{equation}
\log_{2}M=D_{\min}(\mathcal{N}_{A\rightarrow B}(\psi_{RA})\Vert\mathcal{M}%
_{A\rightarrow B}(\psi_{RA})).
\end{equation}
So the output in this latter case is the replacer channel $\mathcal{R}%
_{C\rightarrow D}^{\pi_{M}}$. Taking a supremum over all input states
$\psi_{RA}$ then establishes the inequality in \eqref{eq:exact-distill-dist-lower-bnd-dmin}.

The opposite inequality%
\begin{equation}
D_{\min}(\mathcal{N}\Vert\mathcal{M})\geq D_{d}^{0}(\mathcal{N},\mathcal{M})
\label{eq:exact-distill-dist-upper-bnd-dmin}%
\end{equation}
follows from the data processing inequality for $D_{\min}(\mathcal{N}%
\Vert\mathcal{M})$ under the action of a superchannel. Let $\Theta$ be an
arbitrary superchannel satisfying%
\begin{align}
\Theta(\mathcal{N}_{A\rightarrow B})  &  =\mathcal{R}_{C\rightarrow
D}^{|0\rangle\langle0|},\\
\Theta(\mathcal{M}_{A\rightarrow B})  &  =\mathcal{R}_{C\rightarrow D}%
^{\pi_{M}}.
\end{align}
Then it follows from \eqref{eq:DP-superch} that%
\begin{align}
D_{\min}(\mathcal{N}\Vert\mathcal{M})  &  \geq D_{\min}(\Theta(\mathcal{N}%
)\Vert\Theta(\mathcal{M}))\\
&  =D_{\min}(\mathcal{R}_{C\rightarrow D}^{|0\rangle\langle0|}\Vert
\mathcal{R}_{C\rightarrow D}^{\pi_{M}})\\
&  = D_{\min}(|0\rangle\langle0|\Vert\pi_{M})\\
&  =\log_{2}M,
\end{align}
where the second-to-last equality follows from \eqref{eq:env-seizable-gen-div-equal}, given
that pairs of replacer channels are environment seizable, and the last equality follows by direct evaluation. Since the
exact distillable distinguishability involves an optimization over all
superchannels, the inequality in \eqref{eq:exact-distill-dist-upper-bnd-dmin}
follows, and combined with \eqref{eq:exact-distill-dist-lower-bnd-dmin}, we
conclude \eqref{eq:exact-distill-dist-min}.

\subsection{Channel max-relative entropy as exact one-shot distinguishability
cost}

\label{app:ch-max-rel-ent-exact-dist-cost}

To establish
\eqref{eq:exact-dist-cost-max}, we first prove the inequality%
\begin{equation}
D_{c}^{0}(\mathcal{N},\mathcal{M})\leq D_{\max}(\mathcal{N}\Vert\mathcal{M}).
\label{eq:exact-dist-cost-ach-part}%
\end{equation}
Recall the characterization of $D_{\max}(\mathcal{N}\Vert\mathcal{M})$ from
\eqref{eq:dmax-simplify}. Let $\lambda$ be such that%
\begin{equation}
\mathcal{N}_{A\rightarrow B}(\Phi_{RA})\leq2^{\lambda}\mathcal{M}%
_{A\rightarrow B}(\Phi_{RA}). \label{eq:dmax-ch-condition}%
\end{equation}
Then this means that $2^{\lambda}\mathcal{M}_{A\rightarrow B}(\Phi
_{RA})-\mathcal{N}_{A\rightarrow B}(\Phi_{RA})\geq0$, so that%
\begin{equation}
\omega_{RA}:=\frac{2^{\lambda}\mathcal{M}_{A\rightarrow B}(\Phi_{RA}%
)-\mathcal{N}_{A\rightarrow B}(\Phi_{RA})}{2^{\lambda}-1}%
\end{equation}
is a quantum state. Furthermore, since%
\begin{equation}
\operatorname{Tr}_{B}\left[  \frac{2^{\lambda}\mathcal{M}_{A\rightarrow
B}(\Phi_{RA})-\mathcal{N}_{A\rightarrow B}(\Phi_{RA})}{2^{\lambda}-1}\right]
=\pi_{R},
\end{equation}
where $\pi_{R}$ is the maximally mixed state on system $R$, it follows that
$\omega_{RA}$ is the Choi state of a quantum channel $\mathcal{N}%
_{A\rightarrow B}^{\prime}$, so that%
\begin{equation}
\omega_{RA}=\mathcal{N}_{A\rightarrow B}^{\prime}(\Phi_{RA}).
\end{equation}
Furthermore, by linearity, we have that%
\begin{equation}
\mathcal{N}_{A\rightarrow B}^{\prime}=\frac{2^{\lambda}\mathcal{M}%
_{A\rightarrow B}-\mathcal{N}_{A\rightarrow B}}{2^{\lambda}-1}.
\end{equation}

Then we construct the superchannel $\Theta_{\left(  C\rightarrow D\right)
\rightarrow\left(  A\rightarrow B\right)  }$ as follows. Let $\tau_{C}$ be a
fixed state that is input to the unknown replacer channel $\mathcal{R}%
_{C\rightarrow D}^{|0\rangle\langle0|}$ or $\mathcal{R}_{C\rightarrow D}%
^{\pi_{M}}$, where $M=2^{\lambda}$. Then we perform the following channel on
the output system $D$ and the input system $A$:%
\begin{multline}
\mathcal{P}_{AD}(\rho_{A}\otimes\sigma_{D}):=\mathcal{N}_{A\rightarrow B}%
(\rho_{A})\langle0|\sigma_{D}|0\rangle_{D}\\
+\mathcal{N}_{A\rightarrow B}^{\prime}(\rho_{A})\langle1|\sigma_{D}%
|1\rangle_{D}.
\end{multline}
In the case that the unknown channel is $\mathcal{R}_{C\rightarrow
D}^{|0\rangle\langle0|}$, the channel realized by this process is
$\mathcal{N}_{A\rightarrow B}$. In the case that the unknown channel is
$\mathcal{R}_{C\rightarrow D}^{\pi_{M}}$, the channel realized by this process
is%
\begin{equation}
2^{-\lambda}\mathcal{N}_{A\rightarrow B}+\left(  1-2^{-\lambda}\right)
\mathcal{N}_{A\rightarrow B}^{\prime}=\mathcal{M}_{A\rightarrow B},
\end{equation}
demonstrating that%
\begin{equation}
D_{c}^{0}(\mathcal{N},\mathcal{M})\leq\lambda.
\end{equation}
Now taking an infimum over all $\lambda$ such that
\eqref{eq:dmax-ch-condition} holds, we conclude the inequality in \eqref{eq:exact-dist-cost-ach-part}.

The opposite inequality%
\begin{equation}
D_{\max}(\mathcal{N}\Vert\mathcal{M})\leq D_{c}^{0}(\mathcal{N},\mathcal{M})
\label{eq:dist-cost-one-shot-optimality}%
\end{equation}
follows from the data processing inequality for $D_{\max}(\mathcal{N}%
\Vert\mathcal{M})$ under the action of a superchannel. Let $\Theta$ be an
arbitrary superchannel satisfying%
\begin{align}
\Theta(\mathcal{R}_{C\rightarrow D}^{|0\rangle\langle0|})  &  =\mathcal{N}%
_{A\rightarrow B},\\
\Theta(\mathcal{R}_{C\rightarrow D}^{\pi_{M}})  &  =\mathcal{M}_{A\rightarrow
B}.
\end{align}
Then it follows from \eqref{eq:DP-superch} that%
\begin{align}
\log_{2}M  &  = D_{\max}(|0\rangle\langle0| \Vert\pi_{M})\\
&  = D_{\max}(\mathcal{R}_{C\rightarrow D}^{|0\rangle\langle0|}\Vert
\mathcal{R}_{C\rightarrow D}^{\pi_{M}})\\
&  \geq D_{\max}(\Theta(\mathcal{R}_{C\rightarrow D}^{|0\rangle\langle
0|})\Vert\Theta(\mathcal{R}_{C\rightarrow D}^{\pi_{M}}))\\
&  =D_{\max}(\mathcal{N}\Vert\mathcal{M}),
\end{align}
The first equality follows by direct evaluation, and the second follows
from \eqref{eq:env-seizable-gen-div-equal}, given
that pairs of replacer channels are environment seizable. Since the exact
distinguishability cost involves an optimization over all
superchannels, the inequality in \eqref{eq:dist-cost-one-shot-optimality}
follows, and combined with \eqref{eq:exact-dist-cost-ach-part}, we conclude \eqref{eq:exact-dist-cost-max}.

\subsection{Semi-definite programs for smooth channel min- and max-relative
entropies}

\label{app:SDPs-channel-smooth-entropies}

In this appendix, we prove that the smooth channel min- and max-relative
entropies are characterized by semi-definite programs, starting with the
former. We note that Proposition~\ref{prop:smooth-min-ch-rel-ent-SDP}\ below
was also found in \cite{Faist2019}.

\begin{proposition}
\label{prop:smooth-min-ch-rel-ent-SDP}Let $\mathcal{N}_{A\rightarrow B}$ and
$\mathcal{M}_{A\rightarrow B}$ be quantum channels and $\varepsilon\in\left[
0,1\right]  $. The smooth channel min-relative entropy is given by the
following primal semi-definite program:%
\begin{multline}
D_{\min}^{\varepsilon}(\mathcal{N}\Vert\mathcal{M}%
)=\label{eq:primal-SDP-ch-sm-min-rel-ent}\\
-\log_{2}\inf_{\rho_{R},\Omega_{RB}\geq0}\left\{
\begin{array}
[c]{c}%
\operatorname{Tr}[\Omega_{RB}\Gamma_{RB}^{\mathcal{M}}]:\\
\operatorname{Tr}[\Omega_{RB}\Gamma_{RB}^{\mathcal{N}}]\geq1-\varepsilon,\\
\Omega_{RB}\leq\rho_{R}\otimes I_{B},\ \operatorname{Tr}[\rho_{R}]=1
\end{array}
\right\}  .
\end{multline}
The dual semi-definite program is given by%
\begin{multline}
D_{\min}^{\varepsilon}(\mathcal{N}\Vert\mathcal{M}%
)=\label{eq:dual-SDP-ch-sm-min-rel-ent}\\
-\log_{2}\sup_{\lambda,\mu,Y_{RB}\geq0}\left\{
\begin{array}
[c]{c}%
\mu\left(  1-\varepsilon\right)  -\lambda:\\
\mu\Gamma_{RB}^{\mathcal{N}}\leq\Gamma_{RB}^{\mathcal{M}}+Y_{RB},\\
\operatorname{Tr}_{B}[Y_{RB}]\leq\lambda I_{R}%
\end{array}
\right\}  .
\end{multline}

\end{proposition}

\begin{proof}
By definition, we have that%
\begin{equation}
D_{\min}^{\varepsilon}(\mathcal{N}\Vert\mathcal{M})=\sup_{\psi_{RA}}D_{\min
}^{\varepsilon}(\mathcal{N}_{A\rightarrow B}(\psi_{RA})\Vert\mathcal{M}%
_{A\rightarrow B}(\psi_{RA})),
\end{equation}
where%
\begin{equation}
D_{\min}^{\varepsilon}(\rho\Vert\sigma)=-\log_{2}\inf_{\Lambda\geq0}\left\{
\operatorname{Tr}[\Lambda\sigma]:\operatorname{Tr}[\Lambda\rho]\geq
1-\varepsilon,\Lambda\leq I\right\}  .
\end{equation}
This then means that%
\begin{multline}
D_{\min}^{\varepsilon}(\mathcal{N}\Vert\mathcal{M})=\\
-\log_{2}\inf_{\psi_{RA},\Lambda_{RB}\geq0}\left\{
\begin{array}
[c]{c}%
\operatorname{Tr}[\Lambda_{RB}\mathcal{M}_{A\rightarrow B}(\psi_{RA})]:\\
\operatorname{Tr}[\Lambda_{RB}\mathcal{N}_{A\rightarrow B}(\psi_{RA}%
)]\geq1-\varepsilon,\\
\Lambda_{RB}\leq I_{RB},\ \operatorname{Tr}[\psi_{RA}]=1,\\
\operatorname{Tr}[\psi_{RA}^{2}]=1,\ \psi_{RA}\geq0
\end{array}
\right\}  .
\end{multline}
Consider that we can restrict the infimum above to being over all pure states
$\psi_{RA}$ such that the reduced state $\psi_{R}$ is positive definite, i.e.,
$\psi_{R}>0$, due to the fact that the set of all such states is dense in the
set of all pure bipartite states.\ Note that we can write all such states as
$\psi_{RA}=X_{R}\Gamma_{RA}X_{R}^{\dag}$ for some operator $X_{R}$ such that
$\left\vert X_{R}\right\vert >0$ and $\operatorname{Tr}[X_{R}^{\dag}X_{R}]=1$.
Then it follows that%
\begin{align}
\operatorname{Tr}[\Lambda_{RB}\mathcal{M}_{A\rightarrow B}(\psi_{RA})]  &
=\operatorname{Tr}[\Omega_{RB}\Gamma_{RB}^{\mathcal{M}}],\\
\operatorname{Tr}[\Lambda_{RB}\mathcal{N}_{A\rightarrow B}(\psi_{RA})]  &
=\operatorname{Tr}[\Omega_{RB}\Gamma_{RB}^{\mathcal{N}}],\\
0\leq\Lambda_{RB}\leq I_{RB}\quad &  \Longleftrightarrow\quad0\leq\Omega
_{RB}\leq\rho_{R}\otimes I_{B}.
\end{align}
where we have defined $\Omega_{RB}:=X_{R}^{\dag}\Lambda_{RB}X_{R}$ and
$\rho_{R}:=X_{R}^{\dag}X_{R}$. Then we can rewrite as%
\begin{multline}
D_{\min}^{\varepsilon}(\mathcal{N}\Vert\mathcal{M})=\\
-\log_{2}\inf_{\rho_{R}>0,\Omega_{RB}\geq0}\left\{
\begin{array}
[c]{c}%
\operatorname{Tr}[\Omega_{RB}\Gamma_{RB}^{\mathcal{M}}]:\\
\operatorname{Tr}[\Omega_{RB}\Gamma_{RB}^{\mathcal{N}}]\geq1-\varepsilon,\\
\Omega_{RB}\leq\rho_{R}\otimes I_{B},\ \operatorname{Tr}[\rho_{R}]=1
\end{array}
\right\}  .
\end{multline}
Again using the fact that the set of positive-definite density operators is
dense in the set of all density operators, we conclude \eqref{eq:primal-SDP-ch-sm-min-rel-ent}.

The dual SDP\ is given by \eqref{eq:dual-SDP-ch-sm-min-rel-ent}, and its
optimal value is equal to the optimal value of the primal SDP\ in
\eqref{eq:primal-SDP-ch-sm-min-rel-ent}\ by strong duality.
\end{proof}

\bigskip Semi-definite programs for the channel min-relative entropy $D_{\min
}(\mathcal{N}\Vert\mathcal{M})$ are recovered by setting $\varepsilon=0$ in
\eqref{eq:primal-SDP-ch-sm-min-rel-ent} and \eqref{eq:dual-SDP-ch-sm-min-rel-ent}.

\begin{proposition}
\label{prop:smooth-max-ch-rel-ent-SDP}Let $\mathcal{N}_{A\rightarrow B}$ and
$\mathcal{M}_{A\rightarrow B}$ be quantum channels and $\varepsilon\in\left[
0,1\right]  $. The smooth channel max-relative entropy is given by the
following primal semi-definite program:%
\begin{multline}
D_{\max}^{\varepsilon}(\mathcal{N}\Vert\mathcal{M}%
)=\label{eq:primal-SDP-smooth-max-rel-ent-ch}\\
\log_{2}\inf_{\substack{\lambda,Z_{RB},\\Y_{RB}\geq0}}\left\{
\begin{array}
[c]{c}%
\lambda:\\
Y_{RB}\leq\lambda \Gamma_{RB}^{\mathcal{M}},\ \operatorname{Tr}_{B}%
[Y_{RB}]=I_{R}\\
\varepsilon I_{R}\geq\operatorname{Tr}_{B}[Z_{RB}],\\
Z_{RB}\geq\Gamma_{RB}^{\mathcal{N}}-Y_{RB}%
\end{array}
\right\}  .
\end{multline}
The dual semi-definite program is given by%
\begin{multline}
D_{\max}^{\varepsilon}(\mathcal{N}\Vert\mathcal{M}%
)=\label{eq:dual-SDP-smooth-max-rel-ent-ch}\\
\log_2 \sup_{\substack{L_{RB},P_{R},\\Y_{RB}\geq0,\\Z_{R}\in\operatorname{Herm}%
}}\left\{
\begin{array}
[c]{c}%
\operatorname{Tr}[Z_{R}]-\varepsilon\operatorname{Tr}[P_{R}
]+\operatorname{Tr}[Y_{RB}\Gamma_{RB}^{\mathcal{N}}]:\\
\operatorname{Tr}[\Gamma_{RB}^{\mathcal{M}}L_{RB}]\leq1,\\
Y_{RB}\leq P_{R}\otimes I_{B},\\
Z_{R}\otimes I_{B}+Y_{RB}\leq L_{RB}%
\end{array}
\right\}  .
\end{multline}

\end{proposition}

\begin{proof}
The primal form in \eqref{eq:primal-SDP-smooth-max-rel-ent-ch}\ follows from
the SDP formulation of the max-relative entropy and the SDP\ formulation of the
diamond distance of two channels in \eqref{eq:SDP-diamond-dist-dual}. By
definition, we have that%
\begin{equation}
D_{\max}^{\varepsilon}(\mathcal{N}\Vert\mathcal{M})=\inf_{\widetilde
{\mathcal{N}}:\frac{1}{2}\left\Vert \widetilde{\mathcal{N}}-\mathcal{N}%
\right\Vert _{\diamond}\leq\varepsilon}D_{\max}(\widetilde{\mathcal{N}}%
\Vert\mathcal{M}).
\end{equation}
Considering that%
\begin{align}
D_{\max}(\widetilde{\mathcal{N}}\Vert\mathcal{M})  &  =\log_{2}\inf\left\{
\lambda:\Gamma_{RB}^{\widetilde{\mathcal{N}}}\leq\lambda\Gamma_{RB}%
^{\mathcal{M}}\right\}  ,\\
\frac{1}{2}\left\Vert \widetilde{\mathcal{N}}-\mathcal{N}\right\Vert
_{\diamond}  &  =\inf_{\mu,Z_{RB}\geq0}\left\{
\begin{array}
[c]{c}%
\mu:Z_{RB}\geq\Gamma_{RB}^{\mathcal{N}}-\Gamma_{RB}^{\widetilde{\mathcal{N}}},\\
\mu I_{R}\geq\operatorname{Tr}_{B}[Z_{RB}]
\end{array}
\right\}  ,
\end{align}
the optimization in \eqref{eq:primal-smooth-max-rel-ent-ch-almost-SDP} below
follows by combining these, with $Y_{RB}$ understood as the Choi operator for
the channel $\widetilde{\mathcal{N}}$ being optimized:%
\begin{multline}
D_{\max}^{\varepsilon}(\mathcal{N}\Vert\mathcal{M}%
)=\label{eq:primal-smooth-max-rel-ent-ch-almost-SDP}\\
\log_{2}\inf_{\substack{\lambda,Z_{RB},\\Y_{RB}\geq0}}\left\{
\begin{array}
[c]{c}%
\lambda:\\
Y_{RB}\leq\lambda \Gamma_{RB}^{\mathcal{M}},\ \operatorname{Tr}_{B}%
[Y_{RB}]=I_{R}\\
\varepsilon I_{R}\geq\operatorname{Tr}_{B}[Z_{RB}],\\
Z_{RB}\geq\Gamma_{RB}^{\mathcal{N}}-Y_{RB}%
\end{array}
\right\}  .
\end{multline}

The dual program is given by \eqref{eq:dual-SDP-smooth-max-rel-ent-ch}, and
its optimal value is equal to the optimal value of
\eqref{eq:primal-SDP-smooth-max-rel-ent-ch}\ by strong duality.
\end{proof}

\bigskip Semi-definite programs for the channel max-relative entropy $D_{\max
}(\mathcal{N}\Vert\mathcal{M})$ are recovered by setting $\varepsilon=0$ in
\eqref{eq:primal-SDP-smooth-max-rel-ent-ch} and \eqref{eq:dual-SDP-smooth-max-rel-ent-ch}.

\subsection{Smooth channel min-relative entropy as approximate one-shot
distillable distinguishability}

\label{app:ch-smooth-min-rel-ent-approx-distillable-dist}In order to establish
the equality in \eqref{eq:approx-distill-dist-smooth-min}, we first prove the
following inequality:%
\begin{equation}
D_{d}^{\varepsilon}(\mathcal{N},\mathcal{M})\geq D_{\min}^{\varepsilon
}(\mathcal{N}\Vert\mathcal{M}). \label{eq:approx-distill-dist-smooth-min-ach}%
\end{equation}
Let $\psi_{RA}$ be an arbitrary pure state and $\Lambda_{RB}$ a corresponding
measurement operator satisfying $0\leq\Lambda_{RB}\leq I_{RB}$ and%
\begin{equation}
\operatorname{Tr}[\Lambda_{RB}\mathcal{N}_{A\rightarrow B}(\psi_{RA}%
)]\geq1-\varepsilon.
\end{equation}
Let $\Theta_{\left(  A\rightarrow B\right)  \rightarrow\left(  C\rightarrow
D\right)  }$ be the superchannel that traces out the input $C$, prepares the
pure state $\psi_{RA}$, transmits system $A$ through the unknown channel
$\mathcal{N}$ or $\mathcal{M}$, and then applies the following channel
$\mathcal{P}_{RB\rightarrow X}$ to systems $R$ and $B$, where $B$ is the
output of the unknown channel:%
\begin{multline}
\mathcal{P}_{RB\rightarrow X}(\omega_{RB}):=\operatorname{Tr}[\Lambda
_{RB}\omega_{RB}]|0\rangle\langle0|_{X}\\
+\operatorname{Tr}[\left(  I_{RB}-\Lambda_{RB}\right)  \omega_{RB}%
]|1\rangle\langle1|_{X}.
\end{multline}
With this construction, it follows that both $\Theta_{\left(  A\rightarrow
B\right)  \rightarrow\left(  C\rightarrow D\right)  }(\mathcal{N}%
_{A\rightarrow B})$ and $\Theta_{\left(  A\rightarrow B\right)  \rightarrow
\left(  C\rightarrow D\right)  }(\mathcal{M}_{A\rightarrow B})$ are replacer
channels, and we find that%
\begin{equation}
\Theta_{\left(  A\rightarrow B\right)  \rightarrow\left(  C\rightarrow
D\right)  }(\mathcal{N}_{A\rightarrow B})\approx_{\varepsilon}\mathcal{R}%
_{C\rightarrow D}^{|0\rangle\langle0|}. \label{eq:approx-eq-distill-ch}%
\end{equation}
Furthermore, the following equality holds%
\begin{equation}
\Theta_{\left(  A\rightarrow B\right)  \rightarrow\left(  C\rightarrow
D\right)  }(\mathcal{M}_{A\rightarrow B})=\mathcal{R}_{C\rightarrow D}%
^{\pi_{M}},
\end{equation}
for%
\begin{equation}
M=\frac{1}{\operatorname{Tr}[\Lambda_{RB}\mathcal{M}_{A\rightarrow B}%
(\psi_{RA})]}.
\end{equation}
The equality in \eqref{eq:approx-eq-distill-ch} follows from the reasoning in
\cite[Appendix~F-1]{WW19states}. It then follows that%
\begin{equation}
D_{d}^{\varepsilon}(\mathcal{N},\mathcal{M})\geq-\log_{2}\operatorname{Tr}%
[\Lambda_{RB}\mathcal{M}_{A\rightarrow B}(\psi_{RA})].
\end{equation}
Optimizing over all such $\psi_{RA}$ and $\Lambda_{RB}$ satisfying the
constraints above, we conclude that%
\begin{equation}
D_{d}^{\varepsilon}(\mathcal{N},\mathcal{M})\geq D_{\min}^{\varepsilon
}(\mathcal{N}\Vert\mathcal{M}).
\end{equation}

We now prove the opposite inequality:%
\begin{equation}
D_{\min}^{\varepsilon}(\mathcal{N}\Vert\mathcal{M})\geq D_{d}^{\varepsilon
}(\mathcal{N},\mathcal{M})\label{eq:approx-distill-dist-smooth-min-optimality}%
\end{equation}
Now let $\Theta_{\left(  A\rightarrow B\right)  \rightarrow\left(
C\rightarrow D\right)  }$ be an arbitrary superchannel satisfying%
\begin{align}
\Theta_{\left(  A\rightarrow B\right)  \rightarrow\left(  C\rightarrow
D\right)  }(\mathcal{N}_{A\rightarrow B}) &  \approx_{\varepsilon}%
\mathcal{R}_{C\rightarrow D}^{|0\rangle\langle0|}%
,\label{eq:conv-smooth-min-ch-1}\\
\Theta_{\left(  A\rightarrow B\right)  \rightarrow\left(  C\rightarrow
D\right)  }(\mathcal{M}_{A\rightarrow B}) &  =\mathcal{R}_{C\rightarrow
D}^{\pi_{M}}.\label{eq:conv-smooth-min-ch-2}%
\end{align}
Then we find that%
\begin{align}
D_{\min}^{\varepsilon}(\mathcal{N}\Vert\mathcal{M}) &  \geq D_{\min
}^{\varepsilon}(\Theta(\mathcal{N})\Vert\Theta(\mathcal{M}))\\
&  =D_{\min}^{\varepsilon}(\Theta(\mathcal{N})\Vert\mathcal{R}_{C\rightarrow
D}^{\pi_{M}})\\
&  \geq\log_{2}M.\label{eq:approx-distill-one-shot-final-ineq-conv}%
\end{align}
The first inequality follows from \eqref{eq:DP-superch} and the second
equality from \eqref{eq:conv-smooth-min-ch-2}. The last inequality follows
from reasoning similar to that in \cite[Appendix~F-1]{WW19states}. Let
$\Delta(\cdot)=|0\rangle\langle0|(\cdot)|0\rangle\langle0|+|1\rangle
\langle1|(\cdot)|1\rangle\langle1|$ denote the completely dephasing channel.
Since $\Theta(\mathcal{N})\approx_{\varepsilon}\mathcal{R}_{C\rightarrow
D}^{|0\rangle\langle0|}$, we find from the data processing inequality for
normalized trace distance and an arbitrary input state $\psi_{RC}$ that%
\begin{align}
\varepsilon & \geq\frac{1}{2}\left\Vert \Theta(\mathcal{N})(\psi
_{RC})-\mathcal{R}_{C\rightarrow D}^{|0\rangle\langle0|}(\psi_{RC})\right\Vert
_{1}\\
& \geq\frac{1}{2}\left\Vert \Theta(\mathcal{N})(\psi_{C})-\mathcal{R}%
_{C\rightarrow D}^{|0\rangle\langle0|}(\psi_{C})\right\Vert _{1}\\
& =\frac{1}{2}\left\Vert \Theta(\mathcal{N})(\psi_{C})-|0\rangle
\langle0|\right\Vert _{1}\\
& \geq\frac{1}{2}\left\Vert \Delta(\Theta(\mathcal{N})(\psi_{C}))-\Delta
(|0\rangle\langle0|)\right\Vert _{1}\\
& =1-\langle0|\Theta(\mathcal{N})(\psi_{C})|0\rangle,
\end{align}
which implies that $\langle0|\Theta(\mathcal{N})(\psi_{C})|0\rangle
\geq1-\varepsilon$ for all input states $\psi_{RC}$. Thus, we can take $\Lambda_{RD}%
=I_{R}\otimes|0\rangle\langle0|_{D}$ in the definition of $D_{\min
}^{\varepsilon}(\Theta(\mathcal{N})\Vert\mathcal{R}_{C\rightarrow D}^{\pi_{M}%
})$, and we have that $\operatorname{Tr}[\Lambda_{RD}\Theta(\mathcal{N}%
)(\psi_{RC})]\geq1-\varepsilon$ while $\operatorname{Tr}[\Lambda
_{RD}\mathcal{R}_{C\rightarrow D}^{\pi_M}(\psi_{RC})]=1/M$ for
all input states $\psi_{RC}$. Since $D_{\min}^{\varepsilon}(\Theta
(\mathcal{N})\Vert\mathcal{R}_{C\rightarrow D}^{\pi_{M}})$ involves an
optimization over all measurement operators $\Lambda_{RD}$\ and states
$\psi_{RC}$ satisfying $\operatorname{Tr}[\Lambda_{RD}\Theta(\mathcal{N}%
)(\psi_{RC})]\geq1-\varepsilon$, we conclude the inequality in
\eqref{eq:approx-distill-one-shot-final-ineq-conv}. Since the inequality holds
for an arbitrary superchannel $\Theta_{\left(  A\rightarrow B\right)
\rightarrow\left(  C\rightarrow D\right)  }$ satisfying
\eqref{eq:conv-smooth-min-ch-1}--\eqref{eq:conv-smooth-min-ch-2}, we conclude
\eqref{eq:approx-distill-dist-smooth-min-optimality}. 

Putting together \eqref{eq:approx-distill-dist-smooth-min-ach}\ and
\eqref{eq:approx-distill-dist-smooth-min-optimality}, we conclude the equality
in \eqref{eq:approx-distill-dist-smooth-min}, i.e., $D_{\min}^{\varepsilon
}(\mathcal{N}\Vert\mathcal{M})=D_{d}^{\varepsilon}(\mathcal{N},\mathcal{M})$.

\subsection{Smooth channel max-relative entropy as approximate one-shot
distinguishability cost}

\label{app:ch-smooth-max-rel-ent-approx-dist-cost}In order to establish the
equality in \eqref{eq:approx-dist-cost-smooth-max}, we first prove the
following inequality:%
\begin{equation}
D_{c}^{\varepsilon}(\mathcal{N},\mathcal{M})\leq D_{\max}^{\varepsilon
}(\mathcal{N}\Vert\mathcal{M}). \label{eq:approx-dist-cost-smooth-max-ach}%
\end{equation}
Let $\widetilde{\mathcal{N}}_{A\rightarrow B}$ be a quantum channel satisfying
$\widetilde{\mathcal{N}}_{A\rightarrow B}\approx_{\varepsilon}\mathcal{N}%
_{A\rightarrow B}$. Then by constructing a superchannel as we did in
Appendix~\ref{app:ch-max-rel-ent-exact-dist-cost}, but for $\widetilde{\mathcal{N}}_{A\rightarrow B}$ instead of
$\mathcal{N}_{A\rightarrow B}$, we conclude the following inequality:%
\begin{equation}
D_{c}^{\varepsilon}(\mathcal{N},\mathcal{M})\leq D_{\max}(\widetilde
{\mathcal{N}}\Vert\mathcal{M}).
\end{equation}
Then taking the infimum over all such channels $\widetilde{\mathcal{N}%
}_{A\rightarrow B}$ satisfying $\widetilde{\mathcal{N}}_{A\rightarrow
B}\approx_{\varepsilon}\mathcal{N}_{A\rightarrow B}$, we conclude the
inequality in \eqref{eq:approx-dist-cost-smooth-max-ach}.

For the opposite inequality%
\begin{equation}
D_{c}^{\varepsilon}(\mathcal{N},\mathcal{M})\geq D_{\max}^{\varepsilon
}(\mathcal{N}\Vert\mathcal{M}), \label{eq:approx-dist-cost-smooth-max-conv}%
\end{equation}
let $\Theta$ be an arbitrary superchannel satisfying%
\begin{align}
\Theta(\mathcal{R}_{C\rightarrow D}^{|0\rangle\langle0|})  &  \approx
_{\varepsilon}\mathcal{N}_{A\rightarrow B}%
,\label{eq:super-ch-conditions-approx-dilut-1}\\
\Theta(\mathcal{R}_{C\rightarrow D}^{\pi_{M}})  &  =\mathcal{M}_{A\rightarrow
B}. \label{eq:super-ch-conditions-approx-dilut-2}%
\end{align}
Then consider that%
\begin{align}
\log_{2}M  &  =D_{\max}(|0\rangle\langle0|\Vert\pi_{M})\\
&  =D_{\max}(\mathcal{R}_{C\rightarrow D}^{|0\rangle\langle0|}\Vert
\mathcal{R}_{C\rightarrow D}^{\pi_{M}})\\
&  \geq D_{\max}(\Theta(\mathcal{R}_{C\rightarrow D}^{|0\rangle\langle
0|})\Vert\Theta(\mathcal{R}_{C\rightarrow D}^{\pi_{M}}))\\
&  =D_{\max}(\Theta(\mathcal{R}_{C\rightarrow D}^{|0\rangle\langle0|}%
)\Vert\mathcal{M}_{A\rightarrow B})\\
&  \geq D_{\max}^{\varepsilon}(\mathcal{N}_{A\rightarrow B}\Vert
\mathcal{M}_{A\rightarrow B}).
\end{align}
The second equality follows from \eqref{eq:env-seizable-gen-div-equal}, given
that pairs of replacer channels are environment seizable. The first inequality
follows from \eqref{eq:DP-superch}. The last inequality follows from the
definition in \eqref{eq:ch-smooth-max-rel-ent}. Since the chain of
inequalities holds for all superchannels $\Theta$ satisfying
\eqref{eq:super-ch-conditions-approx-dilut-1}--\eqref{eq:super-ch-conditions-approx-dilut-2},
we conclude \eqref{eq:approx-dist-cost-smooth-max-conv}. 

Putting together
\eqref{eq:approx-dist-cost-smooth-max-ach} and
\eqref{eq:approx-dist-cost-smooth-max-conv}, we conclude the equality in
\eqref{eq:approx-dist-cost-smooth-max}, i.e., $D_{c}^{\varepsilon}%
(\mathcal{N},\mathcal{M})=D_{\max}^{\varepsilon}(\mathcal{N}\Vert\mathcal{M})$.

\subsection{Limits of smooth channel min- and max-relative entropy}

\label{app:limits-smooth-min-max-to-exact}Here we provide an alternate proof
of the limits stated in
\eqref{eq:smooth-to-exact-limits}--\eqref{eq:smooth-to-exact-limits-2},
starting with \eqref{eq:smooth-to-exact-limits}. These proofs use some of the
results from \cite[Appendix~A-3]{WW19states} as a starting point. Let
$\psi_{RA}$ be an arbitrary bipartite state. By the inequality $D_{\min
}^{\varepsilon}(\rho\Vert\sigma)\geq D_{\min}(\rho\Vert\sigma)$, which holds
for all states $\rho$ and $\sigma$ and $\varepsilon\in(0,1)$, we conclude that%
\begin{multline}
D_{\min}^{\varepsilon}(\mathcal{N}_{A\rightarrow B}(\psi_{RA})\Vert
\mathcal{M}_{A\rightarrow B}(\psi_{RA}))\\
\geq D_{\min}(\mathcal{N}_{A\rightarrow B}(\psi_{RA})\Vert\mathcal{M}%
_{A\rightarrow B}(\psi_{RA}))
\end{multline}
for all $\varepsilon\in(0,1)$. Now taking a supremum over all $\psi_{RA}$, we
find that%
\begin{equation}
D_{\min}^{\varepsilon}(\mathcal{N}\Vert\mathcal{M})\geq D_{\min}%
(\mathcal{N}\Vert\mathcal{M}),
\end{equation}
for all $\varepsilon\in(0,1)$. Taking the limit, we conclude that%
\begin{equation}
\liminf_{\varepsilon\rightarrow0}D_{\min}^{\varepsilon}(\mathcal{N}%
\Vert\mathcal{M})\geq D_{\min}(\mathcal{N}\Vert\mathcal{M}).
\label{eq:lim-smooth-dmin-ch-exact-liminf}%
\end{equation}
For the other limit, recall the following inequality from \cite[Appendix~A-3]%
{WW19states}, holding for all states $\rho$ and $\sigma$, for $\varepsilon
\in(0,1)$, and $\alpha\in(0,1)$:%
\begin{multline}
D_{\alpha}(\rho\Vert\sigma)\geq\\
\frac{1}{\alpha-1}\log_{2}\!\left[
\begin{array}
[c]{c}%
\left(  1-\varepsilon\right)  ^{\alpha}\left(  2^{-D_{\min}^{\varepsilon}%
(\rho\Vert\sigma)}\right)  ^{1-\alpha}\\
+\ \varepsilon^{\alpha}\left(  1-2^{-D_{\min}^{\varepsilon}(\rho\Vert\sigma
)}\right)  ^{1-\alpha}%
\end{array}
\right]  .
\end{multline}
Taking an optimization over all input states $\psi_{RA}$ to the channels
$\mathcal{N}_{A\rightarrow B}$ and $\mathcal{M}_{A\rightarrow B}$, we conclude
that%
\begin{multline}
D_{\alpha}(\mathcal{N}\Vert\mathcal{M})\geq\\
\frac{1}{\alpha-1}\log_{2}\!\left[
\begin{array}
[c]{c}%
\left(  1-\varepsilon\right)  ^{\alpha}\left(  2^{-D_{\min}^{\varepsilon
}(\mathcal{N}\Vert\mathcal{M})}\right)  ^{1-\alpha}\\
+\ \varepsilon^{\alpha}\left(  1-2^{-D_{\min}^{\varepsilon}(\mathcal{N}%
\Vert\mathcal{M})}\right)  ^{1-\alpha}%
\end{array}
\right]  .
\end{multline}
Taking the limit as $\varepsilon\rightarrow0$, we conclude that%
\begin{equation}
D_{\alpha}(\mathcal{N}\Vert\mathcal{M})\geq\limsup_{\varepsilon\rightarrow
0}D_{\min}^{\varepsilon}(\mathcal{N}\Vert\mathcal{M})
\end{equation}
for all $\alpha\in(0,1)$. Now taking the limit of the left-hand side as
$\alpha\rightarrow0$, and applying arguments similar to those needed for
\cite[Lemma~10]{Cooney2016}, we conclude that%
\begin{equation}
D_{\min}(\mathcal{N}\Vert\mathcal{M})\geq\limsup_{\varepsilon\rightarrow
0}D_{\min}^{\varepsilon}(\mathcal{N}\Vert\mathcal{M}).
\label{eq:lim-smooth-dmin-ch-exact-limsup}%
\end{equation}
Combining \eqref{eq:lim-smooth-dmin-ch-exact-liminf}\ and
\eqref{eq:lim-smooth-dmin-ch-exact-limsup}, we conclude the limit stated in \eqref{eq:smooth-to-exact-limits}.

Another proof for the inequality in \eqref{eq:smooth-to-exact-limits-2} goes
as follows. By taking $\widetilde{\mathcal{N}}=\mathcal{N}$, we conclude that
$\widetilde{\mathcal{N}}\approx_{\varepsilon}\mathcal{N}$, so that applying
definitions gives%
\begin{equation}
D_{\max}(\mathcal{N}\Vert\mathcal{M})\geq D_{\max}^{\varepsilon}%
(\mathcal{N}\Vert\mathcal{M})
\end{equation}
for all $\varepsilon\in(0,1)$. Then applying a limit gives%
\begin{equation}
D_{\max}(\mathcal{N}\Vert\mathcal{M})\geq\limsup_{\varepsilon\rightarrow
0}D_{\max}^{\varepsilon}(\mathcal{N}\Vert\mathcal{M}).
\label{eq:lim-smooth-dmax-dmax-limsup}%
\end{equation}
Now suppose that $\widetilde{\mathcal{N}}$ is a channel satisfying
$\widetilde{\mathcal{N}}\approx_{\varepsilon}\mathcal{N}$ for $\varepsilon
\in(0,1)$. Then this implies that%
\begin{equation}
\frac{1}{2}\left\Vert \widetilde{\mathcal{N}}_{A\rightarrow B}(\Phi
_{RA})-\mathcal{N}_{A\rightarrow B}(\Phi_{RA})\right\Vert _{1}\leq\varepsilon,
\end{equation}
and applying an inequality from \cite[Appendix~A-3]{WW19states}, we find that%
\begin{multline}
D_{\max}(\widetilde{\mathcal{N}}\Vert\mathcal{M})\geq\\
\log_{2}\!\left(
\begin{array}
[c]{c}%
\left\Vert \left[  \mathcal{M}_{A\rightarrow B}(\Phi_{RA})\right]  ^{-\frac
{1}{2}}\left[  \mathcal{N}_{A\rightarrow B}(\Phi_{RA})\right]  ^{\frac{1}{2}%
}\right\Vert _{\infty}\\
-2\varepsilon\left\Vert \left(  \mathcal{M}_{A\rightarrow B}(\Phi
_{RA})\right)  ^{-1}\right\Vert _{\infty}%
\end{array}
\right)
\end{multline}
Since this bound holds uniformly for all channels $\widetilde{\mathcal{N}}$
satisfying $\widetilde{\mathcal{N}}\approx_{\varepsilon}\mathcal{N}$, we
conclude that%
\begin{multline}
D_{\max}^{\varepsilon}(\mathcal{N}\Vert\mathcal{M})\geq\\
\log_{2}\!\left(
\begin{array}
[c]{c}%
\left\Vert \left[  \mathcal{M}_{A\rightarrow B}(\Phi_{RA})\right]  ^{-\frac
{1}{2}}\left[  \mathcal{N}_{A\rightarrow B}(\Phi_{RA})\right]  ^{\frac{1}{2}%
}\right\Vert _{\infty}\\
-2\varepsilon\left\Vert \left(  \mathcal{M}_{A\rightarrow B}(\Phi
_{RA})\right)  ^{-1}\right\Vert _{\infty}%
\end{array}
\right)
\end{multline}
Now taking the limit $\varepsilon\rightarrow0$, we find that%
\begin{equation}
\liminf_{\varepsilon\rightarrow0}D_{\max}^{\varepsilon}(\mathcal{N}%
\Vert\mathcal{M})\geq D_{\max}(\mathcal{N}\Vert\mathcal{M}).
\label{eq:lim-smooth-dmax-dmax-liminf}%
\end{equation}
Combining \eqref{eq:lim-smooth-dmax-dmax-limsup}\ and
\eqref{eq:lim-smooth-dmax-dmax-liminf}, we conclude the inequality in \eqref{eq:smooth-to-exact-limits-2}.

\section{Upper bound on smooth max-relative entropy of classical--quantum
channels}

\label{app:smooth-max-rel-ent-bnd-cq-ch}

The main purpose of this appendix is to prove
Proposition~\ref{prop:cq-smooth-dmax-up-bnd}, which establishes an upper bound
on the smooth max-relative entropy of classical--quantum channels. We begin by
noting a simple lemma:

\begin{lemma}
Let $\{\rho_{B}^{x}\}_{x\in\mathcal{X}}$ and $\{\sigma_{B}^{x}\}_{x\in
\mathcal{X}}$ be the output states of classical--quantum channels
$\mathcal{N}_{X\rightarrow B}$ and $\mathcal{M}_{X\rightarrow B}$,
respectively, as defined in \eqref{eq:cq-ch-def-1}--\eqref{eq:cq-ch-def-2}.
Then we have that%
\begin{align}
\left\Vert \mathcal{N}_{X\rightarrow B}-\mathcal{M}_{X\rightarrow
B}\right\Vert _{\diamond}  &  =\sup_{x}\left\Vert \rho_{B}^{x}-\sigma_{B}%
^{x}\right\Vert _{1},\\
\widetilde{D}_{\alpha}(\mathcal{N}_{X\rightarrow B}\Vert\mathcal{M}%
_{X\rightarrow B})  &  =\sup_{x}\widetilde{D}_{\alpha}(\rho_{B}^{x}\Vert
\sigma_{B1}^{x}),
\end{align}
where the latter equality holds for all $\alpha \in [1/2,1)\cup(1,\infty)$.

\end{lemma}

\begin{proof}
The second equality follows from \cite[Lemma~26]{BHKW18}. The proof of the
first equality is similar to the proof of the second one. For completeness, we
provide a proof. Let $\psi_{RX}$ be an arbitrary pure bipartite quantum state
($X$ is a quantum system here). Then the states resulting from the action of
the classical--quantum channels on this state are as follows:%
\begin{align}
\omega_{RB}  &  :=\sum_{x}p(x)\psi_{R}^{x}\otimes\rho_{B}^{x},\\
\tau_{RB}  &  :=\sum_{x}p(x)\psi_{R}^{x}\otimes\sigma_{B}^{x},
\end{align}
where%
\begin{align}
p(x)  &  :=\operatorname{Tr}[|x\rangle\langle x|_{X}\psi_{RX}],\\
\psi_{R}^{x}  &  :=\frac{1}{p(x)}\operatorname{Tr}_{X}[|x\rangle\langle
x|_{X}\psi_{RX}].
\end{align}
Then it follows that%
\begin{align}
&  \left\Vert \mathcal{N}_{X\rightarrow B}(\psi_{RX})-\mathcal{M}%
_{X\rightarrow B}(\psi_{RX})\right\Vert _{1}\nonumber\\
&  =\left\Vert \omega_{RB}-\tau_{RB}\right\Vert _{1}\\
&  \leq\sum_{x}p(x)\left\Vert \psi_{R}^{x}\otimes\rho_{B}^{x}-\psi_{R}%
^{x}\otimes\sigma_{B}^{x}\right\Vert _{1}\\
&  =\sum_{x}p(x)\left\Vert \rho_{B}^{x}-\sigma_{B}^{x}\right\Vert _{1}\\
&  \leq\sup_{x}\left\Vert \rho_{B}^{x}-\sigma_{B}^{x}\right\Vert _{1}.
\end{align}
So we have established a uniform upper bound for any possible bipartite input
state. The upper bound is achieved by calculating the value of $x$ that
achieves the optimum and inputting $|x\rangle\langle x|_{X}$ to the channel box.
\end{proof}

\bigskip
The following proposition generalizes \cite[Proposition~11]{WW19states}\ and
\cite[Theorem~3]{ABJT19}. The main proof idea ultimately still has its roots in \cite{JN12}.

\begin{proposition}
\label{prop:cq-smooth-dmax-up-bnd}Let $\{\rho_{B}^{x}\}_{x\in\mathcal{X}}$ and
$\{\sigma_{B}^{x}\}_{x\in\mathcal{X}}$ be the output states of
classical--quantum channels $\mathcal{N}_{X\rightarrow B}$ and $\mathcal{M}%
_{X\rightarrow B}$, respectively. Then the following bound holds for all
$\alpha>1$ and $\varepsilon\in(0,1)$:%
\begin{multline}
D_{\max}^{\varepsilon}(\mathcal{N}_{X\rightarrow B}\Vert\mathcal{M}%
_{X\rightarrow B})\leq\widetilde{D}_{\alpha}(\mathcal{N}_{X\rightarrow B}%
\Vert\mathcal{M}_{X\rightarrow B})\label{eq:dmax-bnd-cq-channels}\\
+\frac{1}{\alpha-1}\log_{2}\!\left(  \frac{1}{\varepsilon^{2}}\right)  +\log
_{2}\!\left(  \frac{1}{1-\varepsilon^{2}}\right)  .
\end{multline}

\end{proposition}

\begin{proof}
Consider that%
\begin{align}
&  2^{D_{\max}^{\varepsilon}(\mathcal{N}\Vert\mathcal{M})}\nonumber\\
&  =\inf_{\widetilde{\mathcal{N}}:\frac{1}{2}\left\Vert \widetilde
{\mathcal{N}}-\mathcal{N}\right\Vert _{\diamond}\leq\varepsilon}2^{D_{\max
}(\widetilde{\mathcal{N}}\Vert\mathcal{M})}\\
&  \leq\inf_{\substack{\{\widetilde{\rho}_{B}^{x}\}_{x\in\mathcal{X}}%
:\\\frac{1}{2}\left\Vert \widetilde{\rho}_{B}^{x}-\rho_{B}^{x}\right\Vert
_{1}\leq\varepsilon}}2^{D_{\max}\left(  \bigoplus\limits_{x}\widetilde{\rho
}_{B}^{x}\middle\Vert\bigoplus\limits_{x}\sigma_{B}^{x}\right)  }\\
&  =\inf_{\substack{\{\widetilde{\rho}_{B}^{x}\}_{x\in\mathcal{X}}:\\\frac
{1}{2}\left\Vert \widetilde{\rho}_{B}^{x}-\rho_{B}^{x}\right\Vert _{1}%
\leq\varepsilon}}\sup_{x\in\mathcal{X}}2^{D_{\max}\left(  \widetilde{\rho}%
_{B}^{x}\Vert\sigma_{B}^{x}\right)  }\\
&  =\inf_{\substack{\{\widetilde{\rho}_{B}^{x}\}_{x\in\mathcal{X}}:\\\frac
{1}{2}\left\Vert \widetilde{\rho}_{B}^{x}-\rho_{B}^{x}\right\Vert _{1}%
\leq\varepsilon}}\sup_{x\in\mathcal{X}}\sup_{\substack{\Lambda_{B}^{x}%
\geq0,\\\operatorname{Tr}[\Lambda_{B}^{x}\sigma_{B}^{x}]\leq1}%
}\operatorname{Tr}[\Lambda_{B}^{x}\widetilde{\rho}_{B}^{x}]\\
%&  =\inf_{\substack{\{\widetilde{\rho}_{B}^{x}\}_{x\in\mathcal{X}}:\\\frac
%{1}{2}\left\Vert \widetilde{\rho}_{B}^{x}-\rho_{B}^{x}\right\Vert _{1}%
%\leq\varepsilon}}\sup_{x\in\mathcal{X}}\sup_{\substack{\Lambda_{B}^{x}%
%\geq0,\\\operatorname{Tr}[\Lambda_{B}^{x}\sigma_{B}^{x}]\leq1}%
%}\operatorname{Tr}[\Lambda_{B}^{x}\widetilde{\rho}_{B}^{x}]\\
&  =\inf_{\substack{\{\widetilde{\rho}_{B}^{x}\}_{x\in\mathcal{X}}:\\\frac
{1}{2}\left\Vert \widetilde{\rho}_{B}^{x}-\rho_{B}^{x}\right\Vert _{1}%
\leq\varepsilon}}\sup_{\{p(x)\}_{x\in\mathcal{X}}}\sup_{\substack{\Lambda
_{B}^{x}\geq0,\\\operatorname{Tr}[\Lambda_{B}^{x}\sigma_{B}^{x}]\leq1}%
}\sum_{x}p(x)\operatorname{Tr}[\Lambda_{B}^{x}\widetilde{\rho}_{B}^{x}]\\
&  =\sup_{\{p(x)\}_{x\in\mathcal{X}}}\sup_{\substack{\Lambda_{B}^{x}%
\geq0,\\\operatorname{Tr}[\Lambda_{B}^{x}\sigma_{B}^{x}]\leq1}}\inf
_{\substack{\{\widetilde{\rho}_{B}^{x}\}_{x\in\mathcal{X}}:\\\frac{1}%
{2}\left\Vert \widetilde{\rho}_{B}^{x}-\rho_{B}^{x}\right\Vert _{1}%
\leq\varepsilon}}\sum_{x}p(x)\operatorname{Tr}[\Lambda_{B}^{x}\widetilde{\rho
}_{B}^{x}].
\end{align}
The first inequality follows by restricting the optimization to be over
classical--quantum channels. The last equality follows because the objective
function $\sum_{x}p(x)\operatorname{Tr}[\Lambda_{B}^{x}\widetilde{\rho}%
_{B}^{x}]$ is jointly concave with respect to $\{p(x)\}_{x\in\mathcal{X}}$ and
$\{\Lambda_{B}^{x}\}_{x\in\mathcal{X}}$, and it is convex with respect to the
states $\{\widetilde{\rho}_{B}^{x}\}_{x\in\mathcal{X}}$. Also, the sets over
which we are optimizing are convex and compact. Thus, the Sion minimax theorem applies \cite{sion58}. For each operator $\Lambda
_{B}^{x}$, let its spectral decomposition be given as%
\begin{equation}
\Lambda_{B}^{x}=\sum_{y}\lambda_{x,y}|\phi_{x,y}\rangle\langle\phi_{x,y}|_{B}.
\end{equation}
Then define the set $\mathcal{S}_{x}$ and the projection $\Pi_{B}^{x}$ as%
\begin{align}
\mathcal{S}_{x}  &  :=\left\{  y:\langle\phi_{x,y}|_{B}\rho_{B}^{x}|\phi
_{x,y}\rangle_{B}>2^{\lambda}\langle\phi_{x,y}|_{B}\sigma_{B}^{x}|\phi
_{x,y}\rangle_{B}\right\}  ,\\
\lambda &  :=\widetilde{D}_{\alpha}(\mathcal{N}\Vert\mathcal{M})+\frac
{1}{\alpha-1}\log_{2}\!\left(  \frac{1}{\varepsilon^{2}}\right)  ,\\
\Pi_{B}^{x}  &  :=\sum_{y\in\mathcal{S}_{x}}|\phi_{x,y}\rangle\langle
\phi_{x,y}|_{B}.
\end{align}
The above implies for all $x \in \mathcal{X}$ that%
\begin{equation}
\frac{\operatorname{Tr}[\Pi_{B}^{x}\rho_{B}^{x}]}{\operatorname{Tr}[\Pi
_{B}^{x}\sigma_{B}^{x}]}>2^{\lambda}.
\end{equation}
Then from the data processing inequality of the sandwiched R\'enyi relative
entropy for $\alpha>1$ \cite{FL13,beigi2013sandwiched} and by dropping terms, we find that%
\begin{align}
&  \widetilde{D}_{\alpha}(\mathcal{N}\Vert\mathcal{M})\nonumber\\
&  \geq\frac{1}{\alpha-1}\log_{2}\!\left(  \operatorname{Tr}[\Pi_{B}^{x}\rho
_{B}^{x}]\left(  \frac{\operatorname{Tr}[\Pi_{B}^{x}\rho_{B}^{x}%
]}{\operatorname{Tr}[\Pi_{B}^{x}\sigma_{B}^{x}]}\right)  ^{\alpha-1}\right) \\
&  \geq\frac{1}{\alpha-1}\log_{2}\!\left(  \operatorname{Tr}[\Pi_{B}^{x}\rho
_{B}^{x}]\left(  2^{\lambda}\right)  ^{\alpha-1}\right) \\
&  =\frac{1}{\alpha-1}\log_{2}\operatorname{Tr}[\Pi_{B}^{x}\rho_{B}%
^{x}]+\lambda,
\end{align}
which in turn implies that%
\begin{equation}
\operatorname{Tr}[\Pi_{B}^{x}\rho_{B}^{x}]\leq\varepsilon^{2},
\end{equation}
for all $x\in\mathcal{X}$. Then we find that%
\begin{equation}
\operatorname{Tr}[\hat{\Pi}_{B}^{x}\rho_{B}^{x}]\geq1-\varepsilon^{2},
\end{equation}
for all $x\in\mathcal{X}$, where%
\begin{equation}
\hat{\Pi}_{B}^{x}:=I_{B}-\Pi_{B}^{x}.
\end{equation}
We define the states%
\begin{align}
\widetilde{\rho}_{B}^{x}  &  :=\frac{\hat{\Pi}_{B}^{x}\rho_{B}^{x}\hat{\Pi
}_{B}^{x}}{\operatorname{Tr}[\hat{\Pi}_{B}^{x}\rho_{B}^{x}]}\\
&  \leq\frac{\hat{\Pi}_{B}^{x}\rho_{B}^{x}\hat{\Pi}_{B}^{x}}{1-\varepsilon
^{2}},
\end{align}
and we note that $F(\widetilde{\rho}_{B}^{x},\rho_{B}^{x})\geq1-\varepsilon
^{2}$ implies%
\begin{equation}
\frac{1}{2}\left\Vert \widetilde{\rho}_{B}^{x}-\rho_{B}^{x}\right\Vert
_{1}\leq\varepsilon
\end{equation}
for all $x\in\mathcal{X}$. Then we find that%
\begin{align}
&  \!\!\!\!\!\left(  1-\varepsilon^{2}\right)  \operatorname{Tr}[\Lambda
_{B}^{x}\widetilde{\rho}_{B}^{x}]\nonumber\\
&  \leq\operatorname{Tr}[\Lambda_{B}^{x}\hat{\Pi}_{B}^{x}\rho_{B}^{x}\hat{\Pi
}_{B}^{x}]\\
&  =\operatorname{Tr}[\hat{\Pi}_{B}^{x}\Lambda_{B}^{x}\hat{\Pi}_{B}^{x}%
\rho_{B}^{x}]\\
&  =\sum_{y\notin\mathcal{S}_{x}}\lambda_{x,y}\langle\phi_{x,y}|_{B}\rho
_{B}^{x}|\phi_{x,y}\rangle_{B}\\
&  \leq2^{\lambda}\sum_{y\notin\mathcal{S}_{x}}\lambda_{x,y}\langle\phi
_{x,y}|_{B}\sigma_{B}^{x}|\phi_{x,y}\rangle_{B}\\
&  \leq2^{\lambda}\operatorname{Tr}[\Lambda_{B}^{x}\sigma_{B}^{x}]\\
&  \leq2^{\lambda}.
\end{align}
So this means that we have the following bound holding for all $x\in
\mathcal{X}$:%
\begin{equation}
\operatorname{Tr}[\Lambda_{B}^{x}\widetilde{\rho}_{B}^{x}]\leq2^{\lambda
+\log_{2}\!\left(  \frac{1}{1-\varepsilon^{2}}\right)  }.
\end{equation}
From this, we conclude \eqref{eq:dmax-bnd-cq-channels}.
\end{proof}

\section{Bounds for general one-shot or $n$-shot parallel channel box
transformations}

\label{app:gen-box-trans-bnds}In this appendix, we establish some bounds for
general channel box transformations, by generalizing the results of
\cite{WW19states}\ from states to channels. We begin with the following proposition:

\begin{proposition}
\label{prop:pseudo-continuity-channel-sand}Let $\mathcal{N}_{A\rightarrow
B}^{0}$, $\mathcal{N}_{A\rightarrow B}^{1}$, and $\mathcal{M}_{A\rightarrow
B}$ be channels such that $D_{\max}(\mathcal{N}^{0}\Vert\mathcal{M})<\infty$.
Then for $\alpha\in(1/2,1)$ and $\beta:=\beta(\alpha)=\alpha/\left(
2\alpha-1\right)  >1$, the following inequality holds%
\begin{equation}
\widetilde{D}_{\beta}(\mathcal{N}^{0}\Vert\mathcal{M})-\widetilde{D}_{\alpha
}(\mathcal{N}^{1}\Vert\mathcal{M})\geq\frac{\alpha}{1-\alpha}\log
_{2}F(\mathcal{N}^{0},\mathcal{N}^{1}),
\label{eq:pseudo-continuity-channel-sandwiched}%
\end{equation}
where $\widetilde{D}_{\beta}(\mathcal{N}^{0}\Vert\mathcal{M})$ and $\widetilde{D}_{\alpha
}(\mathcal{N}^{1}\Vert\mathcal{M})$ are channel sandwiched R\'enyi relative entropies and $F(\mathcal{N}^{0},\mathcal{N}^{1})$ is the channel fidelity, each of which is defined from \eqref{eq:gen-ch-div-def} and the underlying functions of states.

\end{proposition}

\begin{proof}
Recall the following inequality from \cite[Lemma~1]{WW19states} for states
$\rho_{0}$, $\rho_{1}$, and $\sigma$:%
\begin{equation}
\widetilde{D}_{\beta}(\rho_{0}\Vert\sigma)-\widetilde{D}_{\alpha}(\rho
_{1}\Vert\sigma)\geq\frac{\alpha}{1-\alpha}\log_{2}F(\rho_{0},\rho_{1}).
\end{equation}
Let $\psi_{RA}$ be a pure bipartite state. Applying the above inequality for
states, we find that%
\begin{multline}
\widetilde{D}_{\beta}(\mathcal{N}_{A\rightarrow B}^{0}(\psi_{RA}%
)\Vert\mathcal{M}_{A\rightarrow B}(\psi_{RA}))\\
\geq\widetilde{D}_{\alpha}(\mathcal{N}_{A\rightarrow B}^{1}(\psi_{RA}%
)\Vert\mathcal{M}_{A\rightarrow B}(\psi_{RA}))\\
+\frac{\alpha}{1-\alpha}\log_{2}F(\mathcal{N}_{A\rightarrow B}^{0}(\psi
_{RA}),\mathcal{N}_{A\rightarrow B}^{1}(\psi_{RA})).
\end{multline}
Taking a supremum over all input states $\psi_{RA}$ on the left-hand side, and
an infimum on the right-hand side, we find that%
\begin{multline}
\widetilde{D}_{\beta}(\mathcal{N}^{0}\Vert\mathcal{M})\geq\widetilde
{D}_{\alpha}(\mathcal{N}_{A\rightarrow B}^{1}(\psi_{RA})\Vert\mathcal{M}%
_{A\rightarrow B}(\psi_{RA}))\\
+\frac{\alpha}{1-\alpha}\log_{2}F(\mathcal{N}^{0},\mathcal{N}^{1}).
\end{multline}
Since the above inequality holds for all input states $\psi_{RA}$, we finally
take another supremum to conclude \eqref{eq:pseudo-continuity-channel-sandwiched}.
\end{proof}

\begin{proposition}
\label{prop:pseudo-continuity-channel-petz}Let $\mathcal{N}_{A\rightarrow
B}^{0}$, $\mathcal{N}_{A\rightarrow B}^{1}$, and $\mathcal{M}_{A\rightarrow
B}$ be channels such that $D_{\max}(\mathcal{N}^{0}\Vert\mathcal{M})<\infty$.
Then for $\alpha\in(0,1)$ and $\beta:=\beta(\alpha)=2-\alpha>1$, the following
inequality holds%
\begin{multline}
D_{\beta}(\mathcal{N}^{0}\Vert\mathcal{M})-D_{\alpha}(\mathcal{N}^{1}%
\Vert\mathcal{M})\\
\geq\frac{2}{1-\alpha}\log_{2}\!\left[  1-\frac{1}{2}\left\Vert \mathcal{N}%
^{0}-\mathcal{N}^{1}\right\Vert _{\diamond}\right]  .
\end{multline}
where $D_{\beta}(\mathcal{N}^{0}\Vert\mathcal{M})$ and $D_{\alpha
}(\mathcal{N}^{1}\Vert\mathcal{M})$ are channel
Petz--R\'enyi relative entropies, each of which is defined from \eqref{eq:gen-ch-div-def} and the underlying functions of states.

\end{proposition}

\begin{proof}
This is a consequence of the following inequality from \cite[Lemma~4]%
{WW19states}, for states $\rho_{0}$, $\rho_{1}$, and $\sigma$, and the same
reasoning as in the proof of
Proposition~\ref{prop:pseudo-continuity-channel-sand}:%
\begin{equation}
D_{\beta}(\rho_{0}\Vert\sigma)-D_{\alpha}(\rho_{1}\Vert\sigma)\geq\frac
{2}{1-\alpha}\log_{2}\!\left[  1-\frac{1}{2}\left\Vert \rho_{0}-\rho
_{1}\right\Vert _{1}\right]  ,
\end{equation}
concluding the proof.
\end{proof}

\bigskip
We can then use the above bounds for channels to establish converse bounds for
general channel box transformation protocols.

\begin{proposition}
\label{prop:one-shot-box-trans-conv-bnd}Let $\mathcal{N}_{A\rightarrow B}$,
$\mathcal{M}_{A\rightarrow B}$, $\mathcal{K}_{C\rightarrow D}$, $\mathcal{L}%
_{C\rightarrow D}$ be quantum channels, and let $\Theta_{\left(  A\rightarrow
B\right)  \rightarrow\left(  C\rightarrow D\right)  }$ be a superchannel such
that $\Theta(\mathcal{M})=\mathcal{L}$. For $\alpha\in(1/2,1)$ and
$\beta:=\beta(\alpha)=\alpha/\left(  2\alpha-1\right)  >1$, the following inequality holds
\begin{align}
\widetilde{D}_{\beta}(\mathcal{N}\Vert\mathcal{M})  &  \geq\widetilde
{D}_{\alpha}(\mathcal{K}\Vert\mathcal{L})+\frac{\alpha}{1-\alpha}\log
_{2}F(\Theta(\mathcal{N}),\mathcal{K}),
\end{align}
and for
$\alpha^{\prime}\in(0,1)$ and $\beta^{\prime}:=\beta^{\prime}(\alpha^{\prime
}):=2-\alpha^{\prime}\in(1,2)$, the following inequality holds
\begin{align}
D_{\beta^{\prime}}(\mathcal{N}\Vert\mathcal{M})  &  \geq D_{\alpha^{\prime}%
}(\mathcal{K}\Vert\mathcal{L})\nonumber\\
&  \qquad+\frac{2}{1-\alpha^{\prime}}\log_{2}\!\left[  1-\frac{1}{2}\left\Vert
\Theta(\mathcal{N})-\mathcal{K}\right\Vert _{\diamond}\right]  .
\label{eq:one-shot-converse-petz-gen-ch-box}%
\end{align}

\end{proposition}

\begin{proof}
As a consequence of the data processing inequality for channel divergences
with respect to superchannels, we find that%
\begin{align}
&  \widetilde{D}_{\beta}(\mathcal{N}\Vert\mathcal{M})\nonumber\\
&  \geq\widetilde{D}_{\beta}(\Theta(\mathcal{N})\Vert\Theta(\mathcal{M}))\\
&  =\widetilde{D}_{\beta}(\Theta(\mathcal{N})\Vert\mathcal{L})\\
&  \geq\widetilde{D}_{\alpha}(\mathcal{K}\Vert\mathcal{L})+\frac{\alpha
}{1-\alpha}\log_{2}F(\Theta(\mathcal{N}),\mathcal{K}),
\end{align}
where the last inequality follows from
Proposition~\ref{prop:pseudo-continuity-channel-sand}.

The inequality in \eqref{eq:one-shot-converse-petz-gen-ch-box}\ follows
similarly from data processing but then using
Proposition~\ref{prop:pseudo-continuity-channel-petz}.
\end{proof}

We can now use these one-shot bounds to establish converse bounds on the rate
at which it is possible to convert the $n$-fold channel box $(\mathcal{N}%
^{\otimes n},\mathcal{M}^{\otimes n})$ to the $m$-fold channel box
$(\mathcal{K}^{\otimes m},\mathcal{L}^{\otimes m})$.

\begin{proposition}
\label{prop:n-to-m-ch-box-trans-conv-bnd}Let channels $\mathcal{N}%
_{A\rightarrow B}$, $\mathcal{M}_{A\rightarrow B}$, $\mathcal{K}_{C\rightarrow
D}$, $\mathcal{L}_{C\rightarrow D}$ be given and suppose that there exists an
$(n,m,\varepsilon)$ channel box transformation protocol (i.e., a superchannel
$\Theta^{(n)}$ such that $\Theta^{(n)}(\mathcal{N}^{\otimes n})\approx
_{\varepsilon}\mathcal{K}^{\otimes m}$ and $\Theta^{(n)}(\mathcal{M}^{\otimes
n})=\mathcal{L}^{\otimes m}$). Then for $\alpha\in(1/2,1)$ and $\beta
:=\beta(\alpha)=\alpha/\left(  2\alpha-1\right)  >1$, the following bound
holds%
\begin{multline}
\frac{\widetilde{D}_{\beta}^{(n)}(\mathcal{N}\Vert\mathcal{M})}{\widetilde
{D}_{\alpha}^{(m)}(\mathcal{K}\Vert\mathcal{L})}\geq\frac{m}{n}\\
+\frac{2\alpha}{n\left(  1-\alpha\right)  \widetilde{D}_{\alpha}%
^{(m)}(\mathcal{K}\Vert\mathcal{L})}\log_{2}(1-\varepsilon).
\end{multline}
For $\alpha^{\prime}\in(0,1)$ and $\beta^{\prime}:=\beta^{\prime}%
(\alpha^{\prime}):=2-\alpha^{\prime}\in(1,2)$, the following bound holds%
\begin{multline}
\frac{D_{\beta^{\prime}}^{(n)}(\mathcal{N}\Vert\mathcal{M})}{D_{\alpha
^{\prime}}^{(m)}(\mathcal{K}\Vert\mathcal{L})}\geq\frac{m}{n}\\
+\frac{2}{n\left(  1-\alpha'\right)  D_{\alpha^{\prime}}^{(m)}%
(\mathcal{K}\Vert\mathcal{L})}\log_{2}(1-\varepsilon).
\end{multline}
In the above,%
\begin{align}
\widetilde{D}_{\beta}^{(n)}(\mathcal{N}\Vert\mathcal{M})  &  :=\frac{1}%
{n}\widetilde{D}_{\beta}(\mathcal{N}^{\otimes n}\Vert\mathcal{M}^{\otimes
n}),\\
D_{\beta'}^{(n)}(\mathcal{N}\Vert\mathcal{M})  &  :=\frac{1}{n}D_{\beta'
}(\mathcal{N}^{\otimes n}\Vert\mathcal{M}^{\otimes n}),
\end{align}
with a similar definition for the other quantities.
\end{proposition}

\begin{proof}
Applying Proposition~\ref{prop:one-shot-box-trans-conv-bnd}, we conclude that%
\begin{align}
&  \widetilde{D}_{\beta}(\mathcal{N}^{\otimes n}\Vert\mathcal{M}^{\otimes
n})\nonumber\\
&  \geq\widetilde{D}_{\alpha}(\mathcal{K}^{\otimes m}\Vert\mathcal{L}^{\otimes
m})+\frac{\alpha}{1-\alpha}\log_{2}F(\Theta^{(n)}(\mathcal{N}^{\otimes
n}),\mathcal{K}^{\otimes m})\nonumber\\
&  =\widetilde{D}_{\alpha}(\mathcal{K}^{\otimes m}\Vert\mathcal{L}^{\otimes
m})+\frac{2\alpha}{1-\alpha}\log_{2}\sqrt{F}(\Theta^{(n)}(\mathcal{N}^{\otimes
n}),\mathcal{K}^{\otimes m})\nonumber\\
&  \geq\widetilde{D}_{\alpha}(\mathcal{K}^{\otimes m}\Vert\mathcal{L}^{\otimes
m})+\frac{2\alpha}{1-\alpha}\log_{2}(1-\varepsilon).
\end{align}
The second inequality follows from the fact that%
\begin{equation}
\sqrt{F}(\Theta^{(n)}(\mathcal{N}^{\otimes n}),\mathcal{K}^{\otimes m}%
)\geq1-\frac{1}{2}\left\Vert \Theta^{(n)}(\mathcal{N}^{\otimes n}%
)-\mathcal{K}^{\otimes m}\right\Vert _{\diamond}.
\end{equation}
The other inequality follows from similar reasoning but instead using data
processing and \eqref{eq:one-shot-converse-petz-gen-ch-box}.
\end{proof}

\begin{corollary}
Let $(\mathcal{N},\mathcal{M})$ and $(\mathcal{K},\mathcal{L})$ be channel
boxes such that%
\begin{align}
\widetilde{D}_{\beta}^{(n)}(\mathcal{N}\Vert\mathcal{M})  &  =\widetilde
{D}_{\beta}(\mathcal{N}\Vert\mathcal{M}%
),\label{eq:additivity-sandwiched-Renyi}\\
\widetilde{D}_{\alpha}^{(m)}(\mathcal{K}\Vert\mathcal{L})  &  =\widetilde
{D}_{\alpha}(\mathcal{K}\Vert\mathcal{L}),
\label{eq:additivity-sandwiched-Renyi-2}%
\end{align}
for $n,m\geq1$, $\alpha\in(1/2,1)$, and $\beta:=\beta(\alpha)=\alpha/\left(
2\alpha-1\right)  >1$. Then the following bound applies to an
$(n,m,\varepsilon)$ general channel box transformation protocol:%
\begin{multline}
\frac{\widetilde{D}_{\beta}(\mathcal{N}\Vert\mathcal{M})}{\widetilde
{D}_{\alpha}(\mathcal{K}\Vert\mathcal{L})}\geq\frac{m}{n}\\
+\frac{2\alpha}{n\left(  1-\alpha\right)  \widetilde{D}_{\alpha}%
(\mathcal{K}\Vert\mathcal{L})}\log_{2}(1-\varepsilon).
\end{multline}
Alternatively, suppose that $(\mathcal{N},\mathcal{M})$ and $(\mathcal{K}%
,\mathcal{L})$ satisfy%
\begin{align}
D_{\beta^{\prime}}^{(n)}(\mathcal{N}\Vert\mathcal{M})  &  =D_{\beta^{\prime}%
}(\mathcal{N}\Vert\mathcal{M}),\label{eq:additivity-Petz-Renyi}\\
D_{\alpha^{\prime}}^{(m)}(\mathcal{K}\Vert\mathcal{L})  &  =D_{\alpha^{\prime
}}(\mathcal{K}\Vert\mathcal{L}), \label{eq:additivity-Petz-Renyi-2}%
\end{align}
for $n,m\geq1$, $\alpha^{\prime}\in(0,1)$, and $\beta^{\prime}:=\beta^{\prime
}(\alpha^{\prime}):=2-\alpha^{\prime}\in(1,2)$. Then the following bound holds%
\begin{multline}
\frac{D_{\beta^{\prime}}(\mathcal{N}\Vert\mathcal{M})}{D_{\alpha^{\prime}%
}(\mathcal{K}\Vert\mathcal{L})}\geq\frac{m}{n}\\
+\frac{2\alpha}{n\left(  1-\alpha\right)  D_{\alpha^{\prime}}(\mathcal{K}%
\Vert\mathcal{L})}\log_{2}(1-\varepsilon).
\end{multline}

\end{corollary}

\begin{proof}
This is a direct consequence of
Proposition~\ref{prop:n-to-m-ch-box-trans-conv-bnd} and the additivity
relations assumed in
\eqref{eq:additivity-sandwiched-Renyi}--\eqref{eq:additivity-sandwiched-Renyi-2}
and \eqref{eq:additivity-Petz-Renyi}--\eqref{eq:additivity-Petz-Renyi-2}.
\end{proof}

\begin{remark}
The desired additivity relations in
\eqref{eq:additivity-sandwiched-Renyi}--\eqref{eq:additivity-sandwiched-Renyi-2}
and \eqref{eq:additivity-Petz-Renyi}--\eqref{eq:additivity-Petz-Renyi-2}\ hold
for channel boxes that are classical--quantum or environment seizable
\cite{BHKW18}. Thus, by applying reasoning similar to that given in \cite[Appendix~J]{WW19states}, we conclude the following strong converse bound for these channel boxes:
\begin{equation}
\widetilde{R}^p((\mathcal{N},\mathcal{M})\rightarrow(\mathcal{K},\mathcal{L}%
)) \leq 
  \frac{D(\mathcal{N}\Vert\mathcal{M})}{D(\mathcal{K}\Vert\mathcal{L})}.
\end{equation}
The lower bound (achievability)
\begin{equation}
R^p((\mathcal{N},\mathcal{M})\rightarrow(\mathcal{K},\mathcal{L}%
)) \geq 
  \frac{D(\mathcal{N}\Vert\mathcal{M})}{D(\mathcal{K}\Vert\mathcal{L})}.
\end{equation}
follows from combining a distillation  protocol with a dilution protocol for these channel boxes, as well as reasoning similar to that given in \cite[Appendix~J]{WW19states}, and along with the fact that the rates $D(\mathcal{N}\Vert\mathcal{M})$ and $D(\mathcal{K}\Vert\mathcal{L})$ are achievable for these tasks and these channel boxes.
Thus, 
the asymptotic parallel box transformation problem
has a simple solution for these channel boxes.
\end{remark}

\section{Bounding the smooth channel max-relative entropy in terms of channel
relative entropies}

\label{app:smooth-max-lower-bnds}

In this appendix, we provide lower bounds for the smooth channel max-relative
entropy in terms of the channel sandwiched and Petz--R\'enyi relative entropies.

\begin{proposition}
Let $\mathcal{N}_{A\rightarrow B}$ and $\mathcal{M}_{A\rightarrow B}$ be
quantum channels. Then the following bound holds for all $\alpha\in
\lbrack1/2,1)$ and $\varepsilon\in\lbrack0,1)$:%
\begin{equation}
D_{\max}^{\varepsilon}(\mathcal{N}\Vert\mathcal{M})\geq\widetilde{D}_{\alpha
}(\mathcal{N}\Vert\mathcal{M})+\frac{2\alpha}{\alpha-1}\log_{2}\!\left(
\frac{1}{1-\varepsilon}\right)  . \label{eq:smooth-dmax-lower-bnd-alpha}%
\end{equation}

\end{proposition}

\begin{proof}
First fix $\alpha\in(1/2,1)$. Then pick $\widetilde{\mathcal{N}}$ to be a
channel such that $\widetilde{\mathcal{N}}\approx_{\varepsilon}\mathcal{N}$.
We find for $\beta:=\alpha/(2\alpha-1)$ that%
\begin{align}
&  D_{\max}(\widetilde{\mathcal{N}}\Vert\mathcal{M})\nonumber\\
&  \geq\widetilde{D}_{\beta}(\widetilde{\mathcal{N}}\Vert\mathcal{M})\\
&  \geq\widetilde{D}_{\alpha}(\mathcal{N}\Vert\mathcal{M})+\frac{\alpha
}{1-\alpha}\log_{2}F(\widetilde{\mathcal{N}},\mathcal{N})\\
&  =\widetilde{D}_{\alpha}(\mathcal{N}\Vert\mathcal{M})+\frac{2\alpha
}{1-\alpha}\log_{2}\sqrt{F}(\widetilde{\mathcal{N}},\mathcal{N})\\
&  \geq\widetilde{D}_{\alpha}(\mathcal{N}\Vert\mathcal{M})+\frac{2\alpha
}{1-\alpha}\log_{2}(1-\varepsilon).
\end{align}
The first inequality follows from the fact that the sandwiched R\'enyi relative entropies are monotone \cite{muller2013quantum} and $\lim_{\alpha \to \infty} \widetilde{D}_\alpha = D_{\max}$ \cite{muller2013quantum}. The second inequality follows
from Proposition~\ref{prop:pseudo-continuity-channel-sand}. The final
inequality follows because%
\begin{equation}
1-\sqrt{F}(\widetilde{\mathcal{N}},\mathcal{N})\leq\frac{1}{2}\left\Vert
\widetilde{\mathcal{N}}-\mathcal{N}\right\Vert _{\diamond}.
\end{equation}
Since the bound holds for an arbitrary channel $\widetilde{\mathcal{N}}$
satisfying $\widetilde{\mathcal{N}}\approx_{\varepsilon}\mathcal{N}$, we
conclude \eqref{eq:smooth-dmax-lower-bnd-alpha}. The inequality for
$\alpha=1/2$ follows by taking a limit.
\end{proof}

Another lower bound on the smooth channel max-relative entropy is as follows:

\begin{proposition}
Let $\mathcal{N}_{A\rightarrow B}$ and $\mathcal{M}_{A\rightarrow B}$ be
quantum channels. Then the following bound holds for all $\alpha\in[0,1)$ and
$\varepsilon\in\lbrack0,1)$:%
\begin{equation}
D_{\max}^{\varepsilon}(\mathcal{N}\Vert\mathcal{M})\geq D_{\alpha}%
(\mathcal{N}\Vert\mathcal{M})+\frac{2}{\alpha-1}\log_{2}\!\left(  \frac
{1}{1-\varepsilon}\right)  . \label{eq:smooth-dmax-lower-bnd-alpha-petz}%
\end{equation}

\end{proposition}

\begin{proof}
First fix $\alpha\in(0,1)$. Then pick $\widetilde{\mathcal{N}}$ to be a
channel such that $\widetilde{\mathcal{N}}\approx_{\varepsilon}\mathcal{N}$.
We find for $\beta:=2-\alpha$ that%
\begin{align}
&  D_{\max}(\widetilde{\mathcal{N}}\Vert\mathcal{M})\nonumber\\
&  \geq D_{\beta}(\widetilde{\mathcal{N}}\Vert\mathcal{M})\\
&  \geq D_{\alpha}(\mathcal{N}\Vert\mathcal{M})+\frac{2}{1-\alpha}\log
_{2}\left[  1-\frac{1}{2}\left\Vert \widetilde{\mathcal{N}}-\mathcal{N}%
\right\Vert _{\diamond}\right] \\
&  \geq D_{\alpha}(\mathcal{N}\Vert\mathcal{M})+\frac{2\alpha}{1-\alpha}%
\log_{2}(1-\varepsilon).
\end{align}
The first inequality follows from the fact that $D_{\max} \geq D_2$ \cite{JRSWW16,WW19states} and the Petz--R\'enyi relative entropies are monotone with respect to $\beta$ \cite{TCR09}. The second inequality follows
from Proposition~\ref{prop:pseudo-continuity-channel-petz}. Since the bound
holds for an arbitrary channel $\widetilde{\mathcal{N}}$ satisfying
$\widetilde{\mathcal{N}}\approx_{\varepsilon}\mathcal{N}$, we conclude
\eqref{eq:smooth-dmax-lower-bnd-alpha-petz}. The inequality for $\alpha=0$
follows by taking a limit.
\end{proof}

\section{Quantum strategy and sequential channel box transformations}

\label{app:q-strat-trans}

\subsection{Bounds for general $n$-round strategy box transformations}

\label{app:bnds-strat-trans}

In this appendix, we provide bounds for general $n$-round strategy box
transformations. These bounds are similar in some regards to those given in
Appendix~\ref{app:gen-box-trans-bnds}, following essentially the same line of
reasoning to establish them.

\begin{proposition}
\label{prop:pseudo-continuity-channel-sand-seq}Let $\mathcal{N}^{(n)}$,
$\mathcal{L}^{(n)}$, and $\mathcal{M}^{(n)}$ be quantum strategies such that
$D_{\max}(\mathcal{N}^{(n)}\Vert\mathcal{M}^{(n)})<\infty$. Then for
$\alpha\in(1/2,1)$ and $\beta:=\beta(\alpha)=\alpha/\left(  2\alpha-1\right)
>1$, the following inequality holds%
\begin{multline}
\widetilde{D}_{\beta}(\mathcal{N}^{(n)}\Vert\mathcal{M}^{(n)})-\widetilde
{D}_{\alpha}(\mathcal{L}^{(n)}\Vert\mathcal{M}^{(n)}%
)\label{eq:pseudo-continuity-seq-channel-sandwiched}\\
\geq\frac{\alpha}{1-\alpha}\log_{2}F(\mathcal{N}^{(n)},\mathcal{L}^{(n)}),
\end{multline}
where $F(\mathcal{N}^{(n)},\mathcal{L}^{(n)})$ is the strategy fidelity of
\cite{Gutoski2018fidelityofquantum}.
\end{proposition}

\begin{proof}
Recall the following inequality from \cite[Lemma~1]{WW19states} for states
$\rho_{0}$, $\rho_{1}$, and $\sigma$:%
\begin{equation}
\widetilde{D}_{\beta}(\rho_{0}\Vert\sigma)-\widetilde{D}_{\alpha}(\rho
_{1}\Vert\sigma)\geq\frac{\alpha}{1-\alpha}\log_{2}F(\rho_{0},\rho_{1}).
\end{equation}
Applying the above inequality for states, and defining $\tau_{R_{n}B_{n}}$
from $\mathcal{L}^{(n)}$ in an analogous fashion to $\rho_{R_{n}B_{n}}$ and
$\sigma_{R_{n}B_{n}}$ in \eqref{eq:state-rho-seq-ch-1}\ and
\eqref{eq:state-sig-seq-ch-2}, respectively, we find that%
\begin{multline}
\widetilde{D}_{\beta}(\rho_{R_{n}B_{n}}\Vert\sigma_{R_{n}B_{n}})\\
\geq\widetilde{D}_{\alpha}(\tau_{R_{n}B_{n}}\Vert\sigma_{R_{n}B_{n}})\\
+\frac{\alpha}{1-\alpha}\log_{2}F(\rho_{R_{n}B_{n}},\tau_{R_{n}B_{n}}).
\end{multline}
Taking a supremum over all co-strategies $\{\rho_{R_{1}A_{1}},\{\mathcal{A}%
_{R_{i}B_{i}\rightarrow R_{i+1}A_{i+1}}^{i}\}_{i=1}^{n-1}\}$ on the left-hand
side, and an infimum on the right-hand side, we find that%
\begin{multline}
\widetilde{D}_{\beta}(\mathcal{N}^{(n)}\Vert\mathcal{M}^{(n)})\geq
\widetilde{D}_{\alpha}(\tau_{R_{n}B_{n}}\Vert\sigma_{R_{n}B_{n}})\\
+\frac{\alpha}{1-\alpha}\log_{2}F(\mathcal{N}^{(n)},\mathcal{L}^{(n)}).
\end{multline}
Since the above inequality holds for all co-strategies $\{\rho_{R_{1}A_{1}%
},\{\mathcal{A}_{R_{i}B_{i}\rightarrow R_{i+1}A_{i+1}}^{i}\}_{i=1}^{n-1}\}$,
we finally take another supremum to conclude \eqref{eq:pseudo-continuity-seq-channel-sandwiched}.
\end{proof}

\begin{proposition}
\label{prop:pseudo-continuity-channel-petz-seq}Let $\mathcal{N}^{(n)}$,
$\mathcal{L}^{(n)}$, and $\mathcal{M}^{(n)}$ be quantum strategies such that
$D_{\max}(\mathcal{N}^{(n)}\Vert\mathcal{M}^{(n)})<\infty$. Then for
$\alpha\in(0,1)$ and $\beta:=\beta(\alpha)=2-\alpha>1$, the following
inequality holds%
\begin{multline}
D_{\beta}(\mathcal{N}^{(n)}\Vert\mathcal{M}^{(n)})-D_{\alpha}(\mathcal{L}%
^{(n)}\Vert\mathcal{M}^{(n)})\\
\geq\frac{2}{1-\alpha}\log_{2}\!\left[  1-\frac{1}{2}\left\Vert \mathcal{N}%
^{(n)}-\mathcal{L}^{(n)}\right\Vert _{\diamond n}\right]  ,
\end{multline}
where $\left\Vert \mathcal{N}^{(n)}-\mathcal{L}^{(n)}\right\Vert _{\diamond
n}$ denotes the quantum strategy distance of \cite{CDP08a,CDP09,G12}.
\end{proposition}

\begin{proof}
This is a consequence of the following inequality from \cite[Lemma~4]%
{WW19states}, for states $\rho_{0}$, $\rho_{1}$, and $\sigma$, and the same
reasoning as in the proof of
Proposition~\ref{prop:pseudo-continuity-channel-sand-seq}:%
\begin{equation}
D_{\beta}(\rho_{0}\Vert\sigma)-D_{\alpha}(\rho_{1}\Vert\sigma)\geq\frac
{2}{1-\alpha}\log_{2}\!\left[  1-\frac{1}{2}\left\Vert \rho_{0}-\rho
_{1}\right\Vert _{1}\right]  ,
\end{equation}
concluding the proof.
\end{proof}

\bigskip
We can then use the above bounds for quantum strategies to establish converse
bounds for general strategy box transformation protocols.

\begin{proposition}
\label{prop:n-shot-seq-box-trans-conv-bnd}Let $\mathcal{N}^{(n)}$,
$\mathcal{M}^{(n)}$, $\mathcal{K}^{(m)}$, $\mathcal{L}^{(m)}$ be quantum
strategies, and let $\Theta^{(n\rightarrow m)}$ be a physical transformation
such that $\Theta^{(n\rightarrow m)}(\mathcal{M}^{(n)})=\mathcal{L}^{(m)}$.
For $\alpha\in(1/2,1)$ and $\beta:=\beta(\alpha)=\alpha/\left(  2\alpha
-1\right)  >1$, the following inequality holds%
\begin{multline}
\widetilde{D}_{\beta}(\mathcal{N}^{(n)}\Vert\mathcal{M}^{(n)})\geq
\widetilde{D}_{\alpha}(\mathcal{K}^{(m)}\Vert\mathcal{L}^{(m)})\\
+\frac{\alpha}{1-\alpha}\log_{2}F(\Theta^{(n\rightarrow m)}(\mathcal{N}%
^{(n)}),\mathcal{K}^{(m)}).
\end{multline}
For $\alpha^{\prime}\in(0,1)$ and $\beta^{\prime}:=\beta^{\prime}%
(\alpha^{\prime}):=2-\alpha^{\prime}\in(1,2)$, the following inequality holds%
\begin{multline}
D_{\beta^{\prime}}(\mathcal{N}^{(n)}\Vert\mathcal{M}^{(n)})\geq D_{\alpha
^{\prime}}(\mathcal{K}^{(m)}\Vert\mathcal{L}^{(m)})\\
+\frac{2}{1-\alpha^{\prime}}\log_{2}\!\left[  1-\frac{1}{2}\left\Vert
\Theta^{(n\rightarrow m)}(\mathcal{N}^{(n)})-\mathcal{K}^{(m)}\right\Vert
_{\diamond}\right]  . \label{eq:n-shot-seq-converse-petz-gen-ch-box}%
\end{multline}

\end{proposition}

\begin{proof}
As a consequence of the data processing inequality for the quantum strategy
divergence with respect to physical transformations
(Theorem~\ref{thm:DP-strat-div}), we find that%
\begin{align}
&  \widetilde{D}_{\beta}(\mathcal{N}^{(n)}\Vert\mathcal{M}^{(n)})\nonumber\\
&  \geq\widetilde{D}_{\beta}(\Theta^{(n\rightarrow m)}(\mathcal{N}^{(n)}%
)\Vert\Theta^{(n\rightarrow m)}(\mathcal{M}^{(n)}))\\
&  =\widetilde{D}_{\beta}(\Theta^{(n\rightarrow m)}(\mathcal{N}^{(n)}%
)\Vert\mathcal{L}^{(m)})\\
&  \geq\widetilde{D}_{\alpha}(\mathcal{K}^{(m)}\Vert\mathcal{L}^{(m)}%
)\nonumber\\
&  \qquad+\frac{\alpha}{1-\alpha}\log_{2}F(\Theta^{(n\rightarrow
m)}(\mathcal{N}^{(n)}),\mathcal{K}^{(m)}),
\end{align}
where the last inequality follows from
Proposition~\ref{prop:pseudo-continuity-channel-sand-seq}.

The inequality in \eqref{eq:n-shot-seq-converse-petz-gen-ch-box}\ follows
similarly from data processing but then using
Proposition~\ref{prop:pseudo-continuity-channel-petz-seq}.
\end{proof}

We can now use these bounds to establish converse bounds on the rate at which
it is possible to convert the quantum strategy box $(\mathcal{N}%
^{(n)},\mathcal{M}^{(n)})$ to the strategy box $(\mathcal{K}^{(m)}%
,\mathcal{L}^{(m)})$.

\begin{proposition}
\label{prop:n-to-m-seq-ch-box-trans-conv-bnd}Let quantum strategies
$\mathcal{N}^{(n)}$, $\mathcal{M}^{(n)}$, $\mathcal{K}^{(m)}$, $\mathcal{L}%
^{(m)}$ be given and suppose that there exists an $(n,m,\varepsilon)$ strategy
box transformation protocol (i.e., a physical transformation $\Theta
^{(n\rightarrow m)}$ such that $\Theta^{(n\rightarrow m)}(\mathcal{N}%
^{(n)})\approx_{\varepsilon}\mathcal{K}^{(m)}$ and $\Theta^{(n\rightarrow
m)}(\mathcal{M}^{(n)})=\mathcal{L}^{(m)}$). Then for $\alpha\in(1/2,1)$ and
$\beta:=\beta(\alpha)=\alpha/\left(  2\alpha-1\right)  >1$, the following
bound holds%
\begin{multline}
\frac{\widetilde{D}_{\beta}(\mathcal{N}^{(n)}\Vert\mathcal{M}^{(n)}%
)/n}{\widetilde{D}_{\alpha}(\mathcal{K}^{(m)}\Vert\mathcal{L}^{(m)})/m}%
\geq\frac{m}{n}\\
+\frac{2\alpha}{n\left(  1-\alpha\right)  \widetilde{D}_{\alpha}%
(\mathcal{K}^{(m)}\Vert\mathcal{L}^{(m)})/m}\log_{2}(1-\varepsilon).
\end{multline}
For $\alpha^{\prime}\in(0,1)$ and $\beta^{\prime}:=\beta^{\prime}%
(\alpha^{\prime}):=2-\alpha^{\prime}\in(1,2)$, the following bound holds%
\begin{multline}
\frac{D_{\beta^{\prime}}(\mathcal{N}^{(n)}\Vert\mathcal{M}^{(n)})/n}%
{D_{\alpha^{\prime}}(\mathcal{K}^{(m)}\Vert\mathcal{L}^{(m)})/m}\geq\frac
{m}{n}\\
+\frac{2}{n\left(  1-\alpha'\right)  D_{\alpha^{\prime}}(\mathcal{K}%
\Vert\mathcal{L}^{(m)})/m}\log_{2}(1-\varepsilon).
\end{multline}

\end{proposition}

\begin{proof}
Applying Proposition~\ref{prop:n-shot-seq-box-trans-conv-bnd}, we conclude
that%
\begin{align}
&  \widetilde{D}_{\beta}(\mathcal{N}^{(n)}\Vert\mathcal{M}^{(n)})\nonumber\\
&  \geq\widetilde{D}_{\alpha}(\mathcal{K}^{(m)}\Vert\mathcal{L}^{(m)}%
)\nonumber\\
&  \qquad+\frac{\alpha}{1-\alpha}\log_{2}F(\Theta^{(n\rightarrow
m)}(\mathcal{N}^{(n)}),\mathcal{K}^{(m)})\\
&  =\widetilde{D}_{\alpha}(\mathcal{K}^{(m)}\Vert\mathcal{L}^{(m)})\nonumber\\
&  \qquad+\frac{2\alpha}{1-\alpha}\log_{2}\sqrt{F}(\Theta^{(n\rightarrow
m)}(\mathcal{N}^{(n)}),\mathcal{K}^{(m)})\\
&  \geq\widetilde{D}_{\alpha}(\mathcal{K}^{(m)}\Vert\mathcal{L}^{(m)}%
)+\frac{2\alpha}{1-\alpha}\log_{2}(1-\varepsilon).
\end{align}
The second inequality follows from the fact that
\cite{Gutoski2018fidelityofquantum}%
\begin{multline}
\sqrt{F}(\Theta^{(n)}(\mathcal{N}^{(n)}),\mathcal{K}^{(m)})\\
\geq1-\frac{1}{2}\left\Vert \Theta^{(n)}(\mathcal{N}^{(n)})-\mathcal{K}%
^{(m)}\right\Vert _{\diamond}.
\end{multline}
The other inequality follows from similar reasoning but instead using data
processing and \eqref{eq:n-shot-seq-converse-petz-gen-ch-box}.
\end{proof}

\begin{corollary}
\label{cor:seq-ch-box-conv}
Let $(\mathcal{N}^{(n)},\mathcal{M}^{(n)})$ and $(\mathcal{K}^{(m)}%
,\mathcal{L}^{(m)})$ be sequential channel boxes such that%
\begin{align}
\widetilde{D}_{\beta}(\mathcal{N}^{(n)}\Vert\mathcal{M}^{(n)})/n  &
=\widetilde{D}_{\beta}(\mathcal{N}\Vert\mathcal{M}%
),\label{eq:seq-additivity-sandwiched-Renyi}\\
\widetilde{D}_{\alpha}(\mathcal{K}^{(m)}\Vert\mathcal{L}^{(m)})/m  &
=\widetilde{D}_{\alpha}(\mathcal{K}\Vert\mathcal{L}),
\label{eq:seq-additivity-sandwiched-Renyi-2}%
\end{align}
for $n,m\geq1$, $\alpha\in(1/2,1)$, and $\beta:=\beta(\alpha)=\alpha/\left(
2\alpha-1\right)  >1$. Then the following bound applies to an
$(n,m,\varepsilon)$ general channel box transformation protocol:%
\begin{multline}
\frac{\widetilde{D}_{\beta}(\mathcal{N}\Vert\mathcal{M})}{\widetilde
{D}_{\alpha}(\mathcal{K}\Vert\mathcal{L})}\geq\frac{m}{n}%
\label{eq:bnd-seq-simplified-1}\\
+\frac{2\alpha}{n\left(  1-\alpha\right)  \widetilde{D}_{\alpha}%
(\mathcal{K}\Vert\mathcal{L})}\log_{2}(1-\varepsilon).
\end{multline}
Alternatively, suppose that $(\mathcal{N}^{(n)},\mathcal{M}^{(n)})$ and
$(\mathcal{K}^{(m)},\mathcal{L}^{(m)})$ satisfy%
\begin{align}
D_{\beta^{\prime}}(\mathcal{N}^{(n)}\Vert\mathcal{M}^{(n)})/n  &
=D_{\beta^{\prime}}(\mathcal{N}\Vert\mathcal{M}%
),\label{eq:seq-additivity-Petz-Renyi}\\
D_{\alpha^{\prime}}(\mathcal{K}^{(m)}\Vert\mathcal{L}^{(m)})/m  &
=D_{\alpha^{\prime}}(\mathcal{K}\Vert\mathcal{L}),
\label{eq:seq-additivity-Petz-Renyi-2}%
\end{align}
for $n,m\geq1$, $\alpha^{\prime}\in(0,1)$, and $\beta^{\prime}:=\beta^{\prime
}(\alpha^{\prime}):=2-\alpha^{\prime}\in(1,2)$. Then the following bound holds%
\begin{multline}
\frac{D_{\beta^{\prime}}(\mathcal{N}\Vert\mathcal{M})}{D_{\alpha^{\prime}%
}(\mathcal{K}\Vert\mathcal{L})}\geq\frac{m}{n} \label{eq:bnd-seq-simplified-2}%
\\
+\frac{2}{n\left(  1-\alpha'\right)  D_{\alpha^{\prime}}(\mathcal{K}%
\Vert\mathcal{L})}\log_{2}(1-\varepsilon).
\end{multline}

\end{corollary}

\begin{proof}
This is a direct consequence of
Proposition~\ref{prop:n-to-m-seq-ch-box-trans-conv-bnd} and the relations
assumed in
\eqref{eq:seq-additivity-sandwiched-Renyi}--\eqref{eq:seq-additivity-sandwiched-Renyi-2}
and \eqref{eq:seq-additivity-Petz-Renyi}--\eqref{eq:seq-additivity-Petz-Renyi-2}.
\end{proof}

\subsection{Sequential channel box transformations and amortized channel
divergence}

\label{app:seq-ch-trans-amortized-div}

In \cite{BHKW18}, the notion of amortized channel divergence of a channel box
$(\mathcal{N},\mathcal{M})$ was introduced as%
\begin{multline}
\mathbf{D}^{\mathcal{A}}(\mathcal{N}\Vert\mathcal{M}):=\\
\sup_{\rho_{RA},\sigma_{RA}}\mathbf{D}(\mathcal{N}_{A\rightarrow B}(\rho
_{RA})\Vert\mathcal{M}_{A\rightarrow B}(\sigma_{RA}))-\mathbf{D}(\rho
_{RA}\Vert\sigma_{RA}),
\end{multline}
where the optimization is with respect to input states $\rho_{RA}$ and
$\sigma_{RA}$, and the system $R$ has unbounded dimension. The intuition
behind this quantity is that it represents the largest net distinguishability
that can be generated by the channels $\mathcal{N}$ and $\mathcal{M}$ if we
are allowed to start with some distinguishability to begin with, in the form
of the state box $(\rho_{RA},\sigma_{RA})$.

Suppose now that we have a sequential channel box $(\mathcal{N}^{(n)}%
,\mathcal{M}^{(n)})$, where $\mathcal{N}^{(n)}$ consists of a sequence of $n$
uses of $\mathcal{N}$ and $\mathcal{M}^{(n)}$ consists of a sequence of $n$
uses of $\mathcal{M}$. As stated earlier, this sequential channel box is a
special kind of strategy box. Then by employing the same reasoning as in the
proof of \cite[Lemma~14]{BHKW18}, we conclude that the amortized channel
divergence is an upper bound on the normalized strategy divergence of
$(\mathcal{N}^{(n)},\mathcal{M}^{(n)})$:%
\begin{equation}
\frac{1}{n}\mathbf{D}(\mathcal{N}^{(n)}\Vert\mathcal{M}^{(n)})\leq
\mathbf{D}^{\mathcal{A}}(\mathcal{N}\Vert\mathcal{M}).
\label{eq:up-bnd-strat-div-amortized-div}%
\end{equation}

For some channel boxes and choices of divergences, the inequality in
\eqref{eq:up-bnd-strat-div-amortized-div}\ is saturated as a consequence of
the amortized channel divergence collapsing to the usual channel divergence
\cite{BHKW18}. This occurs for all classical--quantum or environment-seizable
channel boxes paired up with the Petz--R\'enyi relative entropy, the
sandwiched R\'enyi relative entropy, or the quantum relative entropy
\cite{BHKW18}. Thus, for these channels, the desired relations in
\eqref{eq:seq-additivity-sandwiched-Renyi}--\eqref{eq:seq-additivity-sandwiched-Renyi-2}
and
\eqref{eq:seq-additivity-Petz-Renyi}--\eqref{eq:seq-additivity-Petz-Renyi-2}\ hold,
so that the bounds in \eqref{eq:bnd-seq-simplified-1} and
\eqref{eq:bnd-seq-simplified-2} hold for these channels.

\end{document}